\documentclass[10pt]{article}

\pdfoutput=1
\usepackage[pdftex]{color}
\usepackage{amssymb}
\usepackage{amsthm} 
\usepackage{amsmath}
\usepackage{latexsym}
\usepackage{amscd}
\usepackage{graphicx}
 \usepackage{longtable}
 \usepackage{array}
\usepackage{amssymb}
\usepackage[pdftex, colorlinks=true, citecolor=green]{hyperref}
\usepackage{lscape}
\usepackage{setspace}
\usepackage{multirow}
\usepackage{wrapfig}

\newcommand{\ben}{\begin{equation}}     
\newcommand{\eeqn}{\end{equation}}
\newcommand{\bey}{\begin{eqnarray}}
\newcommand{\eey}{\end{eqnarray}}

\newtheorem{thm}{Theorem}[section]

\newtheorem{rem}[thm]{Remark}
\usepackage{tikz}
\usetikzlibrary{arrows.meta, positioning, calc, fit, backgrounds}
\usepackage{booktabs}
\usepackage{xcolor}

\definecolor{jfSlate}{HTML}{2F3E46}
\definecolor{jfInk}{HTML}{1B262C}
\definecolor{jfExc}{HTML}{3D7E7A}
\definecolor{jfInh}{HTML}{9B4A3A}
\definecolor{jfDelay}{HTML}{6C7A89}
\definecolor{jfFillA}{HTML}{F2EFEA}
\definecolor{jfFillB}{HTML}{E8EEF1}
\definecolor{jfFillC}{HTML}{EDE7DF}

\tikzset{
  jfnode/.style={draw=jfSlate, line width=0.5pt, rounded corners=2pt,
    inner sep=5pt, align=center, font=\small, text=jfInk,
    minimum height=11mm, minimum width=26mm},
  jfcortex/.style={jfnode, fill=jfFillA},
  jfthal/.style  ={jfnode, fill=jfFillB},
  jfdelay/.style ={jfnode, fill=jfFillC, font=\footnotesize, align=center,
                   minimum width=84mm, minimum height=15mm},
  jfgroup/.style ={draw=jfSlate!40, line width=0.4pt, rounded corners=4pt,
                   dashed, inner sep=4mm},
  jflabel/.style ={font=\scriptsize, text=jfInk, inner sep=1pt,
                   fill=white, fill opacity=0.85, text opacity=1},
  jfexc/.style   ={->, >={Stealth[length=2.2mm,width=1.6mm]},
                   draw=jfExc, line width=0.6pt},
  jfinh/.style   ={->, >={Stealth[length=2.2mm,width=1.6mm,open]},
                   draw=jfInh, line width=0.6pt},
  jfdelarrow/.style={->, >={Stealth[length=2mm,width=1.4mm]},
                   draw=jfDelay, line width=0.5pt, dashed},
}
\begin{document}

\begin{flushleft}
{\Large
\textbf{A distributed-delay Wilson--Cowan model of sleep-related rhythms in the corticothalamic system}
}

Eva Kaslik$^1$, Anca R\v{a}dulescu$^2$, Anca Stanoev$^1$

\emph{Department of Computer Science, West University of Timi\c{s}oara}

\emph{Department of Mathematics, SUNY New Paltz}

\end{flushleft}
\begin{abstract}
\noindent The corticothalamic circuit supports rhythms with timescales that
differ by orders of magnitude: sleep spindles, the sigma-band events of
non-rapid-eye-movement (NREM) sleep, and infra-slow fluctuations near
0.02~Hz that organize when spindles occur. Because the anatomy is the same
in both cases, architecture alone cannot determine which rhythm the circuit
expresses. We ask whether the temporal structure of the circuit's own
feedback can. In a four-population Wilson--Cowan model comprising cortical
excitatory and inhibitory populations, thalamic relay cells, and the
thalamic reticular nucleus (TRN), we first establish how connectivity
controls access to oscillatory behavior, and then introduce temporal
coupling as either a weak Gamma distributed delay or a discrete delay.
 
We investigate three distinct connectivity levels: recurrent cortical excitation gates whether the circuit can oscillate at all, the reciprocal relay--TRN pair determines where the oscillation lies and how it is configured, sustained, and terminated, and reticular self-inhibition limits its extent. We then examine how these connectivity-dependent regimes are affected by delayed coupling. Although delay does not change the equilibria themselves, it can substantially alter their stability and the organization of the resulting oscillatory dynamics. Under weak Gamma integration, short delays support spindle-compatible oscillations in the sigma band, while longer delays give rise to a much slower regime near 0.02~Hz. The discrete-delay formulation produces a qualitatively different and more complex bifurcation structure. Together, these results show that the dynamics of the corticothalamic circuit depend not only on its connectivity, but also on the temporal organization of interactions within the circuit.

\end{abstract}

\section{Introduction}
 
Sleep is not a uniform state, but is organized into recurring stages with distinct patterns of brain activity. Much of sleep is spent in non-rapid-eye-movement (NREM) sleep, which ranges from lighter to deeper stages. One of its most characteristic features is the sleep spindle: a brief, waxing-and-waning burst of neural activity, typically lasting on the order of a second and oscillating in the sigma range of roughly 11--16~Hz. Spindles are especially prominent during the intermediate NREM stage known as N2~\cite{fernandez2020sleep,nir2011regional,mak2017coordination}. They arise within the corticothalamic circuit formed by thalamic relay neurons, the thalamic reticular nucleus (TRN), and cortical excitatory and inhibitory populations. Within this circuit, the reciprocal relay-TRN interaction provides a rhythm-generating core, while corticothalamic feedback helps shape the timing and spatial organization of individual spindle events~\cite{krosigk1993cellular,halassa2011selective,pinault2004thalamic,bastuji2020local}.

The same circuit also participates in activity on a dramatically slower timescale. During NREM sleep, an infra-slow fluctuation near 0.02~Hz modulates \emph{when} spindles occur, grouping them into clusters and contributing to alternating periods of greater sleep continuity and greater fragility~\cite{lecci2017coordinated,lazar2019infraslow,watson2018cognitive,champetier2023age}. This fluctuation is not simply a spindle slowed down. Rather, it reflects a different level of temporal organization, associated with the longer-timescale structure and maintenance of sleep. Thus, the same corticothalamic circuit is associated with rhythms whose characteristic timescales differ by several orders of magnitude.

A first place to look for the origin of these different behaviors is the architecture of the circuit itself. The pattern and strength of connections between neural populations strongly constrain the collective dynamics that a network can support~\cite{destexhe2009wilson,harris2015neocortical}. In the corticothalamic system, these connections also play distinct roles: relay and reticular populations are central to rhythm generation, while cortical excitatory and inhibitory populations influence how that rhythm is expressed, amplified, and organized in time~\cite{pinault2004thalamic,niethard2018cortical,zhang2004cortical}. At the same time, connectivity alone cannot explain the coexistence of spindle and infra-slow dynamics, since the underlying anatomical circuit is the same in both cases. The question is therefore what other properties of the circuit can shift its dynamics between these very different timescales without changing its basic wiring. This question also has clinical relevance: spindle expression is altered in schizophrenia~\cite{manoach2016reduced,manoach2019abnormal}, reorganized by interictal activity in focal epilepsy~\cite{kramer2021focal}, and reduced in Alzheimer's disease and amnestic mild cognitive impairment in association with cognitive decline~\cite{gorgoni2016parietal,weng2020sleep}, while reduced spindle density in Parkinson's disease has been linked to later dementia~\cite{latreille2015sleep}.

A natural candidate is the timescale over which influences on the corticothalamic system are integrated. Neural activity is shaped by processes operating over a wide range of timescales, from axonal conduction, synaptic transmission, and dendritic integration to recurrent network dynamics and slower neuromodulatory or hormonal effects. At the population level, the combined influence of such processes can introduce a dependence on past activity that is spread over time rather than concentrated at a single instant. How that past activity is weighted may itself matter dynamically. Most delayed neural-population models represent this temporal dependence by a single fixed lag, largely for mathematical and computational convenience. Yet there is little reason to expect the effective temporal profile of biological interactions to be universal across networks, physiological states, or functional contexts. Our previous analyses of Wilson--Cowan systems with distributed delays have shown that the shape of the delay kernel, and not only its characteristic timescale, can qualitatively alter the accessible dynamics and the transitions between them~\cite{KaslikEtAl2022,KaslikEtAl2024}. This motivates examining temporal integration alongside connectivity as a factor shaping the dynamics of the corticothalamic system.

Existing models of thalamocortical sleep rhythms fall broadly into two
families, and neither has been used to ask this question. Biophysically
detailed models resolve the intrinsic currents and synaptic kinetics that
generate spindles and slow-wave activity, including the low-threshold calcium
current and rebound bursting of TRN neurons and the gap-junctional coupling
between them~\cite{destexhe1994model,golomb1994synchronization,bazhenov2002model,krishnan2016cellular}. These models reproduce the rhythms faithfully, but their dimensionality makes broad, systematic exploration of the dynamics across parameter space difficult. Mean-field
and neural-mass corticothalamic models occupy the complementary position:
they are low-dimensional enough for stability and continuation analysis and
have reproduced NREM-sleep spectra and their responses to
perturbation~\cite{SchellenbergerCosta2016,Jajcay2022,robinson2002dynamics}, and delayed
corticothalamic feedback has been shown to organize their stability
boundaries~\cite{roberts2008modeling}. Related work has also shown that
thalamocortical mean-field models can sit close to multiple coexisting
regimes, with transitions into pathological synchrony organized by the same
parameters as normal activity~\cite{suffczynski2004dynamics,Sheeba2008}. These
models, however, either omit delay or represent it as a single characteristic
lag. The consequence is that connectivity and timing have largely been
studied separately, and whether temporal dispersion, acting jointly with
connectivity, can select among the timescales that one circuit is able to
support has not been addressed directly.
 
Here we examine that question in a four-population Wilson--Cowan
model~\cite{WilsonCowan1972} of the corticothalamic circuit, comprising cortical
pyramidal and inhibitory populations, thalamic relay cells, and the TRN, with
connectivity calibrated to the anatomical and modeling literature. We compare
three temporal formulations: instantaneous coupling, a weak Gamma distributed delay kernel, and a
discrete delay. Equilibria, limit cycles, and their bifurcations are obtained
by numerical continuation, and are complemented by analytical results
governing stability at small and large values of the delay. The model was tuned to produce oscillatory regimes consistent with both the experimentally observed spindle frequency range and the characteristic relative recruitment patterns of the four populations. More details are presented in Section~\ref{sec:rhythms}.
 
Our central goal is to investigate how connectivity and temporal integration play complementary roles in shaping the dynamics of the corticothalamic circuit. We focus on several connectivity features that have been implicated in spindle dynamics: recurrent cortical excitation and corticothalamic drive, the reciprocal interaction between thalamic relay and reticular populations, and interactions within the TRN. We ask whether these different levels of connectivity control distinct aspects of oscillatory behavior, including access to oscillations, their structure and persistence, and the conditions under which they emerge or disappear. We then investigate how these connectivity-dependent effects are modified by delayed coupling, and whether the form of temporal integration -- weak Gamma distributed or discrete -- changes the dynamical regimes available to the circuit. Finally, we ask to what extent connectivity and delay, separately and in combination, can account for the experimentally observed spindle-band regime and for the much slower infra-slow organization near 0.02~Hz within the same corticothalamic architecture.

\section{Modeling the CTRC loop and generation of sleep rhythms}
\label{sec:model}

The coexistence of spindle-band and infra-slow rhythms described above makes NREM sleep a particularly rich setting for investigating dynamics across widely separated timescales. The thalamic reticular nucleus (TRN) stands within a broader corticothalamic network, interacting reciprocally with thalamic relay populations and with cortical excitatory and inhibitory populations, and it is within this coupled setting that sleep-related dynamics emerge. Our modeling framework is therefore designed to examine how distinct regimes are generated within this circuit, and what factors allow transitions from one to another. In the following sections, we first discuss the functional oscillatory regimes of interest and the criteria by which the model is assessed against them, then summarize the connectivity architecture that forms the biological backdrop for our model, and finally state the equations and parameters themselves.

\subsection{CTRC rhythms across population timescales}
\label{sec:rhythms}

\noindent \textbf{Spindle rhythms.} To calibrate the four-node Wilson--Cowan network, we distinguish between the temporal frequency of the spindle rhythm and the level of activity expressed by each population during that rhythm. Human depth-electrode, intracranial EEG and LFP recordings consistently identify spindle oscillations in the approximately 10--16~Hz range~\cite{mak2017coordination,bastuji2020local,nir2011regional}. We therefore
require a spindle-compatible periodic solution to oscillate on this temporal scale. The four Wilson--Cowan populations share this common oscillatory period; differences among populations are  represented instead through their activity levels, which reflect the substantially different degrees to which the corresponding neuronal populations are recruited during individual spindle cycles.

Where direct human population recordings are available, we use them to establish the spindle-frequency range. Because population-resolved human spiking data are more limited, we supplement these observations with rodent and feline single-unit and population studies describing cycle-by-cycle
participation. The resulting Wilson--Cowan activity ranges should therefore be interpreted as Hz-equivalent mean-field targets: they preserve the experimentally observed distinction between strongly recruited populations and populations that participate only sparsely in the same spindle-frequency network oscillation, rather than representing literal firing rates of every individual neuron.

\vspace{2mm}
\noindent \emph{\textbf{Thalamic reticular population.}} The TRN is a central component of the spindle-generating thalamic circuit.
Optogenetic activation of TRN PV neurons can initiate thalamic bursts and cortical spindles \cite{halassa2011selective}, while classical intracellular recordings demonstrate intrinsic and network-supported TRN rhythmicity in the spindle-frequency range \cite{steriade1987deafferented}. Natural-sleep recordings further show that reticular neurons are strongly, but not uniformly, recruited across spindle cycles. Barth\'o et al.~\cite{bartho2014ongoing}
reported nRT participation probabilities ranging from below approximately
40\% to about 60\%, with firing declining markedly toward spindle termination
and bursts containing approximately 3--5 spikes. The lower-recruitment portions
of this trajectory therefore support an activity scale of only a few Hz,
motivating $WC_{\min}\approx3$--$5$~Hz, whereas the strong recruitment of TRN
neurons, which can participate on successive spindle cycles, supports an upper
population-activity scale approaching the spindle-cycle rate. We consequently
take $WC_{\min}\approx3$--$5$~Hz and $WC_{\max}\approx10$--$15$~Hz for the
reticular population.

\vspace{2mm}
\noindent \emph{\textbf{Thalamic relay population.}} Human thalamic recordings demonstrate a clear spindle-band population rhythm in the 10-16~Hz range \cite{mak2017coordination,bastuji2020local}.
Individual thalamocortical relay cells, however, participate much more sparsely in that population rhythm. Animal recordings show extensive cycle-skipping~\cite{pinault2003cellular}; in natural sleep, Barth\'o et al.~(2014) measured TC-cell participation probabilities of approximately 35-45\% across spindle cycles. Applied to a 10-15~Hz spindle rhythm, this degree of participation corresponds to approximately 3.5-7 recruited
cycle-events per second, providing a direct scale for the upper end of relay population activity and motivating $WC_{\max}\approx4$-$6$~Hz. The intermittent, cycle-skipping character of relay-cell recruitment also supports
a population-activity trough close to zero, motivating $WC_{\min}\approx0$-$1$~Hz. We therefore take $WC_{\min}\approx0$-$1$~Hz and $WC_{\max}\approx4$--$6$~Hz for the relay
population.

\vspace{2mm}
\noindent \emph{\textbf{Cortical pyramidal population.}} Cortical EEG and intracranial recordings show the same spindle-band temporal rhythm observed in thalamus~\cite{nir2011regional,mak2017coordination}, but individual pyramidal neurons participate only sparsely. Recordings during
natural sleep show low average firing and substantial cycle skipping among
pyramidal neurons~\cite{averkin2016identified}; Averkin et al., for example, found that superficial pyramidal cells fired on only a small fraction of spindle cycles, with spindle-phase firing rates near 1--2~Hz in the more strongly modulated cells. Other cell-resolved recordings report mean pyramidal firing around 3-4~Hz during local spindle oscillations~\cite{hartwich2009distinct}. Together, these observations motivate a
low-activity range of $WC_{\min}\approx1$--$2$~Hz and a spindle-recruited
upper range of $WC_{\max}\approx3$--$5$~Hz. Thus cortical pyramidal activity remains sparse even while the population signal oscillates robustly in the spindle band.

\vspace{2mm}
\noindent \emph{\textbf{Cortical inhibitory population.}} Cortical inhibitory neurons are recruited considerably more strongly during spindles than neighboring pyramidal neurons. Human and animal recordings show
strong spindle phase-locking of inhibitory activity~\cite{nir2011regional,niethard2018cortical,averkin2016identified}, with
fast-spiking and PV-dominated populations showing particularly strong
spindle-related recruitment. At the same time, cortical inhibitory populations
are heterogeneous: recordings during local spindle oscillations show firing rates ranging from only a few Hz in more weakly recruited interneuron classes to approximately 10-17~Hz in strongly recruited PV basket cells~\cite{hartwich2009distinct}. We therefore use $WC_{\min}\approx1$-$2$~Hz to represent the low-activity end of the
coarse-grained inhibitory population and $WC_{\max}\approx10$-$15$~Hz to represent the strongly recruited, spindle-locked end. The broad range reflects the markedly greater and more heterogeneous recruitment of cortical inhibitory cells relative to pyramidal neurons.

\vspace{2mm}
\noindent These ranges were used to assess the physiological plausibility of the dynamical regimes produced by the chosen parameter sets. Thus, model validation did not consist of imposing an independent recruitment constraint after matching the spindle frequency. Rather, spindle-compatible behavior was identified by considering both the common spindle-band oscillation and the population-activity scales expected from the experimentally observed differential recruitment of the four populations.\\

\noindent \textbf{Infra-slow fluctuations.} In addition to spindle-band activity, non-rapid-eye-movement sleep exhibits a much slower temporal organization on the infra-slow timescale. In particular, several studies have identified an approximately 0.02 Hz rhythm during NREM sleep, corresponding to a period on the order of 50 seconds, which modulates spindle occurrence and helps organize sleep microstructure rather than constituting a spindle rhythm itself \cite{lecci2017coordinated,watson2018cognitive,lazar2019infraslow}. This infra-slow fluctuation has been linked to the clustering and spacing of faster spindle events, to alternating periods of greater sleep continuity versus fragility, and more broadly to the temporal organization of NREM sleep \cite{lecci2017coordinated,watson2018cognitive,lazar2019infraslow,champetier2023age}. These observations suggest that infra-slow fluctuations play an important functional role in sleep not by replacing faster rhythms such as spindles, but by regulating when such rhythms occur and how they are distributed over longer timescales. In this sense, infra-slow activity appears to provide a broader temporal scaffold for sleep microarchitecture, contributing to sleep maintenance and to the internal organization of NREM sleep. In the present work, we therefore treat infra-slow activity not at the level of detailed population-by-population firing statistics, but as a biologically observed slower timescale relevant to the slower oscillatory regime of the model. This regime is consequently assessed on the frequency constraint alone: population-resolved firing statistics comparable to those available for individual spindle cycles have not been reported for the infra-slow rhythm, so imposing a recruitment constraint there would not be evidence-based. Accordingly, when the Wilson--Cowan dynamics exhibit a slow oscillatory mode in the infra-slow range, we interpret this not as a literal single-cell firing rate, but as a network-level organizing regime associated with infra-slow sleep organization.

\subsection{Description of CTRC neural circuitry}
\label{sec:circuitry}

The connectivity architecture between the four representative neural populations in our system is summarized below, and illustrated in Figure~\ref{fig:architecture}. The corresponding coupling weights and ranges used in our analysis of the model are included in Table~\ref{tab:coupling_refs}, and define the connectivity matrix introduced in Section~\ref{sec:equations}. While not directly empirically driven, the coupling baseline values and ratios were based on existing literature on the CTRC circuit, also included in Table~\ref{tab:coupling_refs}.

\vspace{2mm}
\noindent \emph{\textbf{Corticothalamic projections and collaterals to TRN.}} Layer-VI pyramidal neurons send descending axons not only to thalamic relay nuclei but also branch collaterals into the TRN, providing the principal excitatory (glutamate-mediated) drive that paces TRN burst-firing and helps initiate each spindle cycle~\cite{zhang2004cortical}. Additional studies confirm that TRN cells receive glutamatergic corticothalamic collateral synapses and are closely coupled to their neighbors by GABA\textsubscript{A} receptor-containing inhibitory synapses~\cite{zhang2004corticothalamic}. No converse projection from the TRN to the cortex has been found to be relevant to spindle dynamics~\cite{pinault2004thalamic}. These pathways are represented by the cortico-reticular collateral $w_{RP}$, and by the absence of any $R \to P$ or $R \to I$ entry in the connectivity matrix.

\vspace{2mm}
\noindent \emph{\textbf{Long-range projections between cortex and thalamic relay.}} Relay neurons in sensory (ventrobasal, lateral geniculate) and associative thalamic nuclei issue axon branches en route to cortex that activate TRN cells, closing the relay--TRN--relay inhibitory loop that sustains spindle oscillations. Both these projections are glutamatergic. Thalamocortical axon terminals express vesicular glutamate transporter 2 (VGLUT2), driving excitatory postsynaptic responses in both cortical pyramidal cells and interneurons~\cite{zhang2004corticothalamic}. Conversely, there is a massive, glutamatergic, feedback projection to the thalamus, primarily from the layer 6 corticothalamic projection neurons, both to the specific thalamic relay nucleus that provides its principal input and to higher order thalamic nuclei~\cite{thomson2010neocortical}. Layer 6 pyramidal neurons, in turn, receive extrinsic excitatory input, from the thalamus and from other cortical regions, as further detailed below. These pathways are represented by the ascending weights $w_{PT}$ and $w_{IT}$ and by the descending corticothalamic drive $w_{TP}$.

\vspace{2mm}
\noindent \emph{\textbf{Long-range projections between thalamic relay and TRN.}} Every TRN neuron is GABAergic and projects back onto the relay cells that innervate cortex~\cite{pinault2004thalamic}. Their rhythmic inhibitory postsynaptic potentials drive rebound bursts in these relay neurons, which then re-excite the TRN, generating the spindle-band oscillation. Early intracellular recordings in thalamic slices established that TRN inhibition is both necessary and sufficient for spindle-like oscillations~\cite{krosigk1993cellular}. In turn, the glutamatergic input to TRN depolarizes TRN neurons and helps initiate the inhibitory-rebound cycles underlying spindles, although activation of metabotropic glutamate receptors can, in some settings, produce a slower inhibitory response through potassium conductances~\cite{krosigk1993cellular}. Deafferentation and in vivo recording studies demonstrated that interrupting these projections abolishes spindle rhythms~\cite{steriade1987deafferented}. The spread of inhibition through the TRN as more and more TRN cells are excited by the collaterals of thalamocortical fibers may be responsible for the shortening of the burst discharges~\cite{zhang2004corticothalamic}. Conversely, studies in rodent and other species found the sheet of GABAergic neurons in the TRN to form the principal source of inhibition to the relay neurons of most dorsal thalamic nuclei~\cite{pinault2004thalamic}. There is no evidence of similar interconnections within the thalamic relay, or of significant effects of one relay cell on another, each relay cell acting as an essentially independent link to cortex~\cite{sherman2016thalamus}. This reciprocal pair is represented by the excitatory return $w_{RT}$ and the inhibitory return $w_{TR}$, with no relay self-coupling.

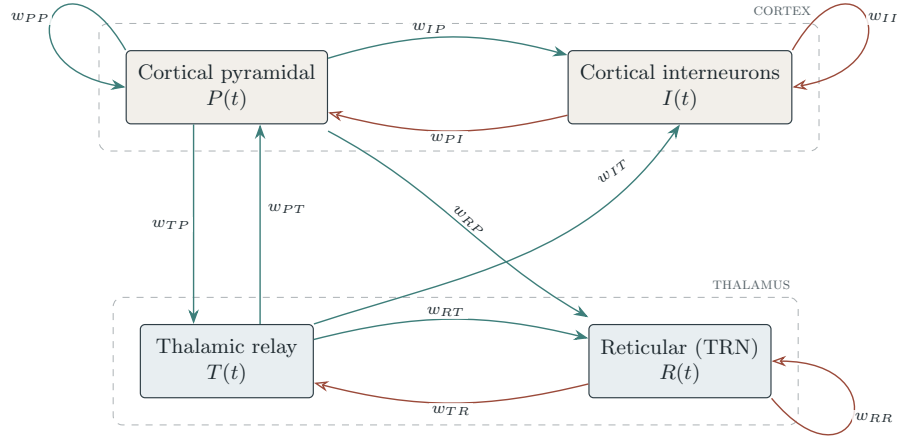
\begin{figure}[htbp]
    \centering
    \scalebox{0.88}{
    \begin{tikzpicture}[node distance=18mm and 36mm]
        \node[jfcortex] (P) {Cortical pyramidal\\$P(t)$};
        \node[jfcortex, right=of P] (I) {Cortical interneurons\\$I(t)$};
        \node[jfthal, below=30mm of P] (T) {Thalamic relay\\$T(t)$};
        \node[jfthal, below=30mm of I] (R) {Reticular (TRN)\\$R(t)$};
        \begin{scope}[on background layer]
            \node[jfgroup, fit=(P)(I)] (cortex) {};
            \node[jfgroup, fit=(T)(R)] (thalamus) {};
        \end{scope}
        \node[font=\scriptsize\itshape, text=jfSlate!80, anchor=south east]
              at (cortex.north east) {\textsc{cortex}};
        \node[font=\scriptsize\itshape, text=jfSlate!80, anchor=south east]
              at (thalamus.north east) {\textsc{thalamus}};
        \draw[jfexc] (P) to[bend left=16]
              node[jflabel,above,pos=0.42,yshift=1pt]{$w_{IP}$} (I);
        \draw[jfinh] (I) to[bend left=14]
              node[jflabel,below,pos=0.50]{$w_{PI}$} (P);
        \draw[jfexc] (P.north west) to[out=120,in=170,looseness=9]
              node[jflabel,left,pos=0.5]{$w_{PP}$} (P.west);
        \draw[jfinh] (I.north east) to[out=60,in=10,looseness=9]
              node[jflabel,right,pos=0.5]{$w_{II}$} (I.east);
        \draw[jfexc] (T) to[bend left=14]
              node[jflabel,above,pos=0.48]{$w_{RT}$} (R);
        \draw[jfinh] (R) to[bend left=14]
              node[jflabel,below,pos=0.50]{$w_{TR}$} (T);
        \draw[jfinh] (R.south east) to[out=-50,in=0,looseness=9]
              node[jflabel,right,pos=0.5]{$w_{RR}$} (R.east);
        \draw[jfexc] ($(P.south)+(-5mm,0)$) --
              node[jflabel,left]{$w_{TP}$} ($(T.north)+(-5mm,0)$);
        \draw[jfexc] ($(T.north)+(5mm,0)$) --
              node[jflabel,right,pos=0.58,xshift=2pt]{$w_{PT}$} ($(P.south)+(5mm,0)$);
        \draw[jfexc] (T.north east) to[out=18,in=235] (I.south);
        \node[jflabel, rotate=38] at ($(I.south)+(-10mm,-7mm)$) {$w_{IT}$};
        \draw[jfexc] ($(P.south east)+(0mm,-1mm)$) to[out=-28,in=150]
              node[jflabel,sloped,above,pos=0.52]{$w_{RP}$}
              ($(R.north west)+(0mm,1mm)$);
    \end{tikzpicture}}
    \caption{\textbf{Architecture of the corticothalamic--reticular Wilson--Cowan model.} Filled teal arrows denote excitatory couplings and hollow brick arrows inhibitory couplings; each weight $w_{XY}$ is the directed connection from population $Y$ to population $X$. Baseline values and exploration ranges are given in Table~\ref{tab:coupling_refs}. The delay is not specific to any one population: every coupling transmits the filtered state $(h_\rho*X)(t)$ rather than the instantaneous state $X(t)$, under either of the two kernels compared in Section~\ref{subsec:delay-models}.}
    \label{fig:architecture}
\end{figure}

\vspace{2mm}
\noindent \emph{\textbf{Intrinsic and electrical coupling within the TRN.}} Intrinsic membrane currents in the TRN, particularly low-threshold T-type calcium channels and hyperpolarization-activated H-currents, are essential for generating the rebound bursts that drive each spindle cycle~\cite{jahnsen1984electrophysiological,huguenard1992simulation}. Moreover, TRN neurons are extensively coupled by gap junctions, forming an electrical syncytium that synchronizes burst timing across the nucleus~\cite{landisman2002electrical,DeleuzeHuguenard2006}. Once initiated, cortical circuits can further amplify and propagate these spindles: intracortical reverberation via cortico-cortical loops sustains and shapes spindle activity beyond the thalamus~\cite{bazhenov2002model}. At the mean-field level, both the chemical inhibition internal to the nucleus and the net population-level consequence of its electrical coupling are absorbed into the single self-coupling $w_{RR}$; this is a deliberate reduction, since gap junctions synchronize rather than inhibit, and a one-variable population description cannot separate the two.

\vspace{2mm}
\noindent \emph{\textbf{Intra-cortical connections~\cite{harris2015neocortical}.}} Pyramidal cells form recurrent connections with local neurons of the same class. In turn, the primary targets of Vip and single-bouquet cells are other interneurons, especially Sst interneurons as well as Pvalb-positive basket cells. The excitatory and inhibitory cortical populations are also interconnected: cortical (Sst and Pvalb) interneurons receive excitatory input from local neurons and in turn inhibit excitatory cortical cells on their dendrites and somata respectively. Let us also remind that the inhibitory interneurons do not send significant long-range projections, including the relay and the TRN in particular~\cite{kepecs2014interneuron}. These interactions are represented by $w_{PP}$, $w_{IP}$, $w_{PI}$ and $w_{II}$, and by the vanishing $I \to T$ and $I \to R$ entries of the connectivity matrix.

\vspace{2mm}
\noindent \emph{\textbf{Brainstem modulatory inputs to TRN.}} Anatomical tracing has mapped dense innervation of the TRN and relay nuclei by cholinergic, noradrenergic and serotonergic systems~\cite{hallanger1987origins,jones2003arousal,mccormick1990noradrenergic}, which dynamically adjust TRN and relay cell excitability and thereby control spindle amplitude and duration in a state-dependent manner~\cite{mccormick1990noradrenergic}. Cholinergic inputs from the basal forebrain likewise modulate the activity of cortical neurons, contributing to the distinct network dynamics of different brain states~\cite{lee2012neuromodulation}. Slow-wave sleep is characterized by a low level of these neuromodulators and by the resulting progressive hyperpolarization of thalamocortical and reticular neurons, which deinactivates a low-threshold Ca\textsuperscript{2+} current and permits the burst-firing mode crucial to spindle generation~\cite{jahnsen1984electrophysiological}. In the present model this state is fixed rather than dynamic: the constant background drive $B$, together with the population-specific offsets $I_T$ and $I_R$, represents the low neuromodulatory tone characteristic of NREM sleep, and sleep--wake state transitions lie outside the present scope.

\subsection{Model equations and parameters}
\label{sec:equations}

Our basic model realizes this connectivity architecture as a system of Wilson--Cowan type equations, in which the four variables ($P$, $I$, $T$ and $R$) represent temporal mean-field activity in our four CTRC populations, respectively: cortical pyramidal cells $P$, cortical interneurons $I$, thalamic relay cells $T$ and reticular thalamic nucleus $R$. These populations integrate the sum of internal and external inputs via sigmoidal functions

\begin{equation}
{\cal S}_{\alpha_X,\theta_X,\beta_X}(u) = \frac{\alpha_X}{1+\exp(-\beta_X[u-\theta_X])} - \frac{\alpha_X}{1+\exp(\beta_X \theta_X)},
\label{eq:sigmoid}
\end{equation}

\noindent where $X \in \{ P,I,T,R \}$ is the respective integrating node. The subtracted constant guarantees ${\cal S}(0)=0$, so that a silent input yields no output; note, however, that it also gives ${\cal S}(u)<0$ for $u<0$, so that the vector field does not confine trajectories to the non-negative orthant. Non-negativity of the four activity variables is therefore an admissibility criterion imposed on solutions rather than a property guaranteed by the equations, and a trajectory leaving the non-negative range, while mathematically well defined, no longer represents a realizable population state; the saturation factors supply the complementary upper bound $X \le 1/r_X$. For every coupling parameter, $w_{XY}>0$ denotes the strength of the directed projection from population $Y$ to population $X$, namely, the first index is the target and the second index is the source. Then we can write the equations governing the behavior of the system, without yet considering distributed delays:

\begin{align}
    \tau \frac{dP}{dt} &= -P + (1 - r_P P) \cdot {\cal S}_{\alpha_P, \theta_P, \beta_P} \left( w_{PP}P - w_{PI}I + w_{PT}T + B + I_P \right) \label{eq:P} \\
    \tau \frac{dI}{dt} &=  -I + (1 - r_I I) \cdot {\cal S}_{\alpha_I, \theta_I, \beta_I} \left( w_{IP}P - w_{II}I + w_{IT}T + B \right) \label{eq:I} \\
    \tau \frac{dT}{dt}   &= -T + (1 - r_T T) \cdot {\cal S}_{\alpha_T, \theta_T, \beta_T} \left( w_{TP}P - w_{TR}R + B + I_T \right) \label{eq:T} \\
    \tau \frac{dR}{dt}   &= -R + (1 - r_R R) \cdot {\cal S}_{\alpha_R, \theta_R, \beta_R} \left( w_{RP}P + w_{RT}T - w_{RR}R + B + I_R \right) \label{eq:R}
\end{align}

\noindent Here $B$ is a common background drive applied to every population, while $I_P$, $I_T$ and $I_R$ are population-specific offsets for the pyramidal, relay and reticular populations; the interneuron population receives only the common drive, equivalently $I_I=0$. At the baseline configuration used throughout, $I_P=0$, so that the pyramidal population likewise receives only the common drive. The excitatory populations $P$ and $T$ share a single response class, ${\cal S}_P \equiv {\cal S}_T \equiv {\cal S}_E$, while the interneuron and reticular populations use ${\cal S}_I$ and ${\cal S}_R$ respectively, the latter with a higher maximal output ($\alpha_R = 75$ against $\alpha_E = \alpha_I = 50$). A single time constant $\tau$ is shared by all four populations, although TRN bursting, relay rebound and cortical integration operate on distinct intrinsic timescales; this is deliberate, since it leaves the mean delay as the only temporal parameter varied and makes the dimensionless ratio $q=\rho/\tau$ of Section~\ref{subsec:linearization} the sole temporal quantity governing spectral stability. Population-specific time constants are deferred to future work.

The values and ranges of the coupling and sigmoidal parameters and of the external inputs were set based on existing modeling literature and empirically demonstrated relationships between their magnitudes. In the tables below, we describe these ranges and the references used.

\begin{table}[htbp]
\centering
\small
\caption{Non-coupling parameters: values, roles and literature support.}
\label{tab:noncoupling_refs}
\vspace{0.15cm}
\begin{tabular}{@{}llll@{}}
\toprule
{\bf Symbol} & {\bf Value} & {\bf Role} & {\bf Reference(s)}\\
\midrule
$(\alpha_E,\theta_E,\beta_E)$ & $(50,\,4.0,\,1.2)$ & Excitatory sigmoid ${\cal S}_E$ ($P$, $T$) & \cite{WilsonCowan1972,Marreiros2008,SchellenbergerCosta2016} \\
$(\alpha_I,\theta_I,\beta_I)$ & $(50,\,3.8,\,2.5)$ & Interneuron sigmoid ${\cal S}_I$ & \cite{WilsonCowan1972,Marreiros2008,SchellenbergerCosta2016} \\
$(\alpha_R,\theta_R,\beta_R)$ & $(75,\,3.8,\,2.5)$ & TRN sigmoid ${\cal S}_R$ (higher output) & \cite{Marreiros2008,SchellenbergerCosta2016,pinault2004thalamic} \\
$\tau$ & $1.0$ & Population time constant ($1$ s) & \cite{WilsonCowan1972} \\
$(r_P,r_I,r_T,r_R)$ & $(0.05,\,0.03,\,0.03,\,0.03)$ & Saturation coefficients & \cite{WilsonCowan1972} \\
$B$ & $0.6$ & Common background drive (all populations) & \cite{WilsonCowan1972,Marreiros2008} \\
$I_P$ & $0$ & Pyramidal-specific offset & --- \\
$I_T$ & $2.8$ & Relay-specific offset & \cite{Cakan2023,SchellenbergerCosta2016} \\
$I_R$ & $1.5$ & Reticular-specific offset & \cite{Cakan2023,SchellenbergerCosta2016} \\
\bottomrule
\end{tabular}
\end{table}

\begin{table}[htbp]
\centering
\small
\caption{Coupling weights and the delay parameter: pathways, baseline values, exploration ranges, functional roles and literature support. Baseline values define the reference configuration; ranges indicate the parameters varied in the continuation and delay studies.}
\label{tab:coupling_refs}
\vspace{0.15cm}
\begin{tabular}{@{}llclll@{}}
\toprule
{\bf Parameter} & {\bf Pathway} & {\bf Baseline} & {\bf Baseline exploration set} & {\bf Interpretation} & {\bf Reference(s)} \\
\midrule
$w_{PP}$ & PYR $\to$ PYR & $1.1$   & $[0.8,\,1.2]$ & Recurrent cortical excitation & \cite{Brunel2000,harris2015neocortical} \\
$w_{PI}$ & IN $\to$ PYR  & $0.4$   & fixed         & Pyramidal inhibition          & \cite{IsaacsonScanziani2011,PotjansDiesmann2014} \\
$w_{PT}$ & TR $\to$ PYR  & $0.4$   & fixed         & Thalamocortical drive         & \cite{Jajcay2022,ReichovaSherman2004} \\
$w_{IP}$ & PYR $\to$ IN  & $0.7$   & fixed         & Interneuron recruitment       & \cite{FreundKatona2007,IsaacsonScanziani2011} \\
$w_{II}$ & IN $\to$ IN   & $0.1$   & fixed         & Interneuron self-inhibition   & \cite{Pfeffer2013} \\
$w_{IT}$ & TR $\to$ IN   & $0.7$   & fixed         & Relay--interneuron drive      & \cite{Cruikshank2007} \\
$w_{TP}$ & PYR $\to$ TR  & $1.25$  & $[1.0,\,1.5]$ & Corticothalamic drive         & \cite{Jajcay2022} \\
$w_{TR}$ & TRN $\to$ TR  & $0.8$   & $[0.6,\,1.0]$ & Reticular inhibition of relay & \cite{SchellenbergerCosta2016,zhang2004corticothalamic} \\
$w_{RP}$ & PYR $\to$ TRN & $0.7$   & fixed         & Cortico-reticular collateral  & \cite{zhang2004cortical} \\
$w_{RT}$ & TR $\to$ TRN  & $0.4$   & $[0.1,\,0.4]$ & Relay--reticular excitation   & \cite{SchellenbergerCosta2016,pinault2004thalamic} \\
$w_{RR}$ & TRN $\to$ TRN & $0.2$   & $[0.1,\,0.3]$ & Reticular self-limitation     & \cite{DeleuzeHuguenard2006,landisman2002electrical} \\
$\rho$   & all couplings & $0$     & $[0,\,5]$     & Mean coupling delay           & --- \\
\bottomrule
\end{tabular}

\end{table}

\noindent We denote
\(
    X=(P,I,T,R)^{\mathsf T}
\),
and define
\[
    \Theta(X)
    =
    \operatorname{diag}
    \bigl(
       1-r_PP,\,
       1-r_II,\,
       1-r_TT,\,
       1-r_RR
    \bigr),
\]
together with the componentwise transfer function
\[
    \mathbf S(z)
    =
    \bigl(
       S_P(z_P),\,
       S_I(z_I),\,
       S_T(z_T),\,
       S_R(z_R)
    \bigr)^{\mathsf T}.
\]
The signed connectivity matrix and the input vector are
\[
C=
\begin{pmatrix}
 w_{PP} & -w_{PI} &  w_{PT} & 0\\
 w_{IP} & -w_{II} &  w_{IT} & 0\\
 w_{TP} & 0       & 0       & -w_{TR}\\
 w_{RP} & 0       & w_{RT}  & -w_{RR}
\end{pmatrix},
\qquad
U=
\begin{pmatrix}
 B+I_P\\
 B\\
 B+I_T\\
 B+I_R
\end{pmatrix}.
\]
The zero entries of $C$ encode the anatomical constraints described in Section~\ref{sec:circuitry}: the absence of TRN projections to cortex, of long-range interneuron projections, and of relay-to-relay coupling. Therefore, the instantaneous model can be written compactly as
\begin{equation}\label{eq.model.no.delay}
    \tau\dot X(t)
    =
    -X(t)
    +
    \Theta(X(t))
    \mathbf S\!\left(CX(t)+U\right).
\end{equation}

\subsection{Temporal coupling formulations}
\label{subsec:delay-models}

In this paper, we also extend the instantaneous model using two complementary delay
formulations  having the same mean-delay
parameter $\rho$: a weak Gamma distributed kernel and a discrete delay. The
weak Gamma kernel is closely related to the exponentially weighted
temporal response underlying classical Wilson--Cowan formulations and is
widely used to represent heterogeneous transmission and integration times
in neural-population models. By contrast, the discrete delay concentrates
the delayed influence at a single time lag and provides a simple,
widely used benchmark. Comparing these formulations allows us to separate
the effects of the characteristic delay from those of temporal dispersion.
The comparison is further motivated by our previous results \cite{KaslikEtAl2022,KaslikEtAl2024}, in which
discrete delays produced richer bifurcation structures and more complex
oscillatory dynamics than Gamma-distributed delays.

We write the general distributed delayed system as
\begin{equation}\label{eq.model.general.delay}
    \tau\dot X(t)
    =
    -X(t)
    +
    \Theta(X(t))
    \mathbf S\!\left(C(h_\rho\ast X)(t)+U\right),
\end{equation}
where $h_\rho$ is the delay kernel with mean delay $\rho$ and $h_\rho\ast X$ represents the convolution.

\subsubsection{weak Gamma distributed delay}

The weak Gamma kernel with mean delay $\rho$ is
\[
    h_\rho(s)
    =
    \frac{1}{\rho}e^{-s/\rho},
    \qquad s\ge0.
\]
Defining the filtered activity
\[
    Y(t)
    =
    \int_0^\infty h_\rho(s)X(t-s)\,ds,
\]
the distributed-delay system is
\begin{equation}\label{eq.model.weak.gamma}
\begin{aligned}
    \tau\dot X(t)
        &=
        -X(t)
        +
        \Theta(X(t))
        \mathbf S\!\left(CY(t)+U\right),\\
    \rho\dot Y(t)
        &=
        X(t)-Y(t).
\end{aligned}
\end{equation}
The second equation represents the linear-chain representation of the
weak Gamma convolution.

\subsubsection{Discrete delay}

For a fixed delay $\rho>0$ (i.e. discrete delay kernel $h_\rho(s)=\delta(s-\rho)$) the corresponding retarded functional
differential equation is
\begin{equation}\label{eq.model.discrete.delay}
    \tau\dot X(t)
    =
    -X(t)
    +
    \Theta(X(t))
    \mathbf S\!\left(CX(t-\rho)+U\right),
\end{equation}
which is associated with a continuous history
\[
    X(\theta)=\varphi(\theta),
    \qquad
    \theta\in[-\rho,0].
\]

\subsection{Linearization  and characteristic equation}
\label{subsec:linearization}

For the models considered in this paper, the location of
the equilibria is independent of both the delay parameter and the choice
of temporal kernel. Indeed, for both the weak Gamma and discrete-delay kernels, the equilibrium equation is
identical to that of the delay-free system:
\[
    -X^\ast
    +
    \Theta(X^\ast)
     \mathbf S(CX^\ast+U)
    =
    0.
\]
Hence, for fixed model parameters, introducing either type of delay does
not change the set of equilibria. It may only change their stability and
the bifurcation structure organized around them.

Considering an equilibrium $X^\ast$, let us denote
\[
    u^\ast=CX^\ast+U,
    \qquad
    \Theta_\ast=\Theta(X^\ast),\qquad 
    \Phi_\ast
    =
    \operatorname{diag}
    \bigl(
       S'_P(u^\ast_P),\,
       S'_I(u^\ast_I),\,
       S'_T(u^\ast_T),\,
       S'_R(u^\ast_R)
    \bigr).
\]
Further denoting
\[
    D=\Theta_\ast^{-1},
    \qquad
    K=\Theta_\ast\Phi_\ast C,
\]
the general linearized equation is
\[
    \tau\dot\xi(t)
    =
    -D\xi(t)
    +
    K\int_0^\infty h_\rho(s)\xi(t-s)\,ds,
\]
and its characteristic equation is
\begin{equation}\label{eq.characteristic-equation}
    \Delta(\lambda;\rho)
    :=
    \det\!\left(
       \tau\lambda I+D-\widehat h_\rho(\lambda)K
    \right)
    =0,
\end{equation}
where $\widehat h_\rho$ represents the Laplace transform of the corresponding delay kernel. More precisely, 
for the two kernels considered here,
\(
    \widehat h_\rho(\lambda)
    =
    (1+\rho\lambda)^{-1}\) for the weak Gamma coupling,
and
\(
    \widehat h_\rho(\lambda)
    =
    e^{-\rho\lambda}
\) for the discrete delay.

In the delay-free case (i.e. $\rho=0$), the characteristic equation \eqref{eq.characteristic-equation} becomes
\begin{equation}\label{eq.char.no.delay}
     \Delta(\lambda;0)
    :=
    \det\!\left(
       \tau\lambda I+D-K
    \right)
    =0,
\end{equation}
and hence, the characteristic roots satisfy $\tau\lambda\in\sigma(K-D)$.

For $\rho>0$, let us assume in general that \(\{h_\rho\}_{\rho>0}\) is a fixed-shape scale family of
normalized delay kernels
\[
    h_\rho(s)
    =
    \frac{1}{\rho}
    h\!\left(\frac{s}{\rho}\right),
    \qquad s\geq 0,\qquad\text{where }\int_0^\infty h(s)\,ds=1.
\]
If, in addition,
\[
    \int_0^\infty s h(s)\,ds=1,
\]
then \(\rho>0\) is the mean delay of \(h_\rho\). Moreover,
\(
    \widehat h_\rho(\lambda)
    =
    \widehat h(\rho\lambda),
\)
and introducing $z=\tau\lambda$ and 
\[
    q=\frac{\rho}{\tau},
\]
the characteristic equation \eqref{eq.characteristic-equation}
can be written as
\begin{equation}\label{eq.scaled-characteristic-equation}
    \widetilde\Delta(z;q)
    :=
    \det\!\left(
       zI
       +
       D
       -
       \widehat h(qz)K
    \right)
    =
    0.
\end{equation}
Consequently, 
the spectral stability of the equilibrium depends on $\tau$ and
$\rho$ only through the dimensionless ratio
$q=\rho/\tau$. It is also important to note that for $q\rightarrow 0$, as $\widehat h(0)=1$, we recover the corresponding delay-free characteristic equation. 

\begin{rem}
    Varying $q$ alone cannot generically create a zero-eigenvalue bifurcation. A saddle-node, transcritical or pitchfork bifurcation can only occur when $\det(D-K)=0$, which gives a boundary determined by the connectivity parameters. If the coupling parameters lie away from this boundary, every generic change of equilibrium stability caused by varying $q$ must occur through a nonzero imaginary pair, hence through a Hopf crossing. More precisely, if the equilibrium is asymptotically stable for the delay-free system, there exists $q_0>0$, such that the equilibrium remains asymptotically stable for any $q\in(0,q_0)$ and any delay kernel $h$ (see Theorem \ref{thm:small-q-stability} in Appendix A).

Now let us assume that the equilibrium is asymptotically stable at \(q=0\) and that a
finite first stability threshold \(q^\ast>0\) exists. If the characteristic
roots reaching the imaginary axis at \(q=q^\ast\) form a simple pair
\[
z(q^\ast)=\pm i\Omega^\ast,
\qquad
\Omega^\ast>0,
\]
and the crossing is transversal, then this first crossing must be
destabilizing:
\[
\left.
\frac{d}{dq}\operatorname{Re}z(q)
\right|_{q=q^\ast}>0.
\]
If, in addition, the standard nonlinear Hopf nondegeneracy conditions hold,
a local branch of periodic solutions bifurcates from the equilibrium. However, it is important to note that
neither stability at $q=0$ nor the first-crossing observation guarantees
that a finite threshold $q^\ast$ exists. For some parameter configurations, the equilibrium may
remain stable for all $q>0$.

A sufficient condition excluding this
possibility is provided by Theorem~\ref{thm:large-delay} (Appendix A). More precisely,
if the limiting characteristic equation
\[
    \Delta_\infty(u)
    =\det\!\bigl(D-\widehat h(u)K\bigr)=0
\]
has a simple root $u_+$ with $\operatorname{Re}u_+>0$, then the
equilibrium is unstable for all sufficiently large $q$. Since it is
stable for sufficiently small $q$, continuity of the characteristic
roots implies the existence of at least one finite critical value
$q^\ast>0$ such that
\(    \widetilde\Delta(i\Omega^\ast;q^\ast)=0
\)
for some $\Omega^\ast>0$.

Furthermore, in the discrete-delay case, where
$\widehat h(qz)=e^{-qz}$, every such critical pair
$(q^\ast,\Omega^\ast)$ generates the infinite sequence
\[
    q_n^\ast
    =q^\ast+\frac{2\pi n}{\Omega^\ast},
    \qquad n\in\mathbb{Z}_+.
\]
Hence, whenever one critical delay exists, there are infinitely many
delay values at which the characteristic equation has imaginary roots.
The later values $q_n^\ast$ may correspond to additional
stability switches and Hopf bifurcations, if the relevant nondegeneracy conditions are fulfilled.
\end{rem}

\section{Modeling results}

We begin by analyzing the stability and bifurcation structure of the delay-free system, both to obtain a broad baseline understanding of the role of the parameters and to calibrate the model. We will then restrict our attention to the key coupling weights and investigate how distributed delays interact with these crucial coupling parameters to shape the system's dynamics. The numerical analysis was performed in MatCont, version 7p4~\cite{Dhooge2008MatCont}, implemented in MATLAB~\cite{MATLAB2024a}.

\subsection{Attractors and transitions for the system with no delays}

\noindent \emph{\textbf{Strength of pyramidal projections $w_{PP}$ and $w_{TP}$.}} We first explored the effect of varying $w_{PP}$ and $w_{TP}$ on the system's temporal dynamics. Increasing $w_{PP}$ strengthens recurrent excitation within the cortical pyramidal population, making cortical activity more self-sustaining; decreasing it weakens this internal drive. Increasing $w_{TP}$ strengthens cortical drive to the thalamic relay population, increasing the influence of cortical activity on relay dynamics. We illustrate this combined dependence in Figure~\ref{wPP_bif}. In each panel (one for each variable $P$, $I$, $T$, and $R$, respectively), we show the dependence on $w_{PP}$ in the form of a bifurcation diagram, with different diagrams corresponding to different fixed values of $w_{TP}$.

The figure identifies $w_{PP}$ as a primary access parameter for oscillatory dynamics. Across all sampled values of $w_{TP}$, increasing recurrent pyramidal excitation first destabilizes the equilibrium through a supercritical Hopf bifurcation (marked with a star), opening a bounded oscillatory window (shown as the shaded region). As $w_{PP}$ continues to increase, the oscillations grow in amplitude, enter and later exit a spindle-compatible window, and eventually terminate at a limit point cycle. The persistence of this same qualitative bifurcation skeleton across different fixed values of $w_{TP}$ shows that the dominant role of $w_{PP}$ is structural rather than just modulatory: it controls whether the cortical component of the CTRC loop becomes sufficiently self-sustained to support rhythmic activity at all. All equilibrium curves additionally presented with a limit point (shown as a colored dot), producing a second stable equilibrium branch and, therefore, a window of hysteresis and bistability in conjunction with the original equilibrium branch and the stable cycle. This branch persists over a wide range of $w_{PP}$ values, and survives as the only attractor for values of $w_{PP}$ beyond the limit point cycle that ends oscillations, but is only biologically viable for intermediate values of $w_{TP}$ (since the $T^*$ component becomes negative for high $w_{TP}$ and exceeds the adopted range for low $w_{TP}$).

Overall, varying corticothalamic drive $w_{TP}$ did not change the underlying bifurcation mechanism, but only shifted the positions of the Hopf point and the LPC, effectively moving the oscillatory window to the right as $w_{TP}$ increased. This suggests that corticothalamic input to the relay population acts primarily as a positioning or enabling parameter, whereas recurrent pyramidal excitation provides the core destabilizing drive. Thus, Figure~\ref{wPP_bif} establishes that recurrent cortical excitation is not simply an amplitude control, but a fundamental determinant of whether the system can access a biologically relevant oscillatory corridor.

\begin{figure}[h!]
\begin{center}
\includegraphics[width=0.9\textwidth]{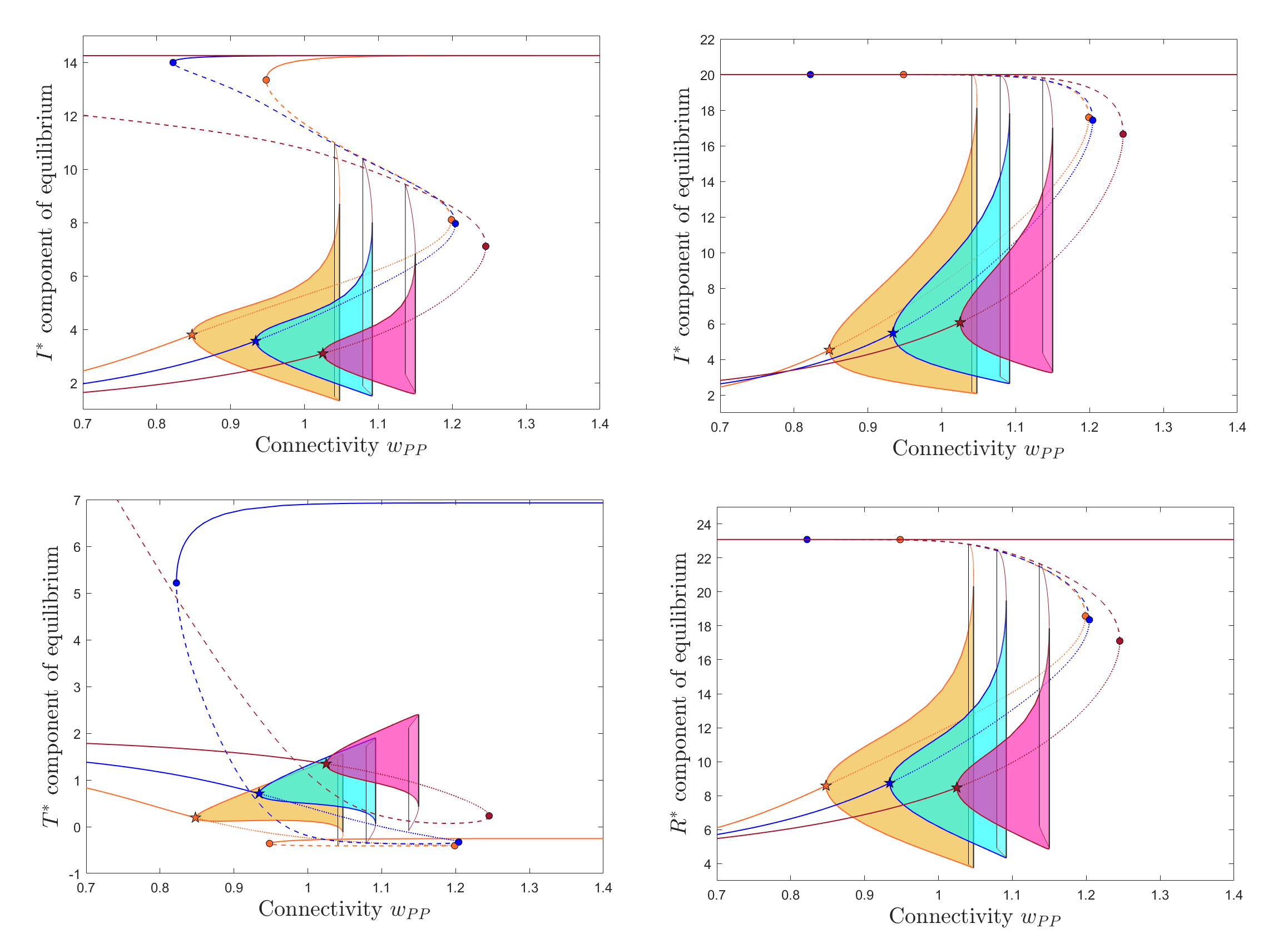}
\end{center}
\caption{\small \emph{{\bf Dependence on $w_{PP}$ for sample $w_{TP}$ values.} The panels show the transitions of the system as $w_{PP}$ is increased, for three different fixed values of $w_{TP}$: $w_{TP} =1$ (orange); $w_{TP}=1.25$ (blue); $w_{TP}=1.5$ (brown). Each panel represents the same bifurcation diagram, from the perspective of a different system component: $P$, $I$, $T$ and $R$ respectively. Stable equilibrium branches are shown as solid curves, unstable branches as dashed and dotted curves (based on the number of unstable directions). Limit points are marked with large dots, and Hopf bifurcations are marked by stars, with the evolution of the corresponding stable cycles shown as shaded areas (with the shade matching the respective equilibrium curve). Limit point cycles are marked by a black bar, where the stable cycle meets the unstable cycle (shown with no color and full transparency). For this simulation, $w_{TR}=0.8$, $w_{RT}=0.4$, $w_{RR}=0.2$, and the other parameters were fixed to their table baseline values.}}
\label{wPP_bif}
\end{figure}

Figure~\ref{wTP_bif} provides a complementary view of how the two pyramidal weights, $w_{PP}$ and $w_{TP}$, cooperate to shape the onset, termination, and placement of oscillatory regimes. In this analysis, we focus on the transitions into and out of oscillations as $w_{TP}$ increases, for several sample values of $w_{PP}$, as specified in the caption. In all three cases, the equilibrium curve exhibits two Hopf points, so that stable oscillatory regimes flank an intermediate range of $w_{TP}$ values for which the equilibrium is locally attracting. However, although both oscillatory windows are mathematically admissible, only the left one falls within the biologically relevant range for $w_{TP}$. Higher values of $w_{PP}$ shift the Hopf entry point of this right branch to the left, but not enough to bring it into a physiologically meaningful range.

More significantly, the location, width, and biological relevance of the left oscillatory window also depend strongly on $w_{PP}$. For smaller values of $w_{PP}$, this window is displaced toward implausible values of the $T$ component or even into the negative $w_{TP}$ range, making it progressively less physiologically meaningful. Thus, the role of $w_{TP}$ cannot be assessed in isolation: its ability to position the system within an admissible oscillatory window depends strongly on the background level of recurrent cortical excitation. Taken together, these results show that viable oscillatory regimes occupy only a relatively narrow and structured region of the $(w_{PP},w_{TP})$ plane. This region narrows further when attention is restricted to spindle-like dynamics.\\

\begin{figure}[h!]
\begin{center}
\includegraphics[width=0.9\textwidth]{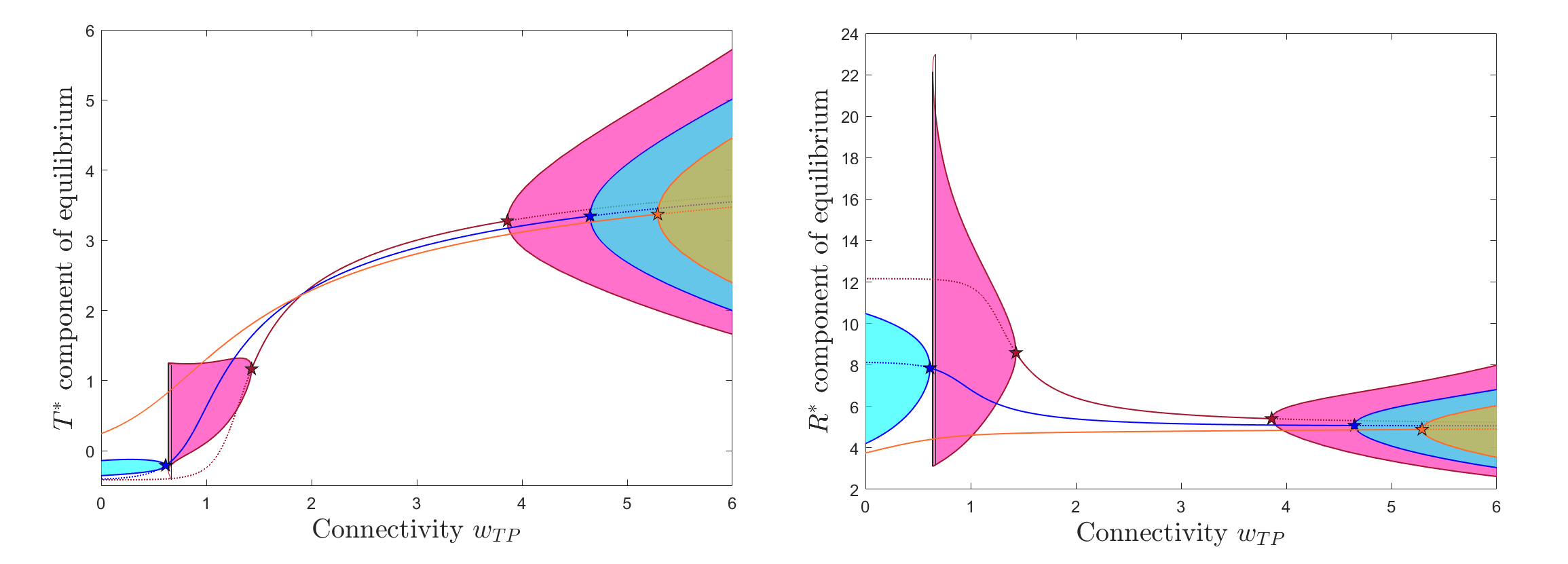}
\end{center}
\caption{\small \emph{{\bf Dependence on $w_{TP}$ for sample $w_{PP}$ values.} The panels show the transitions of the system as $w_{TP}$ is increased, for three different fixed values of $w_{PP}$: $w_{PP} =1$ (brown); $w_{PP}=0.75$ (blue); $w_{PP}=0.5$ (orange). Each panel represents the same bifurcation diagram, from the perspective of two different system components: $T$ and $R$ (left and right panel, respectively). Stable equilibrium branches are shown as solid curves, unstable branches as dotted curves (based on the number of unstable directions). Supercritical Hopf bifurcations are marked by stars, with the evolution of the corresponding stable cycles shown as shaded areas (with the shade matching the respective equilibrium curve). The limit point cycle for $w_{PP}=1$ is marked by a black bar, and the emerging unstable cycle is shown with no color and full transparency. For this simulation, $w_{TR}=0.8$, $w_{RT}=0.4$, $w_{RR}=0.2$, and the other parameters were fixed to their table baseline values.}}
\label{wTP_bif}
\end{figure}

\begin{figure}[h!]
\begin{center}
\includegraphics[width=0.9\textwidth]{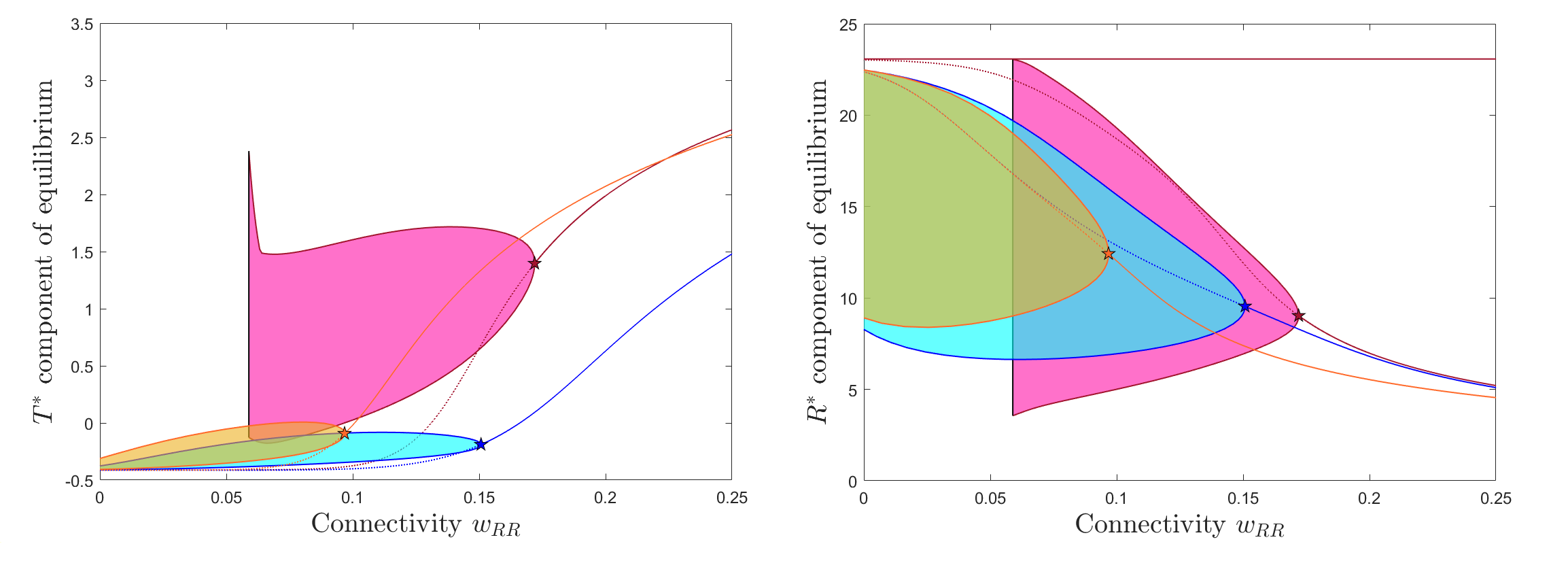}
\end{center}
\caption{\small \emph{{\bf Dependence on $w_{RR}$ for sample $w_{PP}$ and $w_{TP}$ values.} The panels show the transitions of the system as $w_{RR}$ is increased, for three different fixed pairs of $w_{PP}$ and $w_{TP}$: $(w_{PP},w_{TP})=(1,1.75)$ (pink); $(w_{PP},w_{TP})=(0.75,1)$ (blue); $(w_{PP},w_{TP})=(0.75,1.75)$ (orange). Each panel represents the same bifurcation diagram, from the perspective of $T$ and $R$ (left and right panels, respectively). Stable equilibrium branches are shown as solid curves, unstable branches as dashed and dotted curves (based on the number of unstable directions). Limit points are marked with large dots, and Hopf bifurcations are marked by stars, with the evolution of the corresponding stable cycles shown as a shaded area (with the shade matching the respective equilibrium curve). The limit point cycle found for $(w_{PP},w_{TP})=(1,1.75)$ is marked by a black bar. For this simulation, $w_{TR}=0.8$, $w_{RT}=0.4$, and the other parameters were fixed to their table baseline values.}}
\label{wRR_bif}
\end{figure}

\noindent \emph{\textbf{Strength of reticular self-inhibition $w_{RR}$.}} We next investigate how this balance conditions the effect of other coupling weights on the system's dynamics. In particular, we explore the dependence on the strength of reticular self-inhibition $w_{RR}$, reflecting the efficiency of synaptic and electrical interactions within the TRN population. Figure~\ref{wRR_bif} illustrates this dependence for several pairs $(w_{PP},w_{TP})$ chosen from the biologically relevant region suggested by the previous analysis.

The effect of $w_{RR}$ is strongly contingent on the background position of the system in the $(w_{PP},w_{TP})$ plane. In all cases shown, oscillations occur within a window of small values of $w_{RR}$ and disappear as the equilibrium regains stability through a supercritical Hopf bifurcation. However, the location, biological admissibility and geometry of this oscillatory window depend markedly on the chosen $(w_{PP},w_{TP})$ pair -- reinforcing the finding that insufficient $w_{PP}$ leads to biologically inadmissible cycles (negative in the $T^*$ component). Figure~\ref{wRR_bif} therefore supports a potentially hierarchical picture of parameter control: once $(w_{PP},w_{TP})$ place the system near a viable oscillatory regime, $w_{RR}$ governs how that regime is internally structured and how far it can persist. Along these lines, $w_{RR}$ can be interpreted primarily as a self-limitation parameter. Increasing $w_{RR}$ strengthens internal suppression within the reticular population, thereby compressing the oscillatory window, reshaping the associated cycles, and eventually terminating the oscillations altogether. This admissible region narrows further when one restricts attention to particular rhythms, such as spindle-like oscillations.\\

\begin{figure}[h!]
\begin{center}
\includegraphics[width=0.9\textwidth]{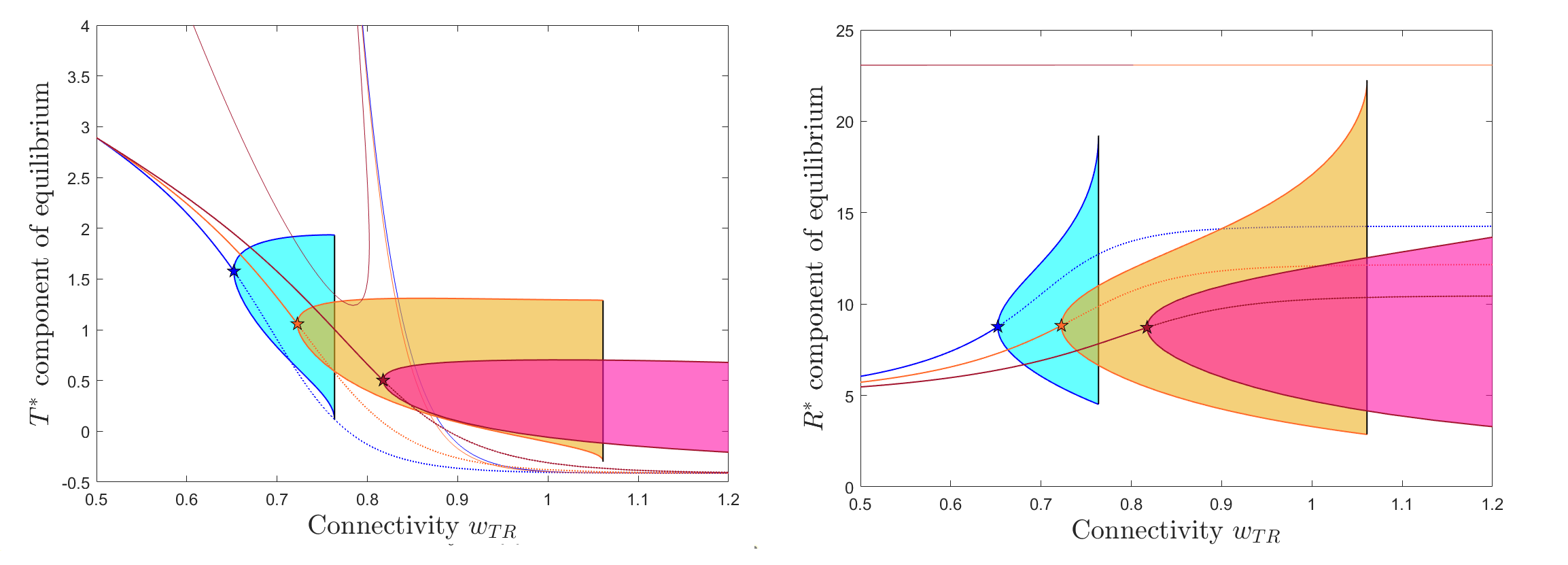}
\end{center}
\caption{\small \emph{{\bf Dependence on $w_{TR}$ for different $w_{PP}$ values.} The panels show the transitions of the system as $w_{TR}$ is increased, for three different fixed values of $w_{PP}$: $w_{PP}=1.1$ (blue); $w_{PP}=1$ (orange); $w_{PP}=0.9$ (brown). Each panel represents the same bifurcation diagram, from the perspective of $T$ and $R$ (left and right panels, respectively). Stable equilibrium branches are shown as solid curves, unstable branches as dashed and dotted curves (based on the number of unstable directions). Limit points are marked with large dots, and Hopf bifurcations are marked by stars, with the evolution of the corresponding stable cycles shown as a shaded area (with the shade matching the respective equilibrium curve). Limit point cycles are marked by a black bar. For this simulation, $w_{RT}=0.4$, $w_{RR}=0.2$ and the other parameters were fixed to their table baseline values.}}
\label{wTR_bif_diff_wPP}
\end{figure}

\noindent \emph{\textbf{Strength of TRN-to-relay coupling $w_{TR}$.}} We simulated the effect of varying $w_{TR}$ on oscillatory dynamics, for a range of recurrent cortical excitation $w_{PP}$ and relay-to-TRN coupling $w_{RT}$. Figure~\ref{wTR_bif_diff_wPP} shows that, when fixing $w_{RT}$, the influence of the TRN-to-relay coupling $w_{TR}$ still depends strongly on the level of recurrent cortical excitation $w_{PP}$. In all three cases, increasing $w_{TR}$ eventually drives the system through a supercritical Hopf threshold (marked by a star), beyond which a stable oscillatory window appears and persists until it terminates at an LPC. However, as expected, the location, width, and biological viability of this oscillatory window vary significantly with $w_{PP}$. For lower values of $w_{PP}$, the oscillatory window is reached earlier in $w_{TR}$ and remains relatively restricted. A second equilibrium branch acts as the system attractor for values of $w_{TR}$ beyond the LPC point marking the end of oscillations; however, this equilibrium quickly becomes biologically unsustainable (negative in the $T^*$ component). As $w_{PP}$ increases, the oscillatory window broadens, but also shifts so that oscillations are physiologically stopped before reaching the LPC, by entering the negative domain for $T$. The level of $w_{PP}$ also controls the oscillation amplitudes and duty cycle (not shown), effectively shaping access to different firing rhythms in the CTRC loop. Overall, this indicates that the same increase in inhibitory reticular drive can have qualitatively different effects depending on how strongly the cortical pyramidal population already supports recurrent excitation. The figure suggests that inhibition from TRN onto relay cells becomes dynamically effective only once the broader corticothalamic loop has already been primed for a specific oscillatory regime. In this sense, $w_{PP}$ can be viewed as acting as a gatekeeper to specific oscillatory regimes, while $w_{TR}$ controls the network's entry into and exit from oscillations within that regime.

\begin{figure}[h!]
\begin{center}
\includegraphics[width=0.9\textwidth]{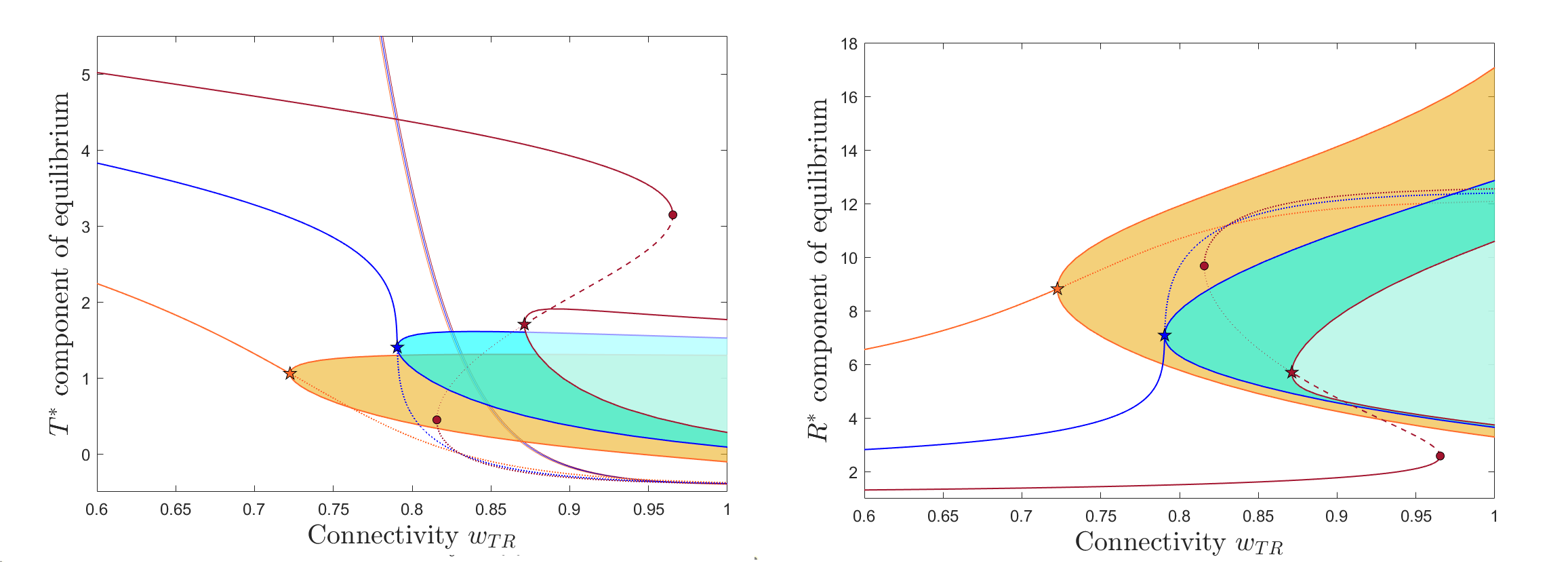}
\end{center}
\caption{\small \emph{{\bf Dependence on $w_{TR}$ for different $w_{RT}$ values.} The panels show the transitions of the system as $w_{TR}$ is increased, for fixed $w_{PP}=1$ and three different fixed values of $w_{RT}$: $w_{RT}=0.4$ (orange); $w_{RT}=0.2$ (blue); $w_{RT}=0.1$ (brown). Each panel represents the same bifurcation diagram, from the perspective of $T$ and $R$ (left and right panels, respectively). Stable equilibrium branches are shown as solid curves, unstable branches as dashed and dotted curves (based on the number of unstable directions). Limit points are marked with large dots, and Hopf bifurcations (both supercritical and subcritical) are marked by stars. The evolution of the stable cycles is shown as a shaded area (with the shade matching the respective equilibrium curve); unstable cycles are shown in no color and full transparency.}}
\label{wTR_bif_wPP_1_diff_wRT}
\end{figure}

\begin{figure}[h!]
\begin{center}
\includegraphics[width=0.9\textwidth]{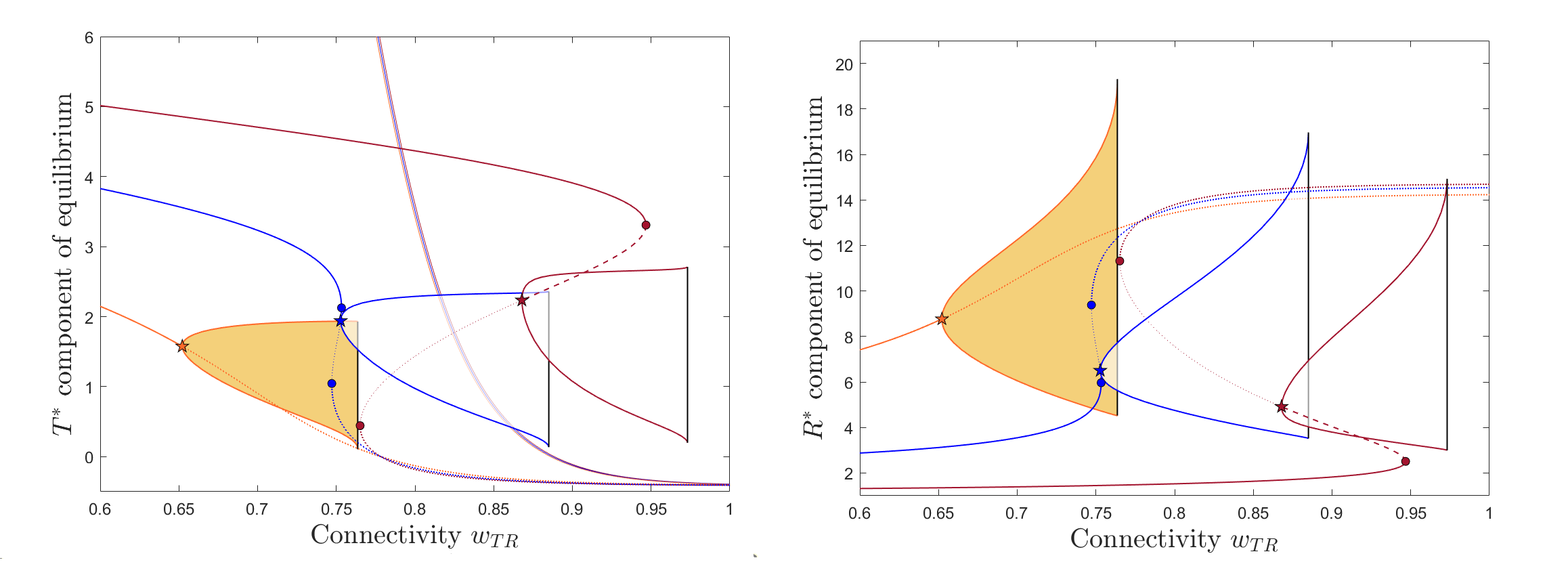}
\end{center}
\caption{\small \emph{{\bf Dependence on $w_{TR}$ for different $w_{RT}$ values.} The panels show the transitions of the system as $w_{TR}$ is increased, for fixed $w_{PP}=1.1$ and three different fixed values of $w_{RT}$: $w_{RT}=0.4$ (orange); $w_{RT}=0.2$ (blue); $w_{RT}=0.1$ (brown). Each panel represents the same bifurcation diagram, from the perspective of $T$ and $R$ (left and right panels, respectively). Stable equilibrium branches are shown as solid curves, unstable branches as dashed and dotted curves (based on the number of unstable directions). Limit points are marked with large dots, and Hopf bifurcations (both supercritical and subcritical) are marked by stars. The evolution of the stable cycles is shown as a shaded area (with the shade matching the respective equilibrium curve); unstable cycles are shown in no color and full transparency. Limit points of cycles are marked by a black bar.}}
\label{wTR_bif_wPP_1_1_diff_wRT}
\end{figure}

We next show how, once $w_{PP}$ grants access to an oscillatory behavior, the pair $(w_{TR},w_{RT})$ can further modulate the detailed structure of the oscillatory window.  Figures~\ref{wTR_bif_wPP_1_diff_wRT} and~\ref{wTR_bif_wPP_1_1_diff_wRT} illustrate the direct dependence on $w_{TR}$ for a range of $w_{RT}$ values, for a lower versus a higher value of cortical excitation ($w_{PP}=1$ and $w_{PP}=1.1$, respectively). In both cases, $w_{RT}$ further gates access to stable oscillations, and controls the onset point, width and geometry of the oscillatory window shown with respect to $w_{TR}$. 

Figure~\ref{wTR_bif_wPP_1_diff_wRT} illustrates this dual control when $w_{PP}=1$. For high $w_{RT}$ ($w_{RT}=0.4$, orange diagram), the system has a stable equilibrium in the biological range for low $w_{TR}$. This equilibrium crosses into stable oscillations at the Hopf bifurcation (orange star), sustained as $w_{TR}$ increases throughout its biological range. This scenario remains qualitatively similar when $w_{RT}$ is decreased to $w_{RT}=0.2$, with the reduced relay-to-TRN excitation producing only a shift to the right on the Hopf bifurcation, and modulation of the oscillation amplitudes. However, if $w_{RT}$ is decreased below a threshold level of relay-to-TRN excitation, the system undergoes a phase transition that presents with the creation of two limit points and hysteresis along the main equilibrium branch, which subsequently changes the nature of the Hopf bifurcation. The cycles born at the Hopf point along the brown equilibrium (representing $w_{RT}=0.1$) are unstable, leaving the system with a bistability window and no access to oscillatory behavior.

Figure~\ref{wTR_bif_wPP_1_1_diff_wRT} confirms that the interpretation emerging from Figure~\ref{wTR_bif_wPP_1_diff_wRT} is robust and not tied to a single cortical operating point. When $w_{PP}$ is increased, the pair $(w_{TR},w_{RT})$ continues to dominate the detailed organization of the oscillatory window in a consistent way, by controlling reticular recruitment and the inhibitory return onto the relay population. On this background, the exact location of the phase transitions and the extent of the oscillatory regimes shift with the cortical background. For example, a visible difference in higher $w_{PP}=1.1$ compared to lower $w_{PP}=1$ is that the cycles born at the Hopf point are unstable for both $w_{RT}=0.2$ and for $w_{RT}=0.1$, hence one does not need to lower $w_{RT}$ as much as in the previous case in order to block stable oscillations for the entire $w_{TR}$ range.

\begin{figure}[h!]
\begin{center}
\includegraphics[width=0.9\textwidth]{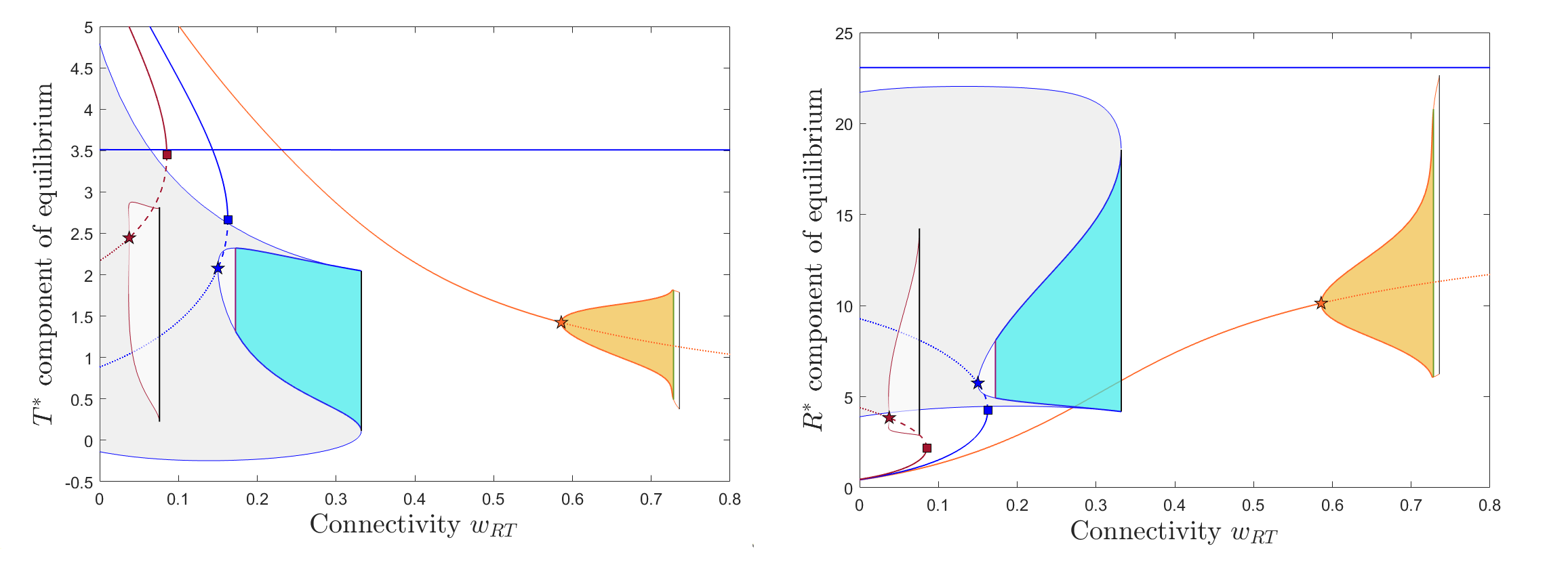}
\includegraphics[width=0.4\textwidth]{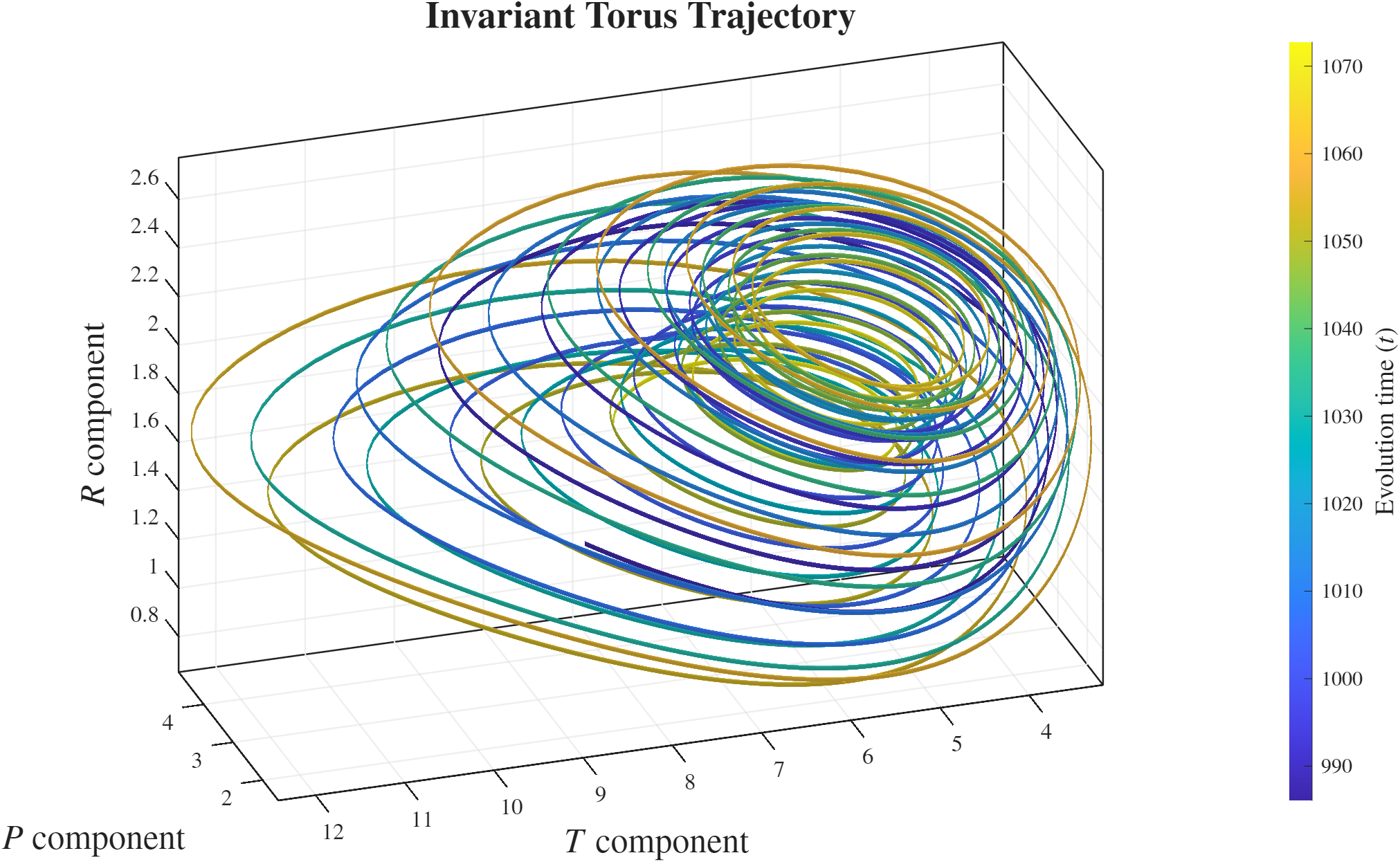}
\end{center}
\caption{\small \emph{{\bf Dependence on $w_{RT}$ for different $w_{TR}$ values.} The top panels show the transitions of the system as $w_{RT}$ is increased, for fixed $w_{PP}=1.1$ and three different fixed values of $w_{TR}$: $w_{TR}=0.6$ (orange); $w_{TR}=0.8$ (blue); $w_{TR}=1$ (brown). Each panel represents the same bifurcation diagram, from the perspective of $T$ and $R$ (left and right panels, respectively). Stable equilibrium branches are shown as solid curves, unstable branches as dashed and dotted curves (based on the number of unstable directions). Limit points are marked with large dots, and Hopf bifurcations (both supercritical and subcritical) are marked by stars. The evolution of the stable cycles is shown as a shaded area (with the shade matching the respective equilibrium curve); unstable cycles are shown in no color and full transparency. Limit points of cycles are marked by a black bar. For $w_{TR}=0.6$, the orange attracting cycle changes stability at a Neimark--Sacker bifurcation, shown as a green vertical bar. For $w_{TR}=0.8$, the unstable cycle born at the Hopf point becomes stable at a Neimark--Sacker bifurcation, shown as a purple vertical bar. The bottom panel shows (in a three-dimensional $(P,T,R)$ slice) an instance of the stable invariant torus that emerges via this bifurcation, with the color evolving with time as shown in the attached gradient bar.}}
\label{wRT_biff_diff_wTR}
\end{figure}

The fine-control pair $(w_{TR},w_{RT})$ acts therefore as a thalamo-reticular ``core:'' once the cortical background has been appropriately tuned, these two parameters can determine how the oscillations are configured, sustained, and terminated. Cortical excitation still matters, but no longer as the main player. Once cortical parameters determine accessibility to a regime, thalamo-reticular couplings determine its detailed expression. This motivated us to center the remainder of the analysis (in the no-delay system, as well as when considering distributed delays) on the relay-TRN coupling pair for fixed baseline values of $w_{PP}$, $w_{TP}$ and $w_{RR}$, rather than on the full parameter space. To reinforce the interplay between $w_{TR}$ and $w_{RT}$ in controlling oscillatory rhythms in the CTRC circuit, we provide a complementary perspective in Figure~\ref{wRT_biff_diff_wTR} by showing the dependence on $w_{RT}$ for fixed values of $w_{TR}$. 
For lower values of $w_{TR}$  (e.g., $w_{TR}=0.6$, orange diagram), the system can only be pushed into stable oscillations by values of $w_{RT}$ which are too high for the biological range, and these oscillations may never reach the amplitudes required for functional relevance. For values of $w_{TR}$ on the high end of the table range (e.g., $w_{TR}=1$, brown diagram) the system cannot enter oscillations at all (the subcritical Hopf bifurcation renders unstable cycles only). Middle range values of $w_{TR}$ open up the possibility of cycling for an adequately tuned window of $w_{RT}$. The blue diagram illustrates this situation for $w_{TR}=0.8$. The Hopf point at $w_{RT} \sim 0.15$ gives birth to an unstable cycle, which later gains stability via a Neimark--Sacker bifurcation at $w_{RT} \sim 0.18$ (shown as a purple vertical bar) before it ends at a LPC bifurcation at $w_{RT}\sim 0.35$. Notice that for this setup, the system has access not only to stable, spindle-like cycles (cyan shaded region), but also to aperiodic oscillations, for values of $w_{RT}$ slightly lower than the NS bifurcation (the insert shows the invariant torus obtained for $w_{RT}=0.17$).

This emphasizes the idea that this fine control is not provided by either coupling in isolation, but by the functional balance within the coupling scheme: relay activity must recruit TRN strongly enough to sustain the oscillation in the right regime, while TRN inhibition must return strongly enough to generate rebound structure without collapsing the dynamics into quiescence. This observation is especially important for the delay analysis in the next section, because delays are expected to have their strongest effect precisely where excitation and inhibition already form a tightly balanced cyclic loop. Figure~\ref{wRT_biff_diff_wTR} therefore motivates the study of introducing distributed delays in the system -- not only physiologically, but also dynamically. Throughout the remainder of the analysis, we focus on the relay–TRN pair, where the main structurally meaningful organization of the oscillatory dynamics is concentrated.

\subsection{Attractors and transitions for  weak Gamma distributed delays}

Having established in Section~\ref{subsec:linearization} that the delay does not
change the equilibrium locations and enters the dimensionless characteristic
equation through the relative delay \(q=\rho/\tau\), we now investigate
numerically how weak Gamma temporal integration changes the stability and
oscillatory dynamics of the system. Since \(\tau=1\) is fixed throughout the
numerical analysis, the numerical values of \(q\) and \(\rho\) coincide. This analysis was also performed in MatCont 7p4~\cite{Dhooge2008MatCont}.

We first focus on the dependence of the dynamics on the thalamo-reticular couplings \(w_{TR}\) and \(w_{RT}\), considering several representative values of the mean delay \(\rho\). To complement the bifurcation analysis with direct numerical illustrations, we explore how the system's dynamic transitions vary across these coupling parameters and over a range of values of \(\rho\). Compared with the delay-free system, \(\rho\) provides an additional control mechanism that shifts the coupling combinations permitting oscillations, influences their onset, termination, and stability, and modifies their amplitude, period, and temporal profile. This allows us to examine the interplay between connectivity strength and delay in shaping temporal rhythms in the system.

\begin{figure}[h!]
\begin{center}
\includegraphics[width=0.5\textwidth]{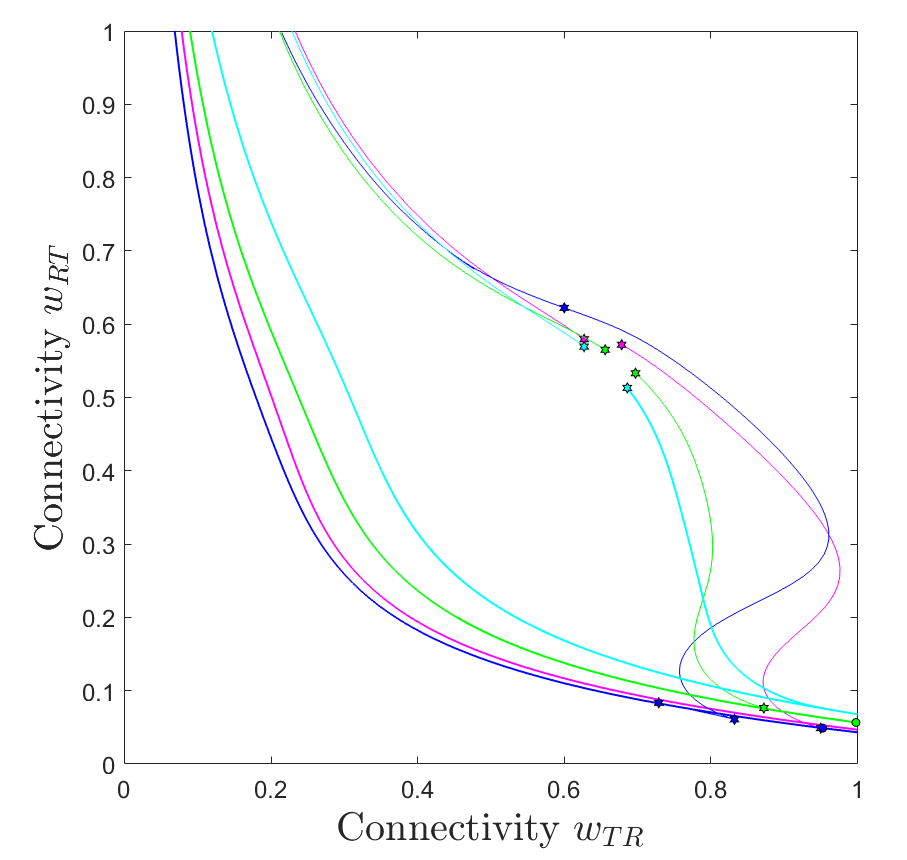}
\end{center}
\caption{\small \emph{{\bf Onset and termination of stable cycles shown in the $(w_{TR},w_{RT})$ parameter plane, for different values of the delay $\rho$}, as follows: $\rho=0.5$ (magenta curves); $\rho=1$ (blue); $\rho=3$ (green); $\rho=5$ (cyan). For each value of $\rho$, the supercritical Hopf curve marking the onset of stable oscillations is shown as a thick curve, and the Limit Point Cycle or Period Doubling curve marking the end of the oscillatory window is shown as a thin curve in the same color. Codimension two bifurcations are marked along these curves as dots in the same color as the corresponding curve.}}
\label{wTR_wRT}
\end{figure}

To begin with, Figure~\ref{wTR_wRT} illustrates the overall deformation of the oscillatory landscape in the $(w_{TR},w_{RT})$ parameter plane -- as the average delay $\rho$ is increased by one order of magnitude within its biological range, from short delays ($\rho=0.5$) to long delays ($\rho=5$). The broad picture shows qualitatively different behavior across connectivity profiles. For high $w_{RT}$ (above $w_{RT} \sim 0.6$), the oscillatory window lies in a lower $w_{TR}$ range, and slightly (yet consistently) shrinks as the delay $\rho$ is increased. For low $w_{RT}$ (below $w_{RT} \sim 0.6$), the oscillatory window shifts overall to higher $w_{TR}$, but the termination curves are no longer monotonically dependent on the delay, and have much more complex dependence on both connectivity and $\rho$ (with the plots visibly crossing and swapping each other, and showing codimension two bifurcation points). To clarify and better understand the underpinnings of these behaviors, we constructed one-parameter cross-sections of this figure, by fixing $w_{TR}$ and $w_{RT}$ to two respective values (one in the respective high range and one in the low range), and plotting the corresponding bifurcation diagram with respect to the remaining parameter, for all four chosen values of $\rho$.

 Figure~\ref{wTR_T_R} first shows this deformation along the $w_{TR}$ direction, for two fixed levels of relay-to-TRN excitation $w_{RT}$. For $w_{RT}=0.4$ (bottom row panels), increasing the average delay produces a regular, monotone contraction of the oscillatory window in terms of $w_{TR}$. As $\rho$ increases through the four values considered, the Hopf onset shifts progressively to the right, from approximately $w_{TR}=0.22$ for $\rho=0.5$ to $w_{TR}=0.27$, $0.31$, and $0.36$ for $\rho=1$, $3$, and $5$, respectively. At the same time, the terminal LPC moves steadily to the left, from approximately $w_{TR}=0.91$ to $0.88$, $0.79$, and $0.75$. Thus, for this higher value of $w_{RT}$, longer delays both postpone the onset of oscillations and advance their termination, leading to a consistent narrowing of the oscillatory window.

In the regime of lower $w_{RT}=0.2$ (top row panels), the dependence on $\rho$ of the upper end of the oscillatory window is no longer monotone, consistently with Figure~\ref{wTR_wRT}. At $\rho=0.5$, the oscillatory window begins near $w_{TR}=0.54$ and terminates near $w_{TR}=0.94$. When the delay is increased to $\rho=1$, the Hopf onset moves sharply to the left, to approximately $w_{TR}=0.37$, while the terminal bifurcation shifts to about $w_{TR}=0.82$, so that the oscillatory window initially broadens rather than contracts. For larger delays, the trend reverses: the Hopf onset moves back to the right, while the termination point continues to move left, respectively.

The temporal panels underneath the bifurcation landscapes clarify that the dynamical consequences of the interplay between the delay and the connectivity parameters go beyond simply shifting the oscillation window. For $\rho=0.1$ and $\rho=0.5$, the sampled parameter values lie close to the oscillation onset, and support more rapid recurrent oscillations ($\sim 0.1-0.2$Hz). By $\rho=2$, the sample point is closer to the right side of the oscillatory window, and the surviving cycles are substantially slower ($\sim 0.06$Hz). Higher values of $w_{TR}$ successively cross into the stable-equilibrium regime. For example, at $\rho=5$, the point $w_{TR}=0.5$ falls below the oscillation offset for $w_{RT}=0.4$, but is outside of the oscillatory window for $w_{RT}=0.2$. In turn, $w_{TR}=0.8$ is beyond the limit point cycle in both cases. This provides a compelling illustration of the coupling-delay interaction: for the same delay and the same value of $w_{TR}$, the rhythm may either be suppressed or persist as a very slow oscillation, depending only on the strength of the excitatory return from the relay population to the TRN. It also shows that the duty cycle changes significantly across the oscillation window, an important idea which we will revisit separately later.

\clearpage
\begin{landscape}
\begin{figure}[h!]
\begin{center}
\includegraphics[width=0.24\linewidth]{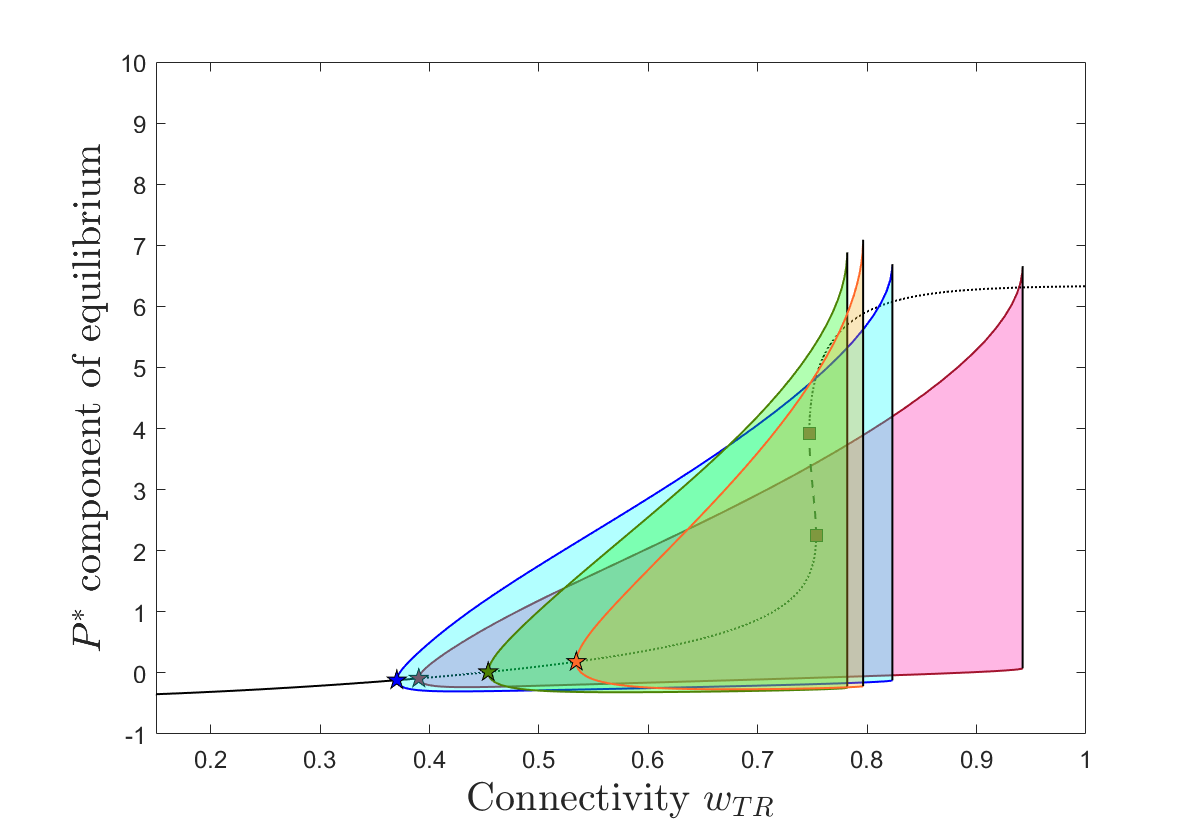}
\includegraphics[width=0.24\linewidth]{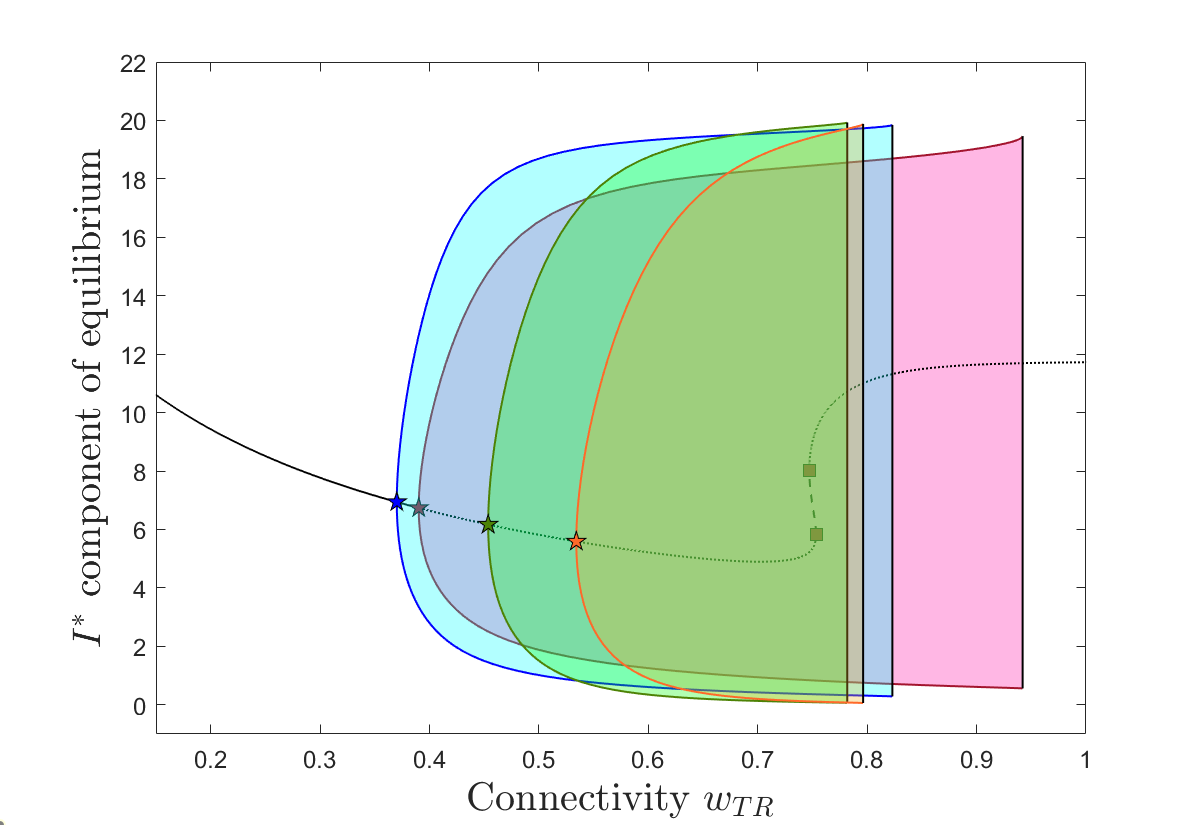}
\includegraphics[width=0.24\linewidth]{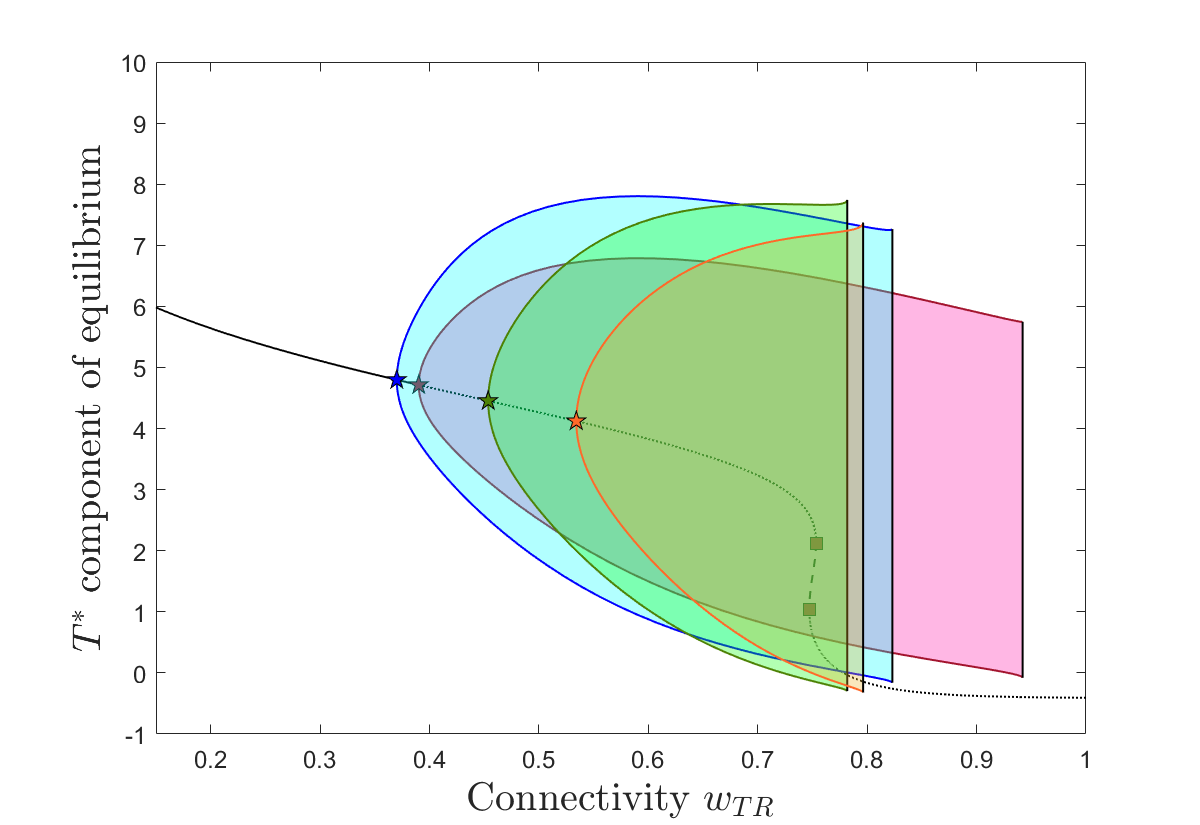}
\includegraphics[width=0.24\linewidth]{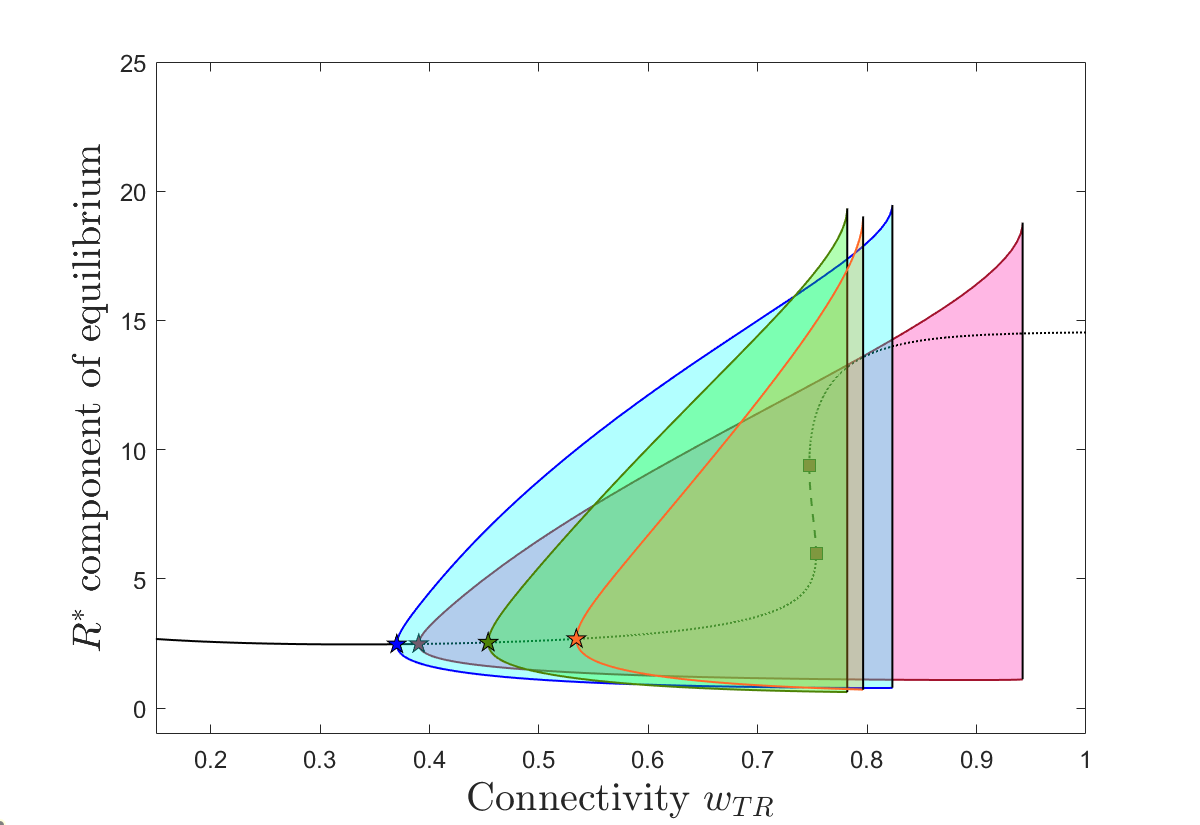}

\includegraphics[width=0.24\linewidth]{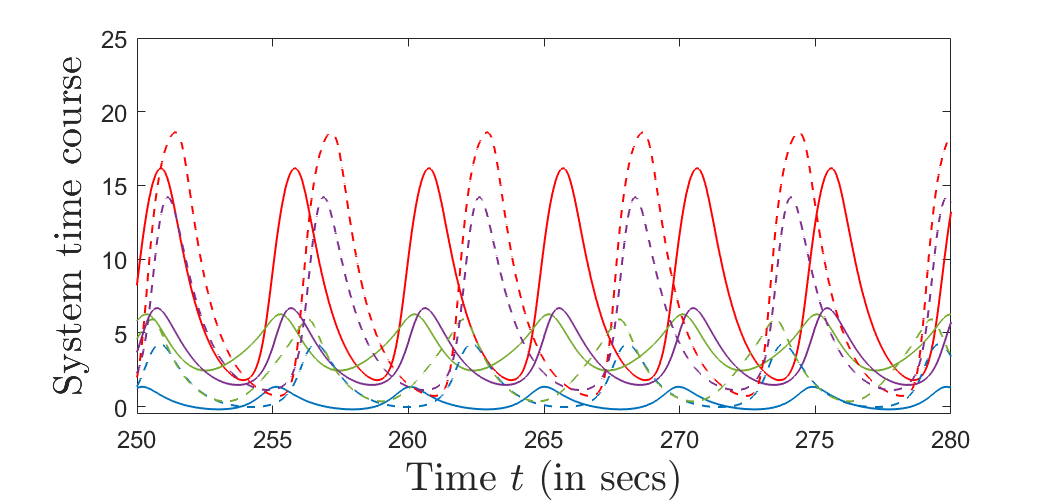}
\includegraphics[width=0.24\linewidth]{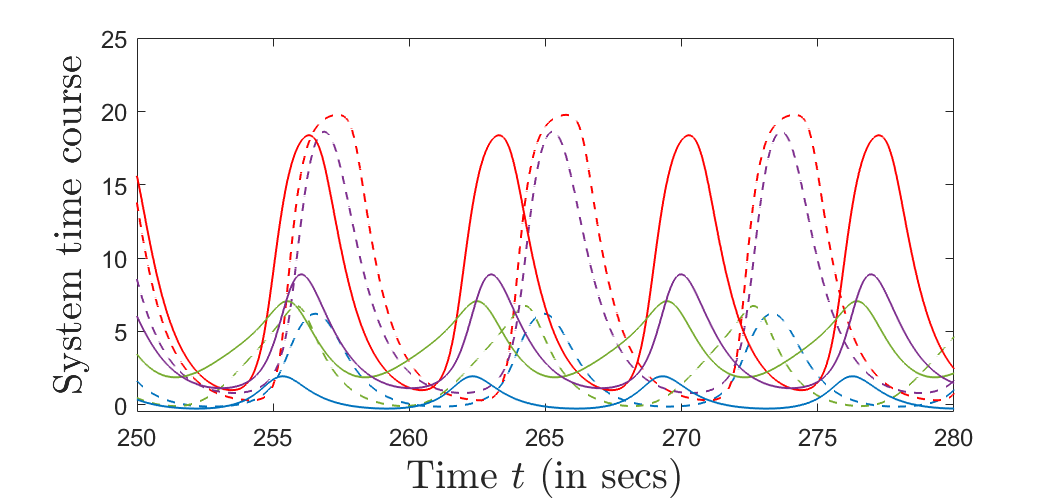}
\includegraphics[width=0.24\linewidth]{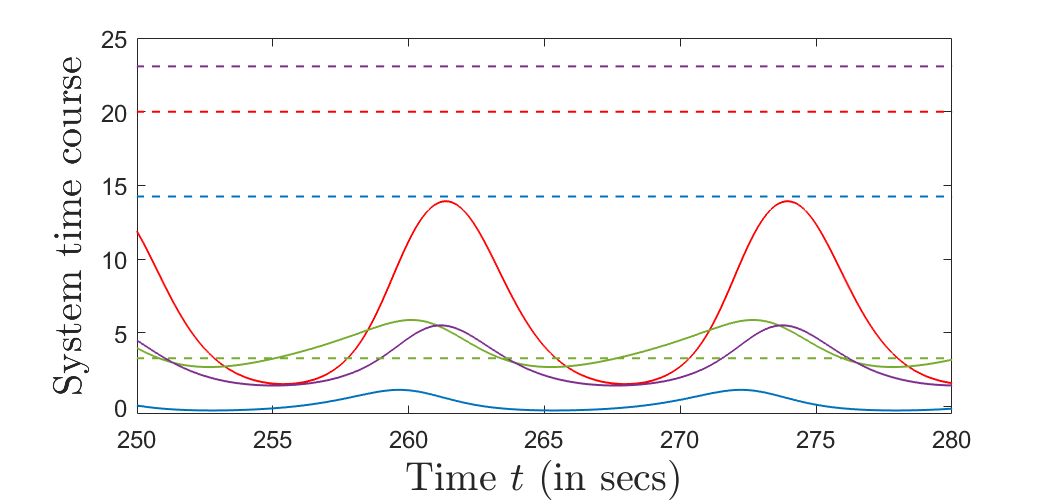}
\includegraphics[width=0.24\linewidth]{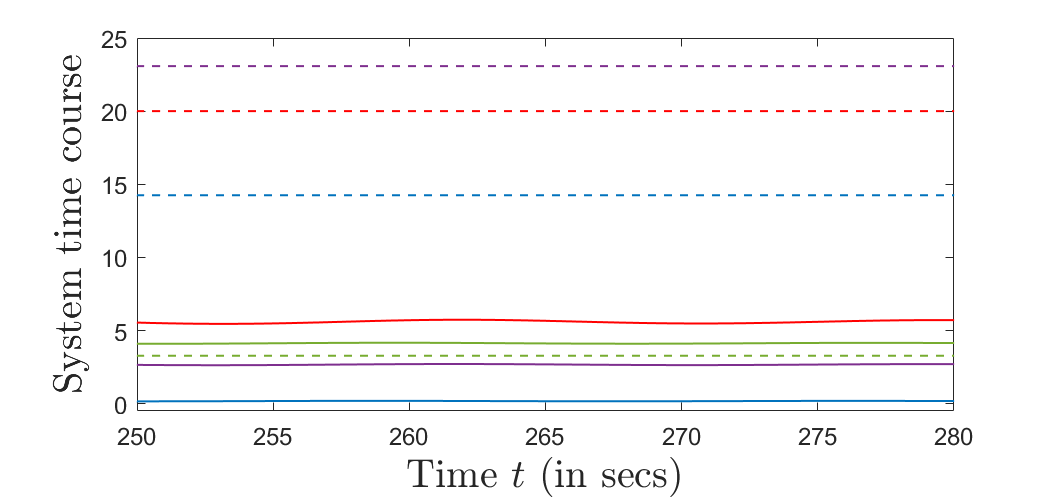}

\vspace{1cm}
\includegraphics[width=0.24\linewidth]{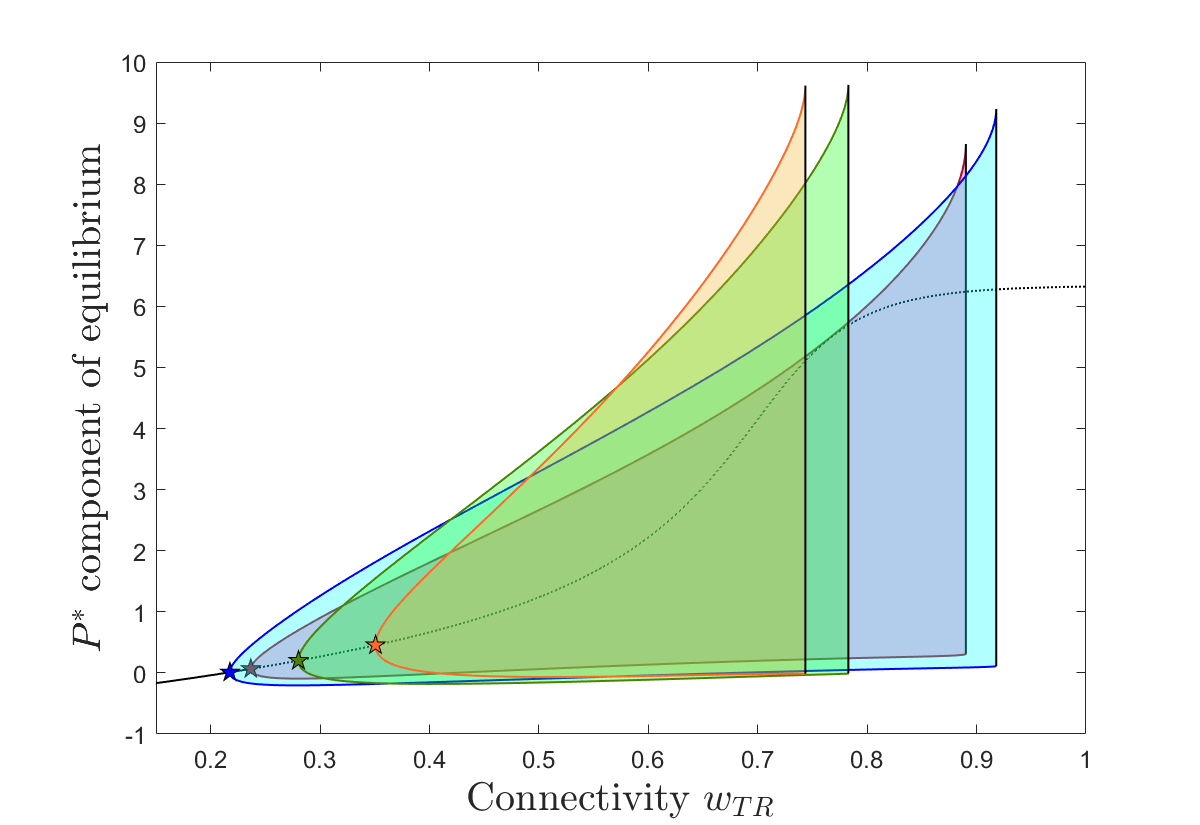}
\includegraphics[width=0.24\linewidth]{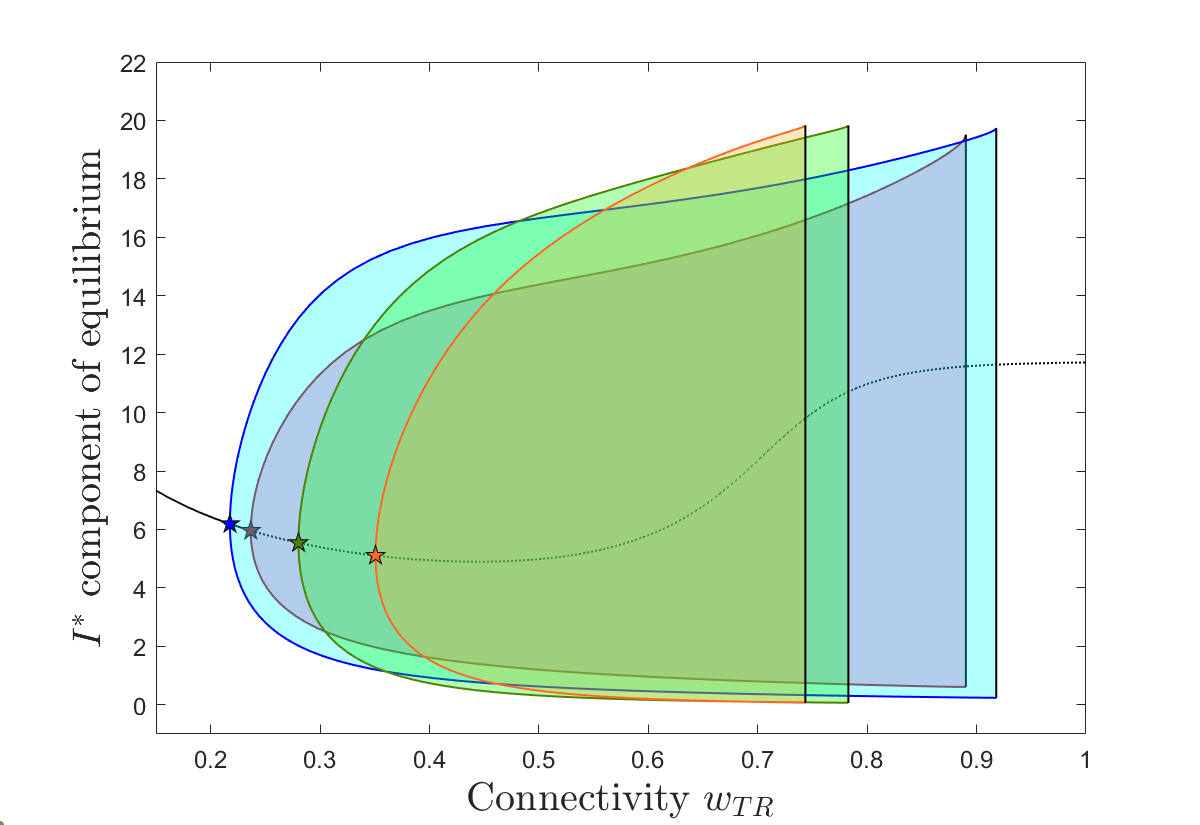}
\includegraphics[width=0.24\linewidth]{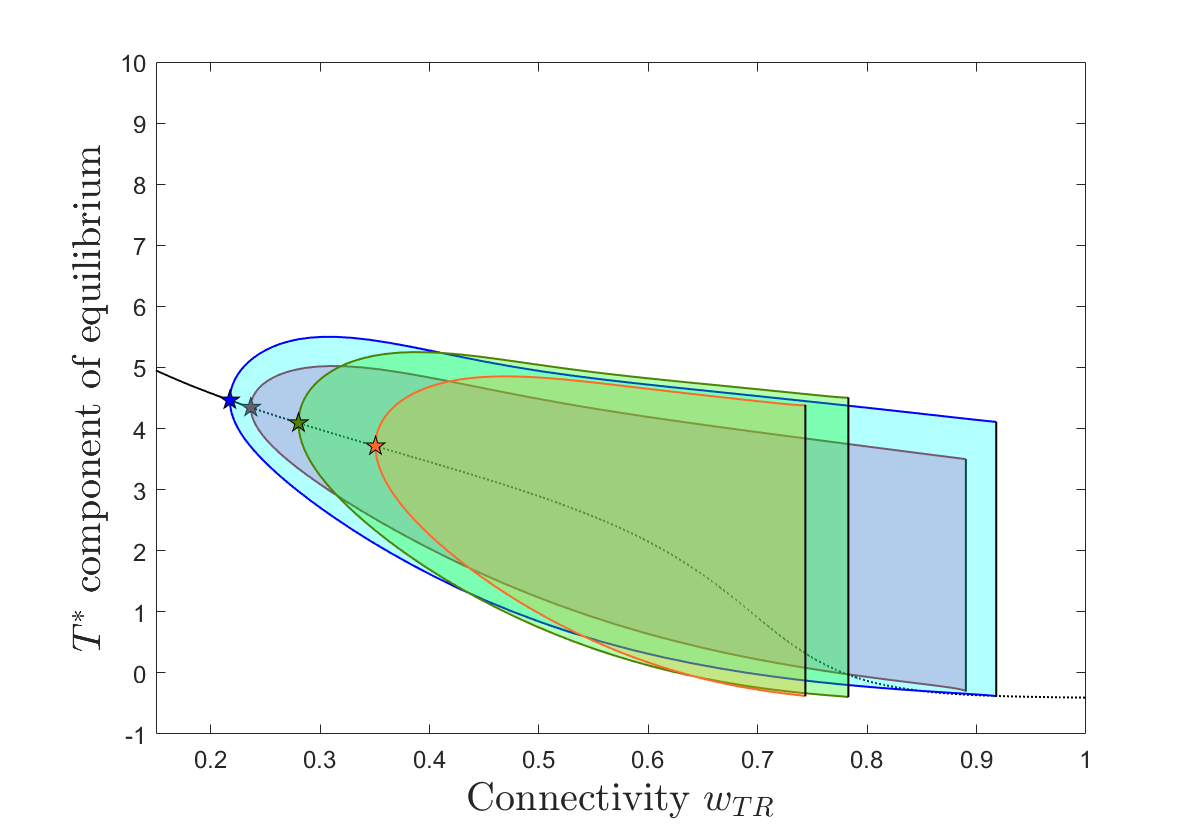}
\includegraphics[width=0.24\linewidth]{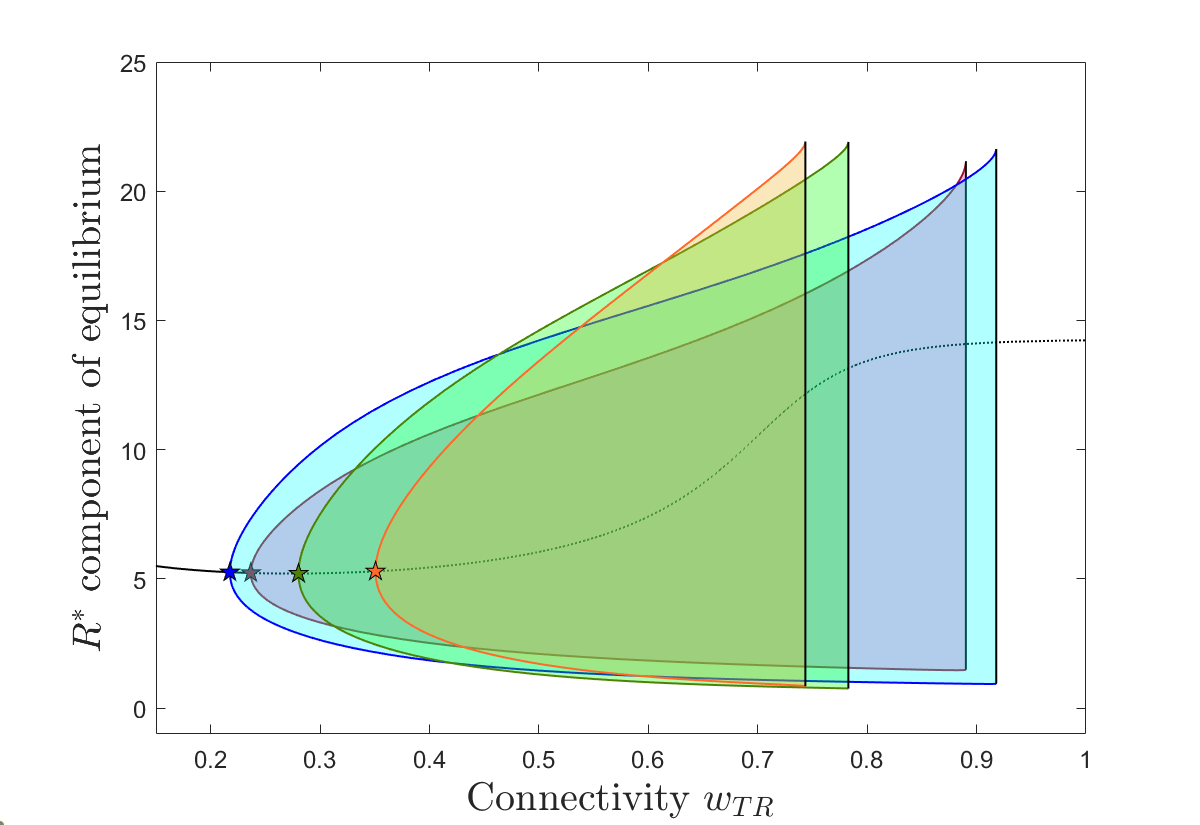}

\includegraphics[width=0.24\linewidth]{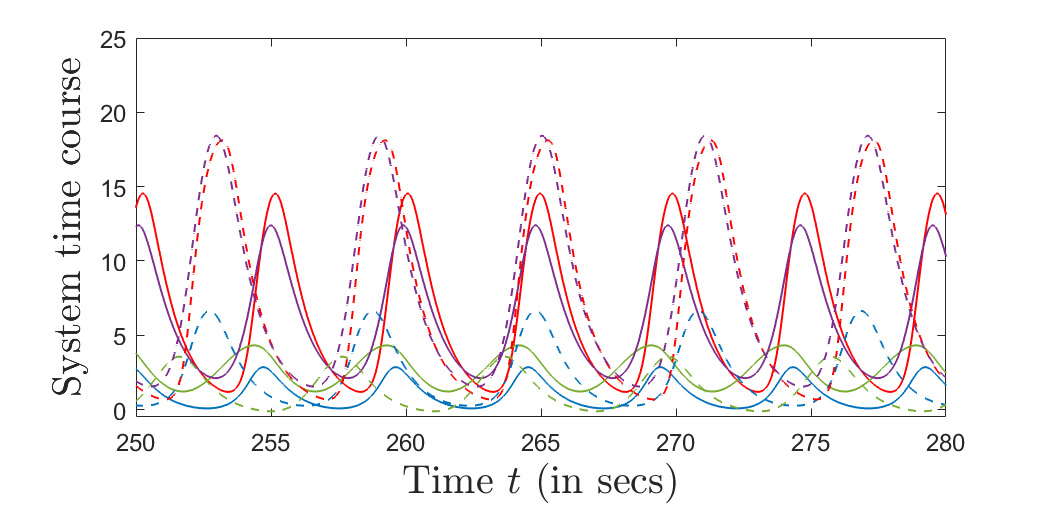}
\includegraphics[width=0.24\linewidth]{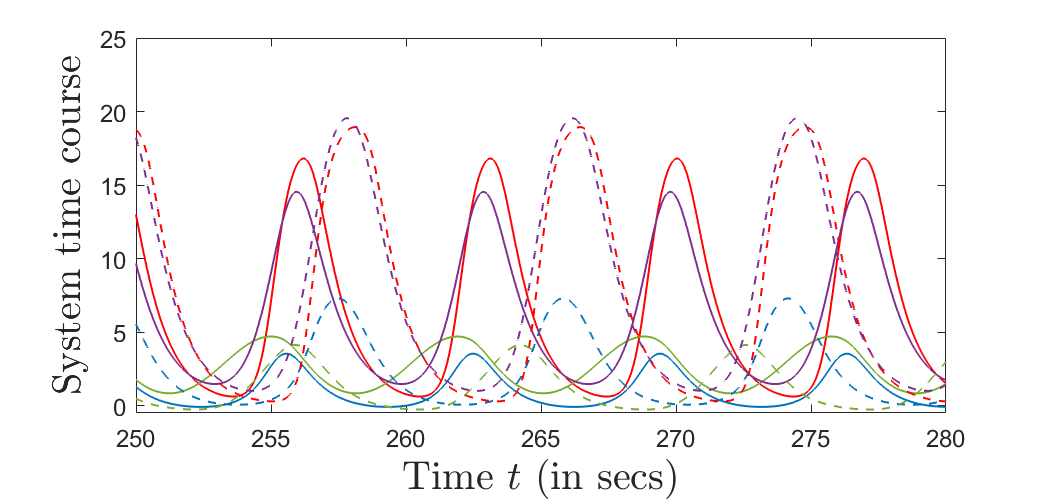}
\includegraphics[width=0.24\linewidth]{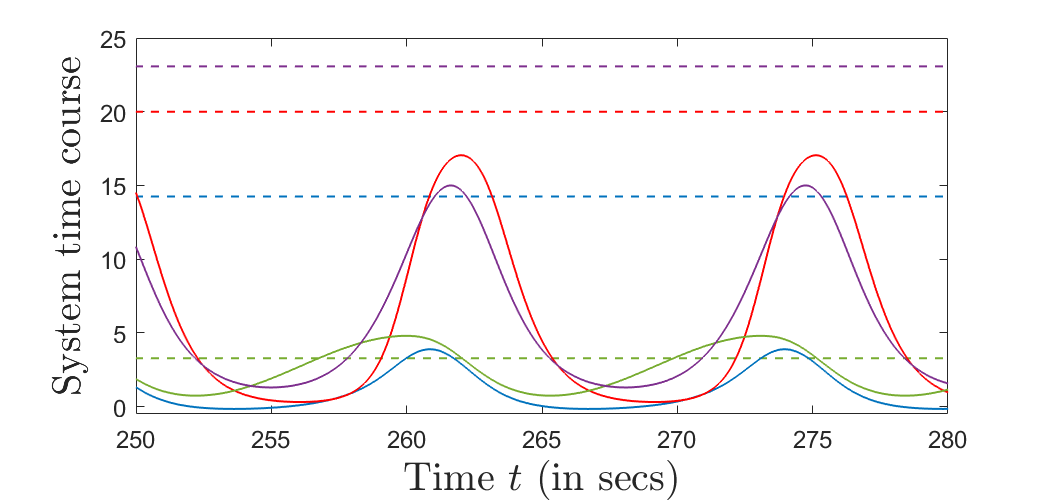}
\includegraphics[width=0.24\linewidth]{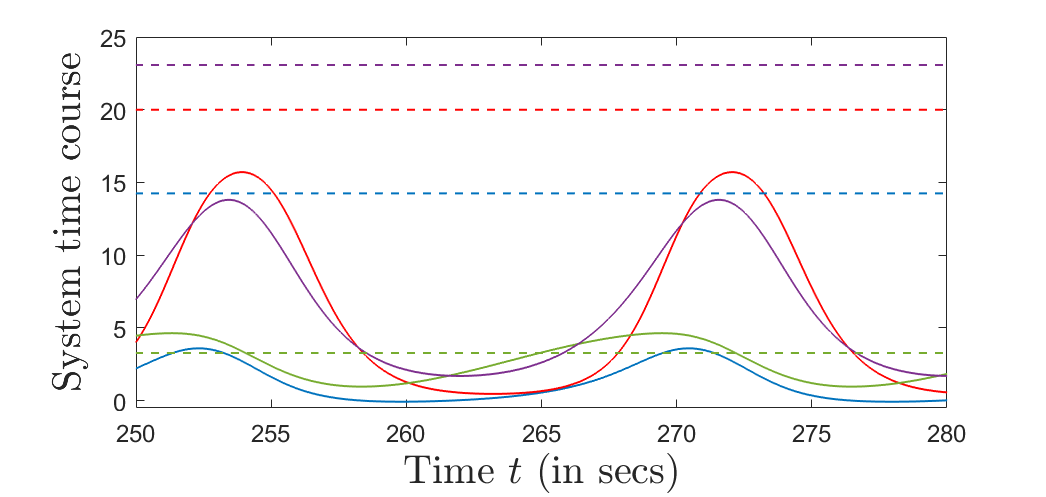}

\end{center}

\caption{\small \emph{{\bf Dependence on $w_{TR}$ for different $w_{RT}$ values and different delays $\rho$.}} {\bf Top:} $w_{RT}=0.2$. {\bf Bottom:} $w_{RT}=0.4$. Different diagrams stand for different average delays: $\rho=0.5$ (pink); $\rho=1$ (blue); $\rho=3$ (green); $\rho=5$ (orange). The temporal panels underneath illustrate a sample solution for the scenario of fixed $w_{RT}=0.2$ (top), and for fixed $w_{RT}=0.4$ (bottom) for $w_{TR}=0.5$ (solid curves) and $w_{TR}=0.8$ (dashed curves), for four different sample delays in each case (left to right): $\rho=0.1$. $\rho=0.5$, $\rho=2$, $\rho=5$. The system components are coded by color as follows: $P$ (blue), $I$ (red); $T$ (green) and $R$ (purple). For these simulations, $w_{PP}=1.1$, $w_{RR}=0.2$ and all other parameters were fixed to their baseline tabulated values.}
\label{wTR_T_R}
\end{figure}
\end{landscape}

\clearpage
\begin{landscape}
\begin{figure}[h!]
\begin{center}
\includegraphics[width=0.24\linewidth]{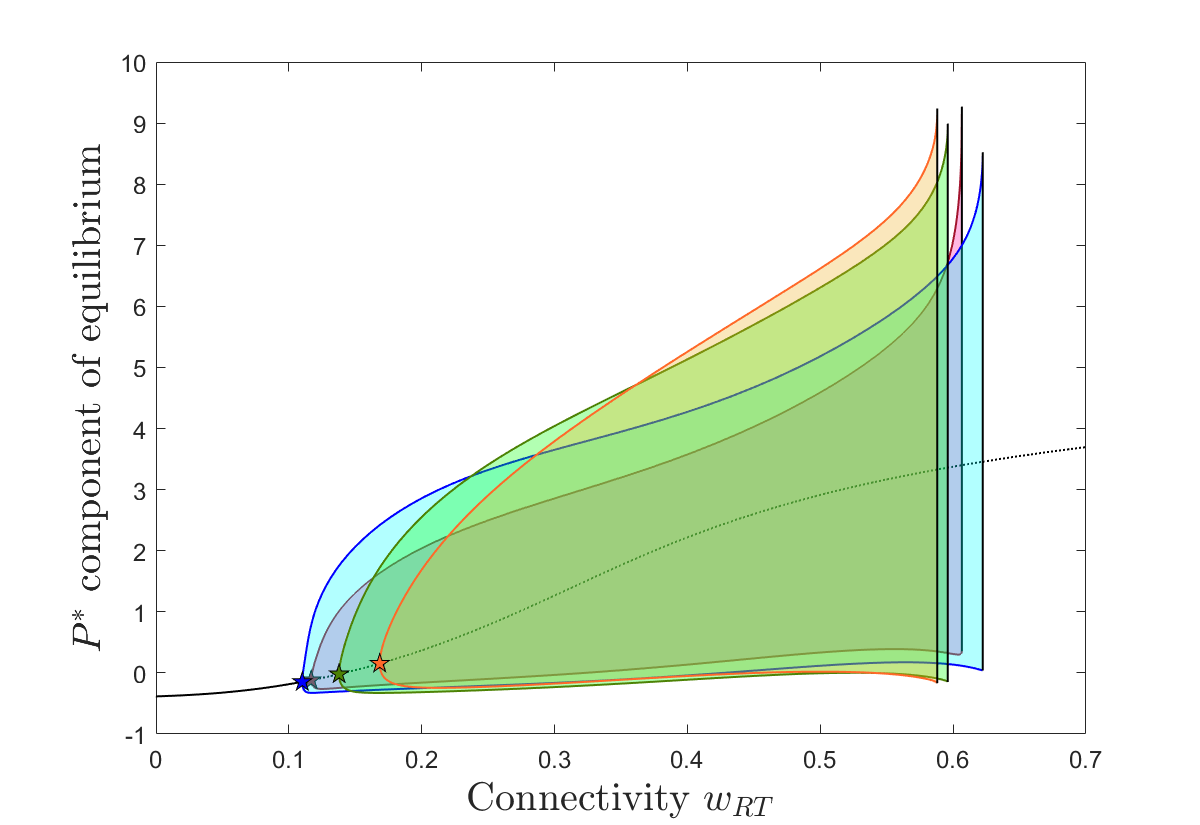}
\includegraphics[width=0.24\linewidth]{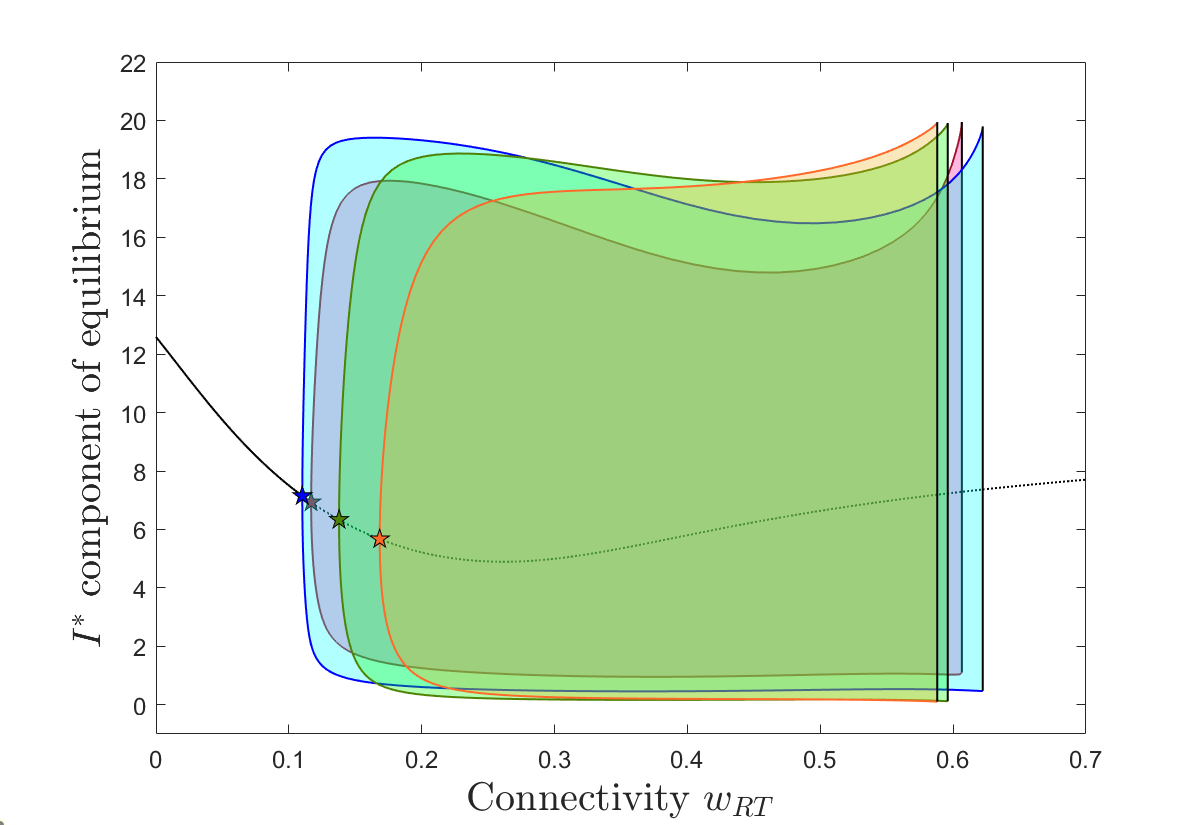}
\includegraphics[width=0.24\linewidth]{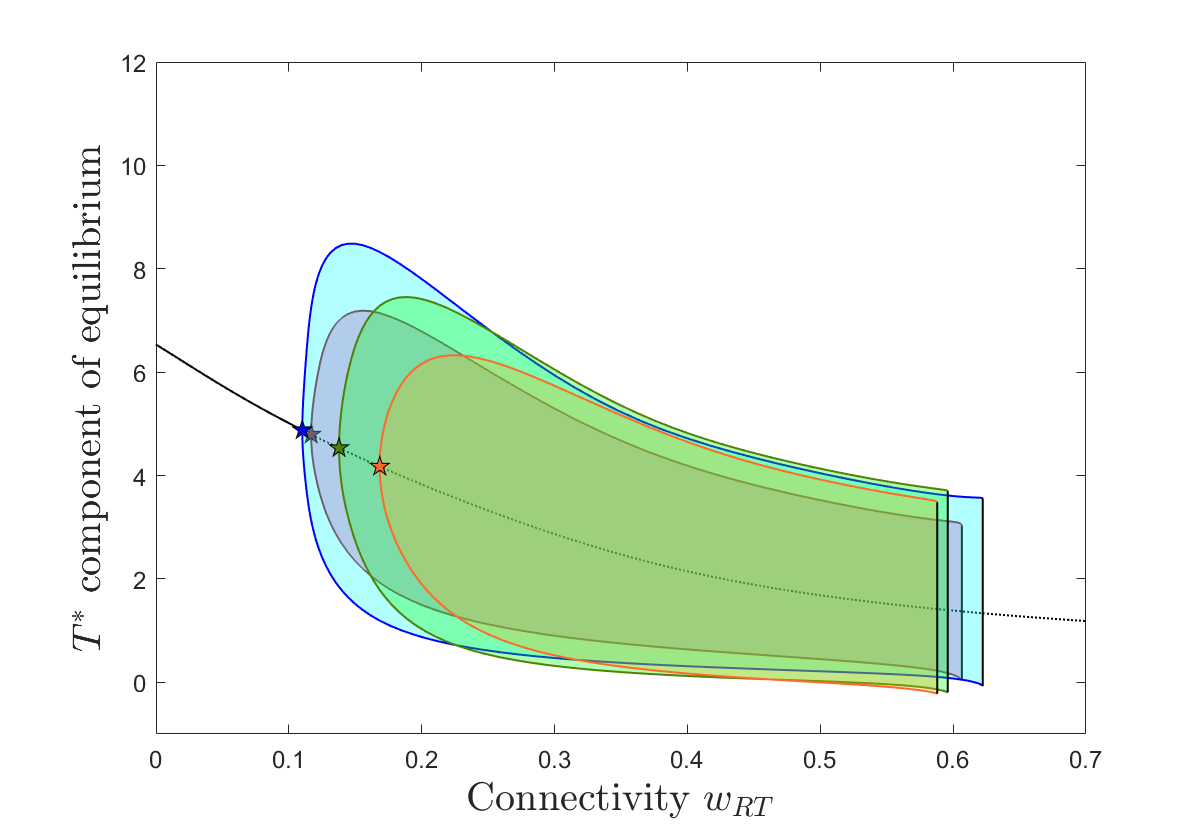}
\includegraphics[width=0.24\linewidth]{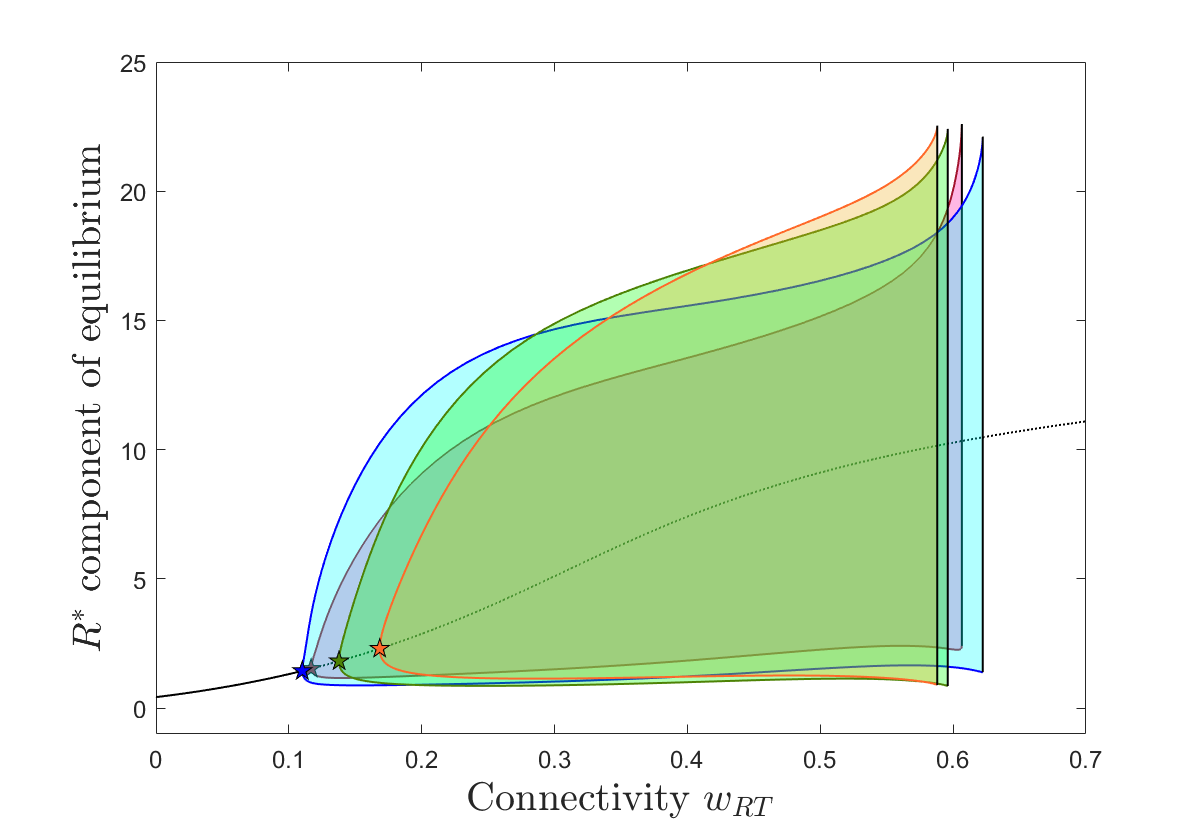}

\includegraphics[width=0.24\linewidth]{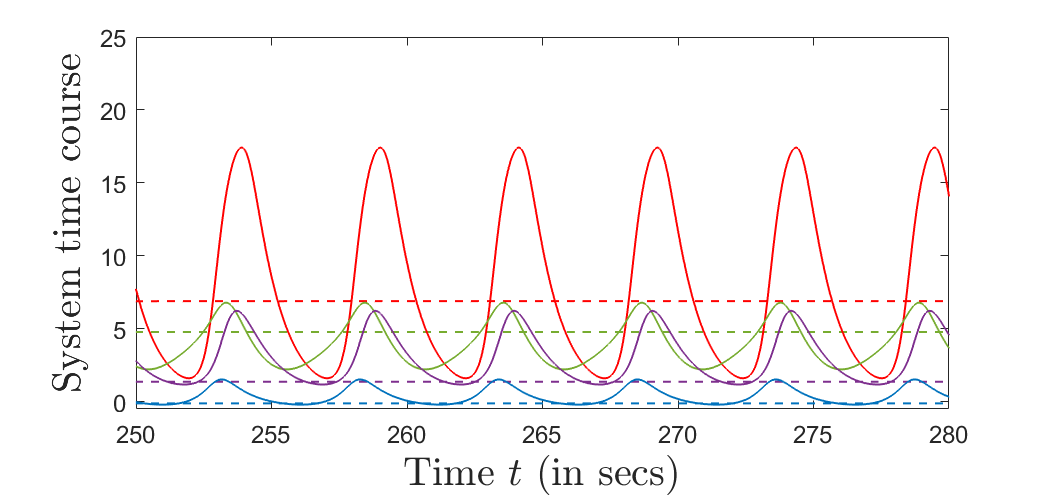}
\includegraphics[width=0.24\linewidth]{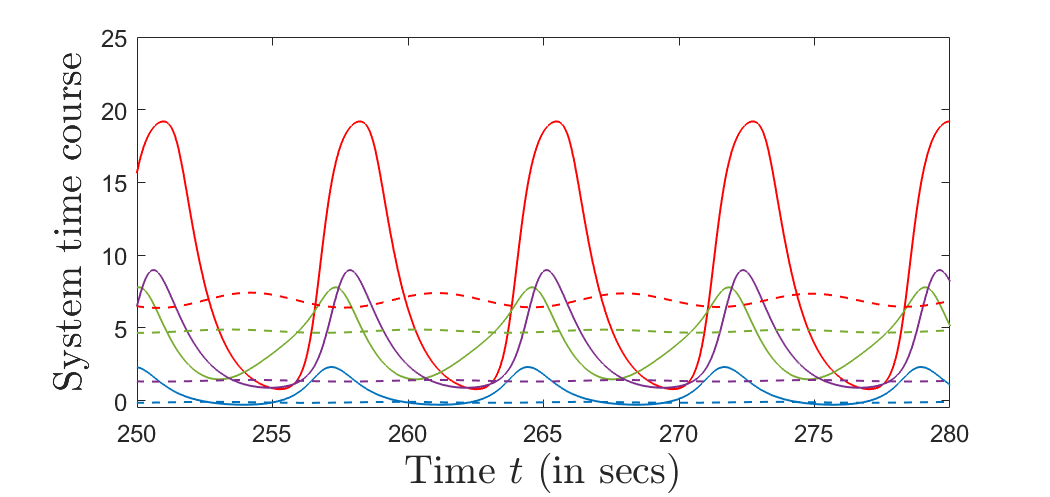}
\includegraphics[width=0.24\linewidth]{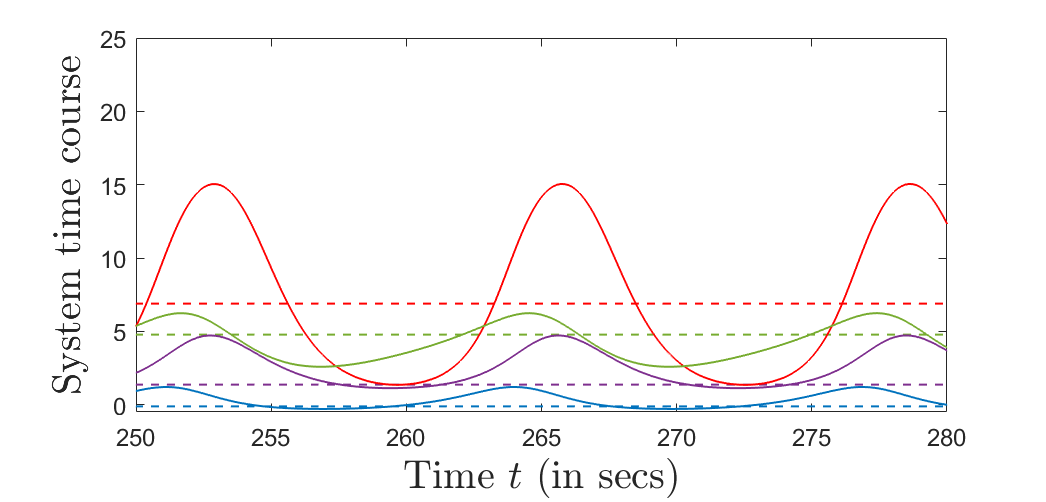}
\includegraphics[width=0.24\linewidth]{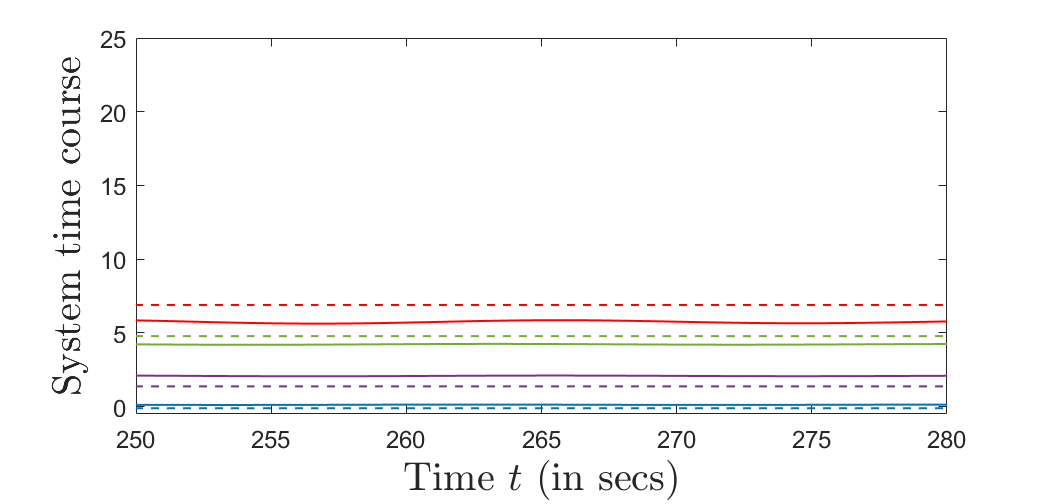}

\vspace{1cm}
\includegraphics[width=0.24\linewidth]{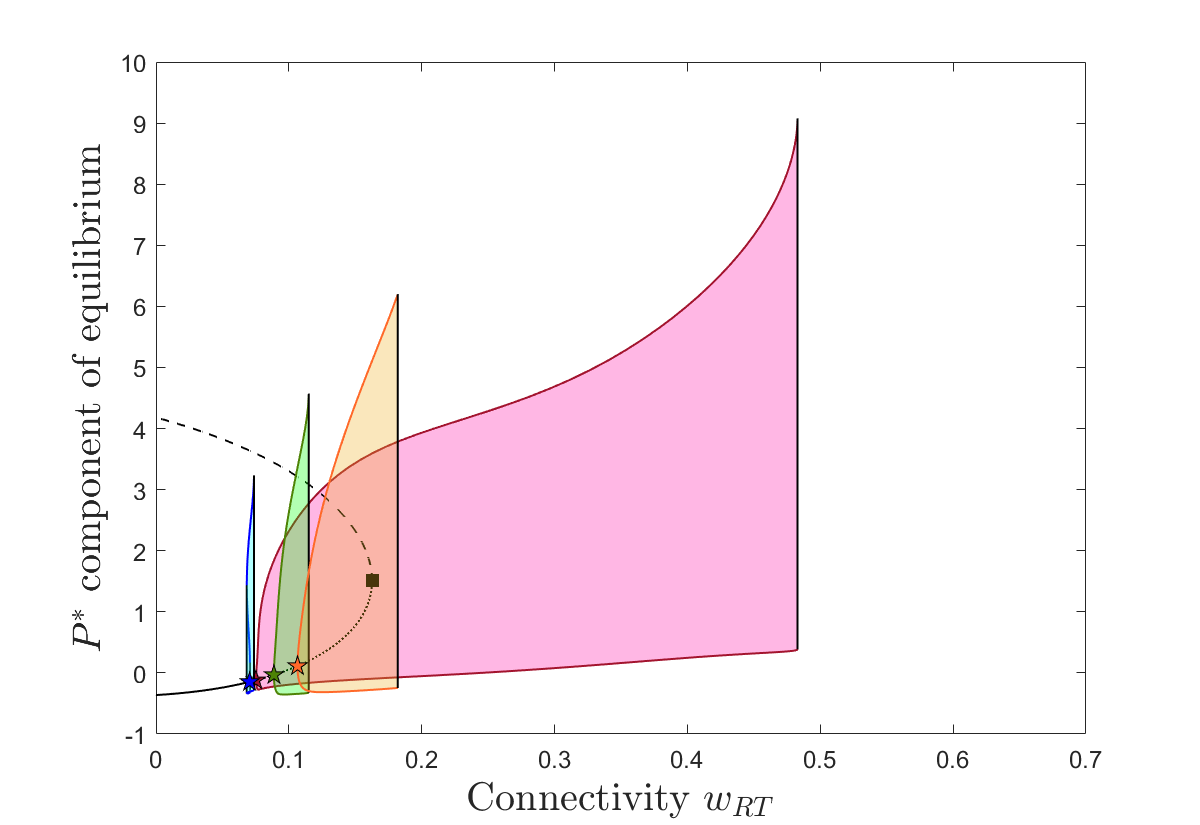}
\includegraphics[width=0.24\linewidth]{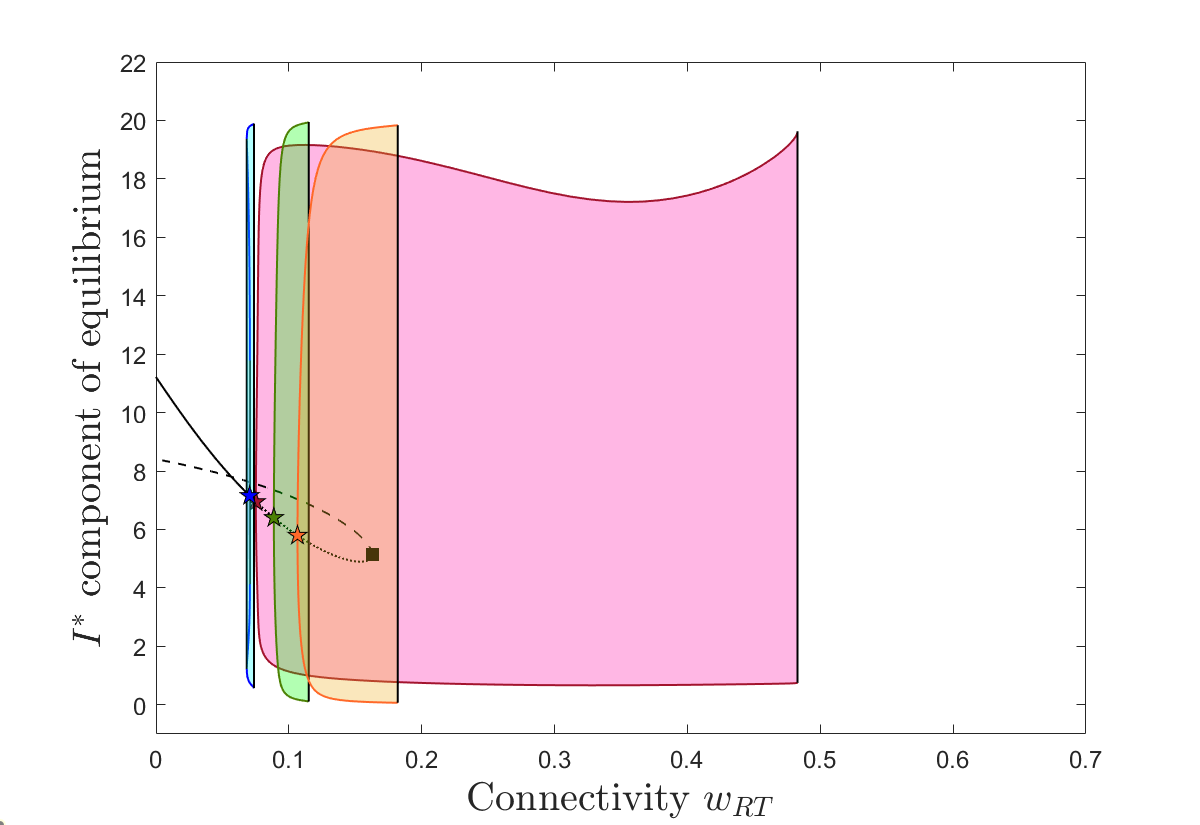}
\includegraphics[width=0.24\linewidth]{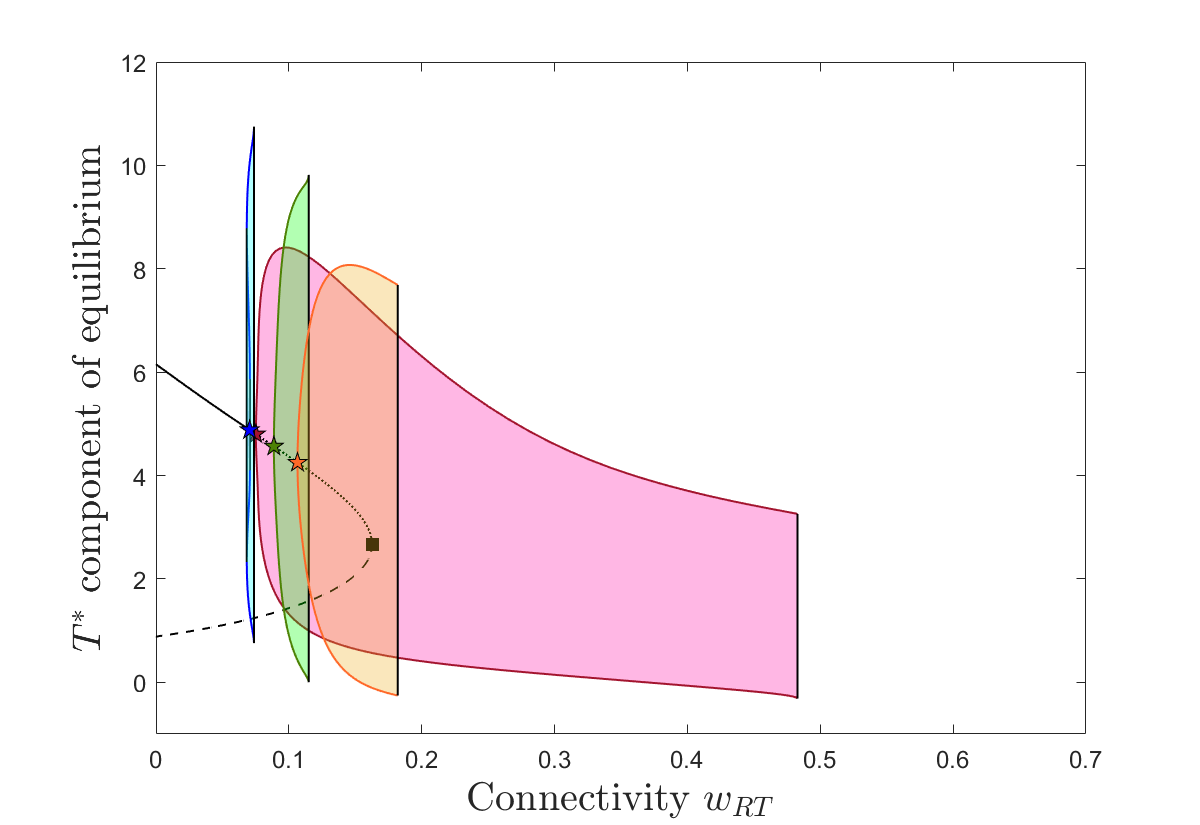}
\includegraphics[width=0.24\linewidth]{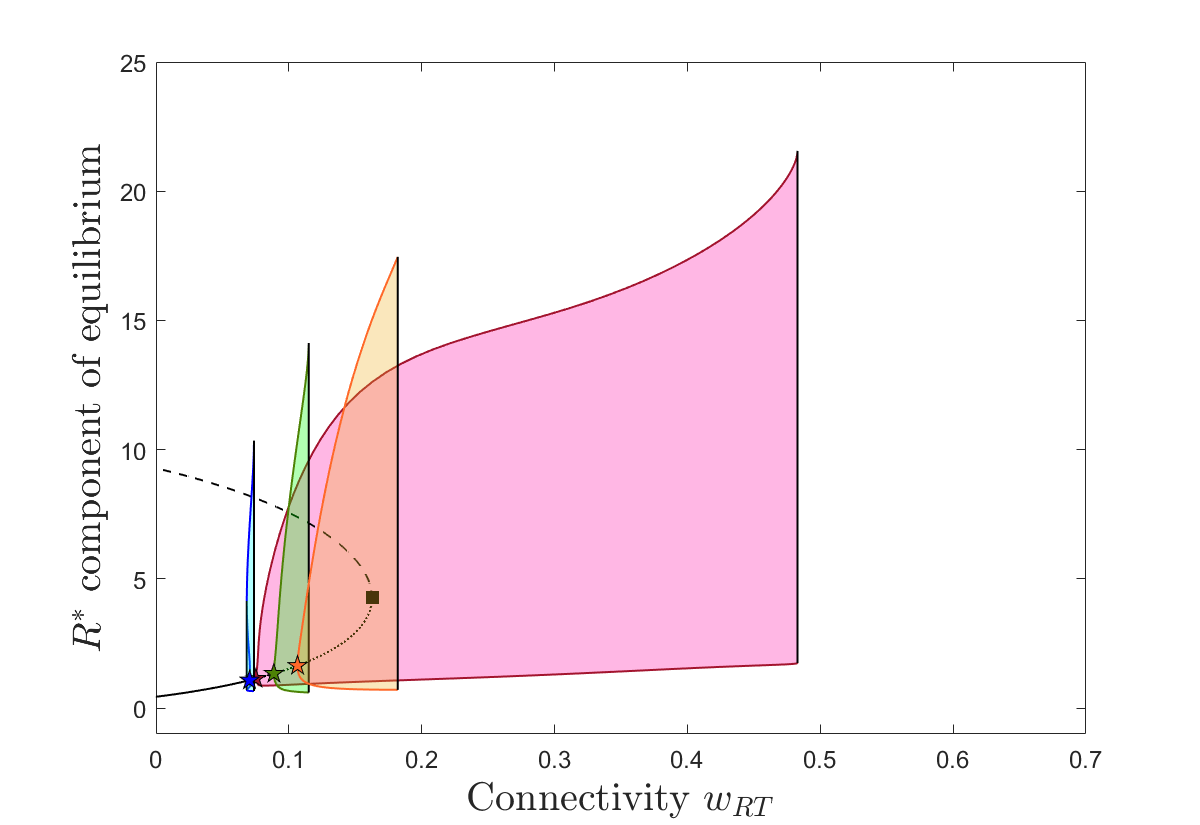}

\includegraphics[width=0.24\linewidth]{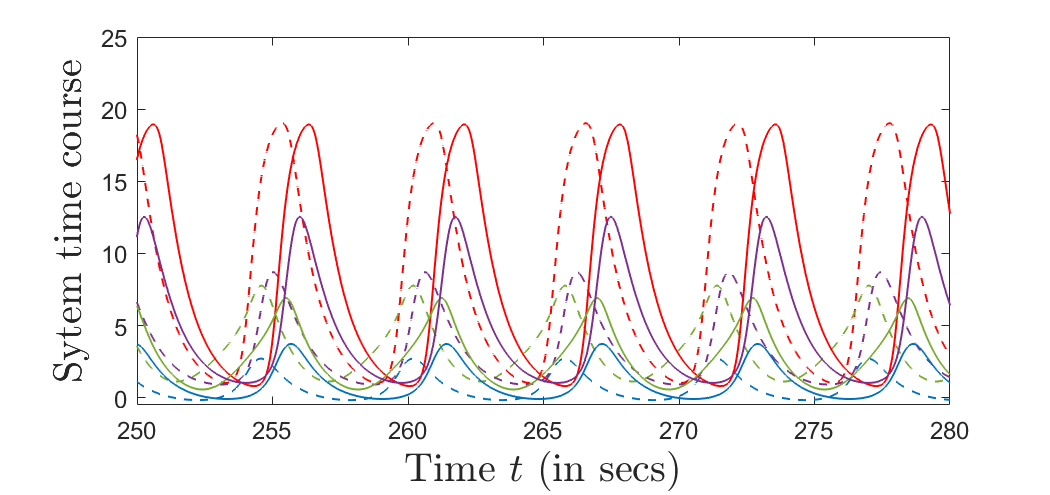}
\includegraphics[width=0.24\linewidth]{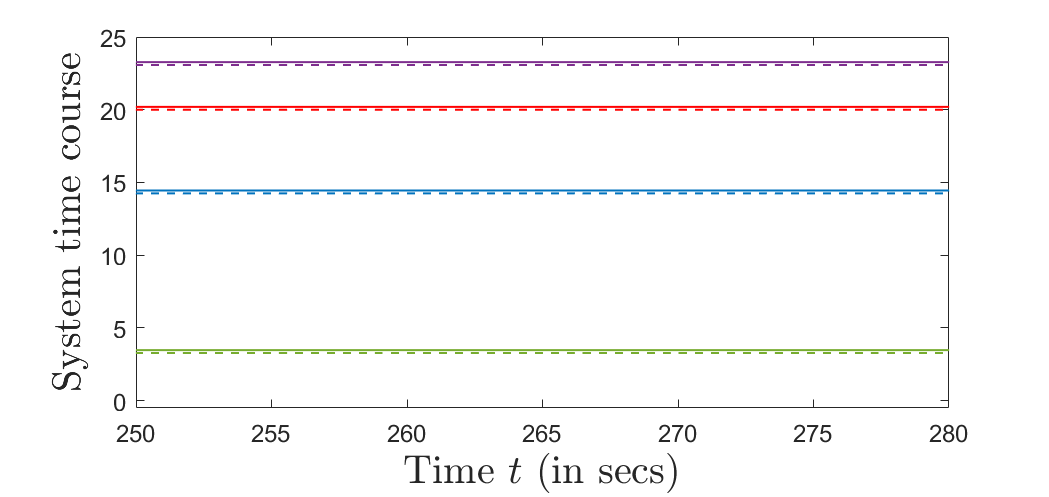}
\includegraphics[width=0.24\linewidth]{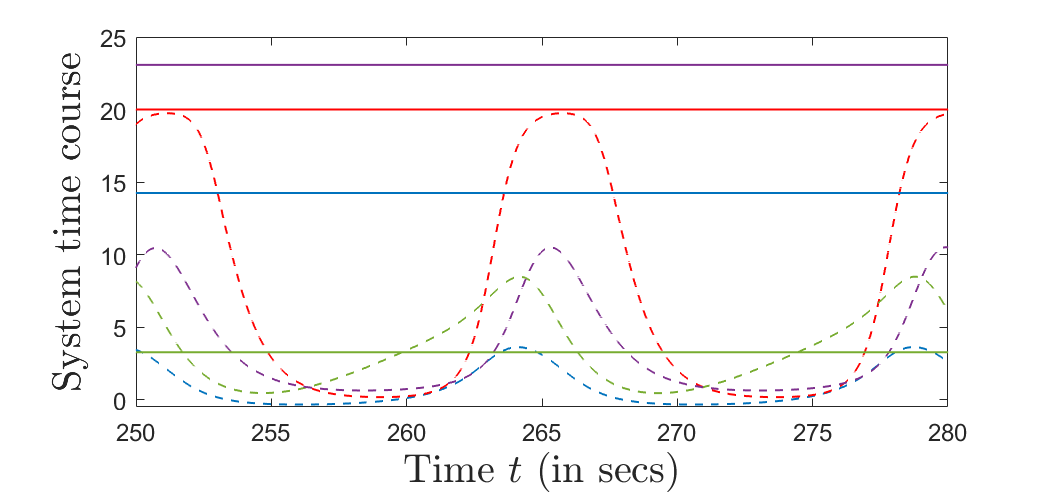}
\includegraphics[width=0.24\linewidth]{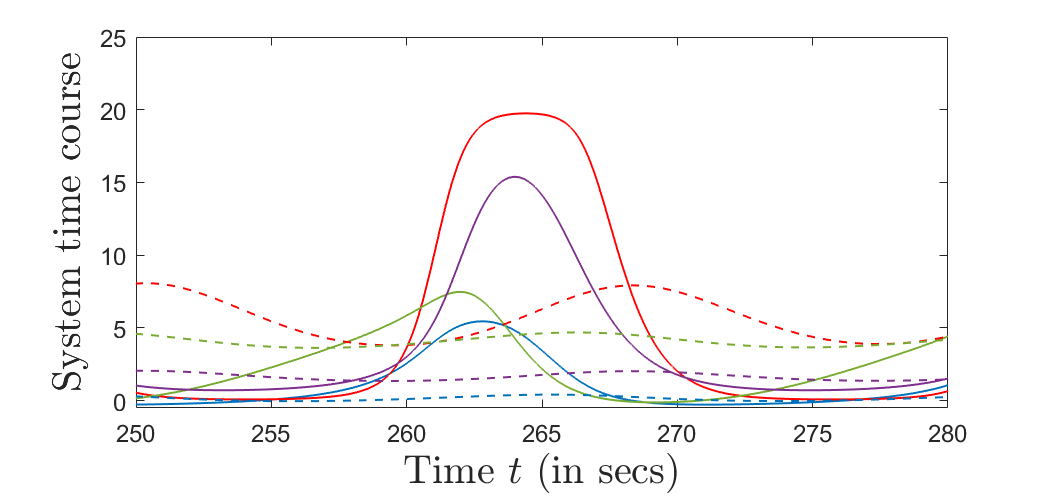}

\end{center}

\caption{\small \emph{{\bf Dependence on $w_{RT}$ for different $w_{TR}$ values and different delays $\rho$.}} {\bf Top:} $w_{TR}=0.6$. {\bf Bottom:} $w_{TR}=0.8$. Different diagrams stand for different average delays: $\rho=0.5$ (pink); $\rho=1$ (blue); $\rho=3$ (green); $\rho=5$ (orange). The temporal panels underneath illustrate a sample solution for $w_{RT}=0.15$ (solid curves) and $w_{RT}=0.1$ (dashed curves), for four different sample delays in each case (left to right): $\rho=0.5$. $\rho=1$, $\rho=3$, $\rho=5$. The system components are coded by color as follows: $P$ (blue), $I$ (red); $T$ (green) and $R$ (purple). For these simulations, $w_{PP}=1.1$, $w_{RR}=0.2$ and all other parameters were fixed to their baseline tabulated values.}
\label{wRT_T_R}
\end{figure}
\end{landscape}

Figure~\ref{wRT_T_R} flips the view, considering bifurcations with respect to $w_{RT}$, for two levels of TRN-to-relay inhibition: a lower value $w_{TR}=0.6$, and a higher value $w_{TR}=0.8$. For $w_{TR}=0.6$, the dependence of the oscillatory window on the average delay is monotone. As $\rho$ increases, the Hopf onset shifts overall toward larger values of $w_{RT}$, and the terminal LPC moves slightly to the left. Thus, for this lower $w_{TR}$, increasing the delay produces an overall contraction of the oscillatory window, primarily by progressively removing its low-$w_{RT}$ portion. This progression changes dramatically at the higher $w_{TR}=0.8$. For $\rho=0.5$, the oscillatory window is broad, extending approximately from $w_{RT}=0.07$ to $w_{RT}=0.49$. When the delay is increased to $\rho=1$, this region collapses to a very narrow interval near $w_{RT}=0.07$--$0.08$. For larger delays, oscillations reappear within windows that shift progressively to the right. The dependence on $\rho$ is therefore strongly nonmonotone in this regime: increasing the delay first almost eliminates stable oscillations and then restores them within a different range of relay-to-TRN coupling strengths. As before, the temporal panels underneath illustrate the consequences of these changes on the duty cycle of the oscillations. In particular, increasing the delay can suppress and subsequently restore rhythmic activity by moving the oscillatory window across the connectivity landscape, while simultaneously stretching the timescale of the surviving cycles.

Having established that $\rho$ contributes to reshaping the oscillatory boundary and rhythms across the $(w_{TR},w_{RT})$ connectivity landscape in nontrivial ways, we next want to focus in more depth on analyzing the implications of increasing the delay for a fixed connectivity profile. This will help us better understand to what extent the dynamics of a system can be reshaped by simply increasing the delay in responses, without changing the architecture or connectivity strengths. This analysis aims to look beyond transitions in and out of cycling behavior, but also consider the biological plausibility of the oscillation, and its physiological meaning -- through interpreting its baseline, amplitude, frequency and duty cycle.

\begin{figure}[h!]
\begin{center}
\includegraphics[width=\textwidth]{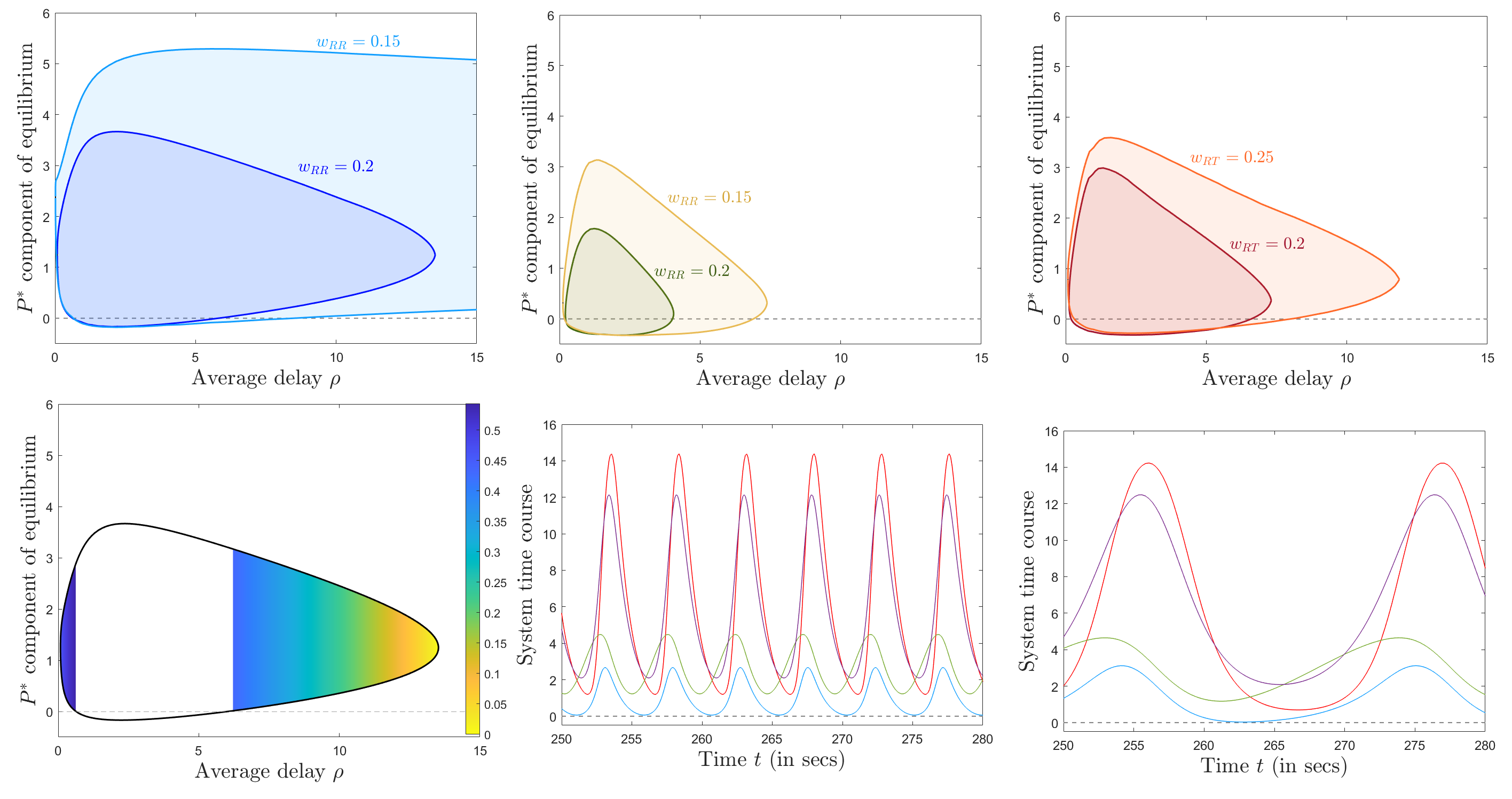}
\end{center}
\caption{\small \emph{{\bf Evolution of oscillations with increasing the average weak Gamma delay $\rho$.} On the top row, the evolution of the $P$ component of the stable cycle is illustrated as $\rho$ increases,  for different connectivity landscapes. {\bf A.} $w_{TR}=0.5$,
$w_{RT}=0.4$ and two different values of $w_{RR}$: $w_{RR}=0.15$ (cyan) and $w_{RR}=0.2$ (deep blue); {\bf B.}  $w_{TR}=0.5$,
$w_{RT}=0.2$ and two different values of $w_{RR}$: $w_{RR}=0.15$ (yellow) and $w_{RR}=0.2$ (green);{\bf C.}  $w_{TR}=0.6$,
$w_{RR}=0.2$ and two different values of $w_{RT}$: $w_{RT}=0.25$ (orange) and $w_{RT}=0.2$ (brown). All other parameters were set to their baseline values in Table~\ref{tab:noncoupling_refs} and~\ref{tab:coupling_refs}. A dotted gray line marks the boundary $P=0$ of the biological range of oscillations. The bottom panels detail the behavior of the blue cycles in the first panel (for $w{TR}=0.5$, $w_{RT}=0.4$, $w_{RR}=0.2$. The color of the shading encodes the rhythm and physiological meaning of the oscillations from deep blue (spindle range) to deep yellow (infra-slow range), with the transient showing the intermediate rhythms between them. The cycle is shown in white when it is outside of the biological range ($P$ enters the negative domain, force-stopping oscillations). The two subsequent panels show actual solutions of the system for one delay value sampled in the spindle range ($\rho=0.5$) and one sampled in the infra-slow range ($\rho=6.5$). All components of the system are shown, color coded as follows: $P$ (blue); $I$ (red); $T$ (green); $R$ (purple).}}
\label{fig:delay_rho}
\end{figure}

Figure~\ref{fig:delay_rho} provides a complementary view of the effect of the average weak Gamma delay $\rho$, by following the geometry and biological admissibility of the stable periodic orbit itself as $\rho$ is varied. The top row compares this dependence across several closely related thalamo-reticular connectivity landscapes, allowing the effects of $w_{RR}$, $w_{RT}$, and $w_{TR}$ to be separated. Although the detailed amplitude and extent of the oscillatory branch change substantially with connectivity, the overall picture remains one in which delay and coupling strength act jointly to determine not only whether stable oscillations exist, but also whether those oscillations remain within the biologically admissible region.

The first two panels isolate the effect of reticular self-inhibition $w_{RR}$ at two different fixed values of the relay-to-reticular coupling $w_{RT}$. In the top left panel, for fixed $w_{TR}=0.5$ and $w_{RT}=0.4$, decreasing $w_{RR}$ from $0.2$ to $0.15$ produces a pronounced increase in the amplitude of the $P$ oscillation and substantially extends the stable periodic branch toward larger values of $\rho$. Notice that, for $w_{RR}=0.15$, the cycle persists throughout essentially the entire delay interval displayed. The same qualitative effect is visible in the top middle panel, where $w_{TR}=0.5$ and $w_{RT}=0.2$: reducing $w_{RR}$ from $0.2$ to $0.15$ again increases the cycle amplitude and extends the range of delays over which oscillations persist. However, both branches in this panel terminate at considerably smaller values of $\rho$ than their counterparts in the left panel. Thus, weakening reticular self-inhibition generally favors larger and more persistent oscillations, but the magnitude of this effect depends strongly on the background relay-to-reticular excitation. In particular, $w_{RR}$ does not act independently, but modulates the dynamics of a loop whose effectiveness is already conditioned by $w_{RT}$.

The top right panel directly illustrates this latter dependence. Here $w_{TR}=0.6$ and $w_{RR}=0.2$ are fixed, while $w_{RT}$ is increased from $0.2$ to $0.25$. The stronger relay-to-reticular coupling produces both a larger oscillation amplitude and a substantial extension of the periodic branch toward larger delays. The complementary effect of $w_{TR}$ can be seen by comparing the $w_{RT}=0.2$, $w_{RR}=0.2$ curves in the top middle and right panels. Increasing $w_{TR}$ from $0.5$ to $0.6$ enlarges the oscillations and extends their persistence in $\rho$. Taken together, these comparisons reinforce the interpretation suggested by the previous bifurcation diagrams: $w_{TR}$ and $w_{RT}$ operate as a coupled thalamo-reticular control pair. Strengthening either side of this reciprocal loop can substantially enlarge the delay range over which stable rhythmic activity is supported, whereas $w_{RR}$ regulates this activity through internal suppression of the reticular population.

The bottom panels examine in greater detail one representative case (from the top left panel), namely
$w_{TR}=0.5$, $w_{RT}=0.4$, and $w_{RR}=0.2$. They show that mathematical persistence of the stable cycle does not imply biological admissibility throughout the entire branch. At short delays, the oscillations fall within the spindle-compatible regime; at substantially larger delays, the same periodic branch reaches an infra-slow regime. Between these two ranges, however, the minimum of the $P$ component becomes negative. We therefore treat this intermediate portion as biologically inadmissible even though the periodic orbit remains mathematically stable. In this sense, the intervening ``no-go'' region is not a region without oscillations, but a region without biologically viable oscillations.

An important feature of this inadmissible interval is that the violation of positivity is relatively small: the lower envelope of the cycle remains close to the boundary $P=0$ rather than cutting more deeply into the negative domain. Consequently, the precise entry into and exit from the no-go region are sensitive to the connectivity landscape. Relatively modest changes in $w_{RR}$, $w_{RT}$, or $w_{TR}$ can shift the points at which the periodic orbit crosses the positivity boundary, thereby moving, narrowing, or enlarging the biologically inaccessible interval. The comparisons in the top row of Figure~\ref{fig:delay_rho} illustrate precisely this sensitivity. Thus, the existence of an intermediate inadmissible range appears to be a robust feature of the weak Gamma dynamics considered here, but its exact boundaries are not universal parameter thresholds.

From a dynamical perspective, this organization suggests a potential mechanism by which distributed delays can act as a selector of collective timescale rather than simply slowing a single rhythm continuously. For sufficiently short $\rho$, delayed feedback remains rapid enough to coordinate the recurrent thalamo-reticular interactions associated with the fast oscillatory regime. As $\rho$ increases, this coordination is disrupted before the distributed feedback has acquired the temporal scale required to support the much slower organization of the infra-slow regime. The stable periodic branch therefore passes through a biologically inadmissible interval rather than providing a physiologically meaningful continuum between the two rhythms. At still larger delays, a second admissible operating regime emerges. The two time courses shown in Figure~\ref{fig:delay_rho} for $\rho=0.5$ and $\rho=6.5$ illustrate the markedly different temporal organizations reached on the two sides of this interval.

This separation is consistent with the distinct physiological roles associated with these rhythms during NREM sleep. Spindles are rapidly coordinated thalamocortical events, whereas approximately $0.02$ Hz infra-slow fluctuations organize spindle occurrence and sleep fragility over much longer timescales \cite{lazar2019infraslow,lecci2017coordinated,watson2018cognitive}. At the Wilson--Cowan population scale, the larger values of $\rho$ required for the second regime should therefore not be interpreted as literal axonal or synaptic transmission delays. Rather, in the weak Gamma formulation, $\rho$ represents the characteristic timescale over which past population activity contributes to the recurrent feedback, and can therefore be interpreted as an effective integration or memory timescale of the population-level dynamics. Values of $\rho$ on the order of several to tens of seconds are consequently not intrinsically incompatible with the infra-slow regime. An approximately $0.02$Hz oscillation has a period on the order of $50$ seconds, and the physiological processes implicated in organizing infra-slow NREM dynamics likewise operate over substantially longer timescales than individual synaptic or axonal events. Thus, values such as $\rho=6$, $10$, or even larger should be understood as representing a distributed population history extending over a substantial fraction of an infra-slow cycle, rather than an implausibly long point-to-point neural transmission delay. In this interpretation, the large effective delays required by the model are consistent with the slow temporal organization of the processes associated with infra-slow activity.

A sufficiently extended population history can therefore provide the phase organization required for slow collective activity, whereas intermediate temporal integration may be poorly matched to either function: too slow to maintain the faster coordinated regime, yet insufficient to establish the broader infra-slow organization. In this way, the same underlying corticothalamic architecture can support two very different functional timescales, with the distributed delay selecting between them and the connectivity parameters controlling the accessibility and boundaries of the corresponding regimes.

\subsection{Attractors and transitions for the system with discrete delays}


Having established that, under weak Gamma distributed delays, both the cycling regimes and the properties of the oscillations depend significantly and nontrivially on the interplay between $\rho$ and the connectivity landscape, we next turn to discrete delays. Discrete delays have been extensively studied and are often the default choice in delayed neural-population models, since a single fixed lag provides a particularly direct and computationally tractable representation of non-instantaneous feedback~\cite{KaslikEtAl2022,KaslikEtAl2024}. At the same time, the actual temporal-integration profile implemented biologically remains unknown and may not be universal: different kernels may operate across individuals or brain networks, or within the same network under different states, tasks, or contexts. This possibility is further motivated by our previous results showing that the kernel itself can qualitatively reshape the accessible dynamics and the transitions between them. In particular, discrete delays produced broader oscillatory domains and richer bifurcation sequences, including transitions among periodic, quasi-periodic, and chaotic behavior~\cite{KaslikEtAl2022,KaslikEtAl2024}. The following section examines which features of the weak Gamma results persist under discrete delays and what additional dynamical complexity emerges. The numerical bifurcation analysis was carried out using MatCont 7.6~\cite{Dhooge2008MatCont}. For the discrete-delay system, the delayed state was approximated by a 10-stage linear chain, corresponding to an Erlang approximation of the discrete-delay kernel with mean $\rho$, thereby yielding a finite-dimensional ODE system suitable for numerical continuation.

\begin{figure}[h!]
\begin{center}
\includegraphics[width=0.5\textwidth]{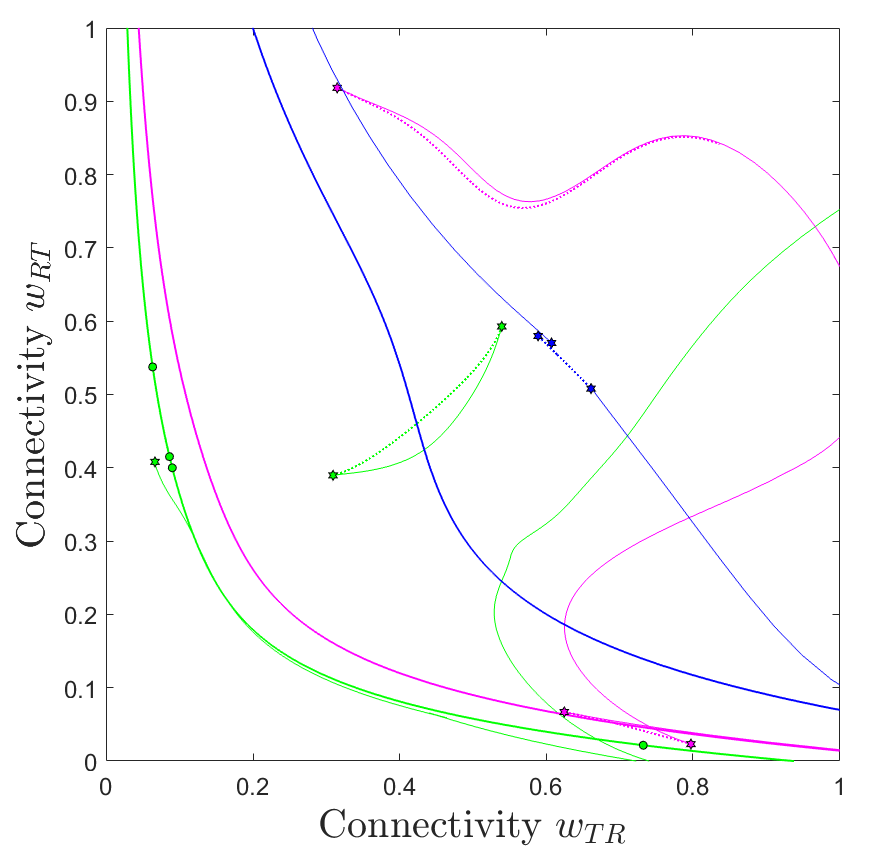}
\caption{\small \emph{{\bf Regions of stable oscillations in the $(w_{TR},w_{RT})$ parameter plane for three discrete delay values:} $\rho=0.1$ (blue), $\rho=0.5$ (pink), and $\rho=1$ (green). In each case, the Hopf curve associated with the onset of oscillations is shown as a thick line in the corresponding color, while the LPC or PD curves marking their termination are shown as thin lines in the same color. For $\rho=1$ (green), stable cycles are not generated directly at the Hopf curve, but instead emerge through an LPC mechanism, as clarified in Figure~\ref{bif_wTR_for_multiple_rho}. A second relevant LPC branch is therefore also shown in thin green, since it contributes to delimiting the region of stable oscillations in this case.}}
\label{codim2_discrete}
\end{center}
\end{figure}

To obtain a broad picture of how the system's dynamics change with the discrete delay $\rho$, Figure~\ref{codim2_discrete} shows the corresponding bifurcation structure in the $(w_{TR},w_{RT})$ plane. This can be readily compared to the corresponding illustration for weak Gamma delays in Figure~\ref{wTR_wRT}. This principal equilibrium branch can undergo multiple Hopf bifurcations. The larger number of Hopf points is not surprising for a system with discrete delays, nor is the substantially richer bifurcation structure that emerges compared with the weak Gamma case. For sufficiently small delays, however, the basic organization identified for the weak Gamma kernel remains recognizable. For $\rho=0.1$ and $\rho=0.5$, stable cycles are generated directly at the first Hopf point and persist until they lose stability through an LPC or period-doubling bifurcation. For $\rho=1$, access to this ``principal'' stable oscillatory regime is already mediated by a more complicated LPC structure, as will be clarified by the one-parameter continuations below. Thus, at short discrete delays, the onset and termination of this principal stable cycle can still be described approximately by a Hopf onset and an LPC/PD offset, much as in the weak Gamma system, although this organization occurs over considerably smaller values of $\rho$. As the delay increases, this relatively simple mechanism breaks down and the organization of the stable oscillatory regimes becomes progressively more complex. These bifurcation boundaries describe only the existence and stability of this principal cycle family and therefore provide only the first layer of the analysis. As in the weak Gamma case, we must additionally determine over which parameter ranges these cycles are biologically admissible; how their amplitude, frequency, and duty cycle vary with the parameters and what these changes imply physiologically; and whether the discrete-delay system supports additional stable cycles or other attractors.

Figure~\ref{bif_wTR_for_multiple_rho} clarifies in greater detail the transitions into and out of oscillatory behavior summarized in Figure~\ref{codim2_discrete}, by fixing the same representative values of $\rho$ and constructing one-parameter bifurcation diagrams with respect to $w_{TR}$ (for fixed $w_{RT}=0.4$). The complexity of the bifurcation structure increases rapidly as the delay is increased. In particular, progressively more Hopf bifurcations appear along the principal equilibrium branch. For each Hopf point detected by the MatCont~7p6 continuation, we followed the corresponding periodic-orbit branch through successive LPC and period-doubling (PD) bifurcations, continuing it until we were reasonably confident that no additional interval of full stability would be recovered. Stable cycles are shown in solid cyan shading with thick contour curves, irrespective of their biological relevance. Their positivity and the possible physiological interpretation of the attainable rhythms will be considered separately below. Unstable portions of the cycle extensions are included for dynamical context and are represented using thin contours and transparent shading in separate colors, to distinguish them more clearly.

For $\rho=0.1$, a single supercritical Hopf bifurcation occurs within the biological range of $w_{TR}$, giving rise to a stable cycle that persists until an LPC, where the branch collides with an unstable cycle and disappears. The cycle remains positive in all four components throughout its stable existence, up to approximately $w_{TR}\sim0.78$.

For $\rho=0.5$, the principal equilibrium branch undergoes two supercritical Hopf bifurcations. The first generates a stable cycle, which persists to approximately $w_{TR}\sim0.9$ before terminating at an LPC. By the time the second Hopf point is reached, however, the equilibrium already possesses unstable directions inherited from the previous stability loss. Consequently, although this Hopf is also supercritical, the cycle generated there is unstable, since it inherits transverse instability from the equilibrium.

\begin{figure}[h!]
\begin{center}
\includegraphics[width=0.9\textwidth]{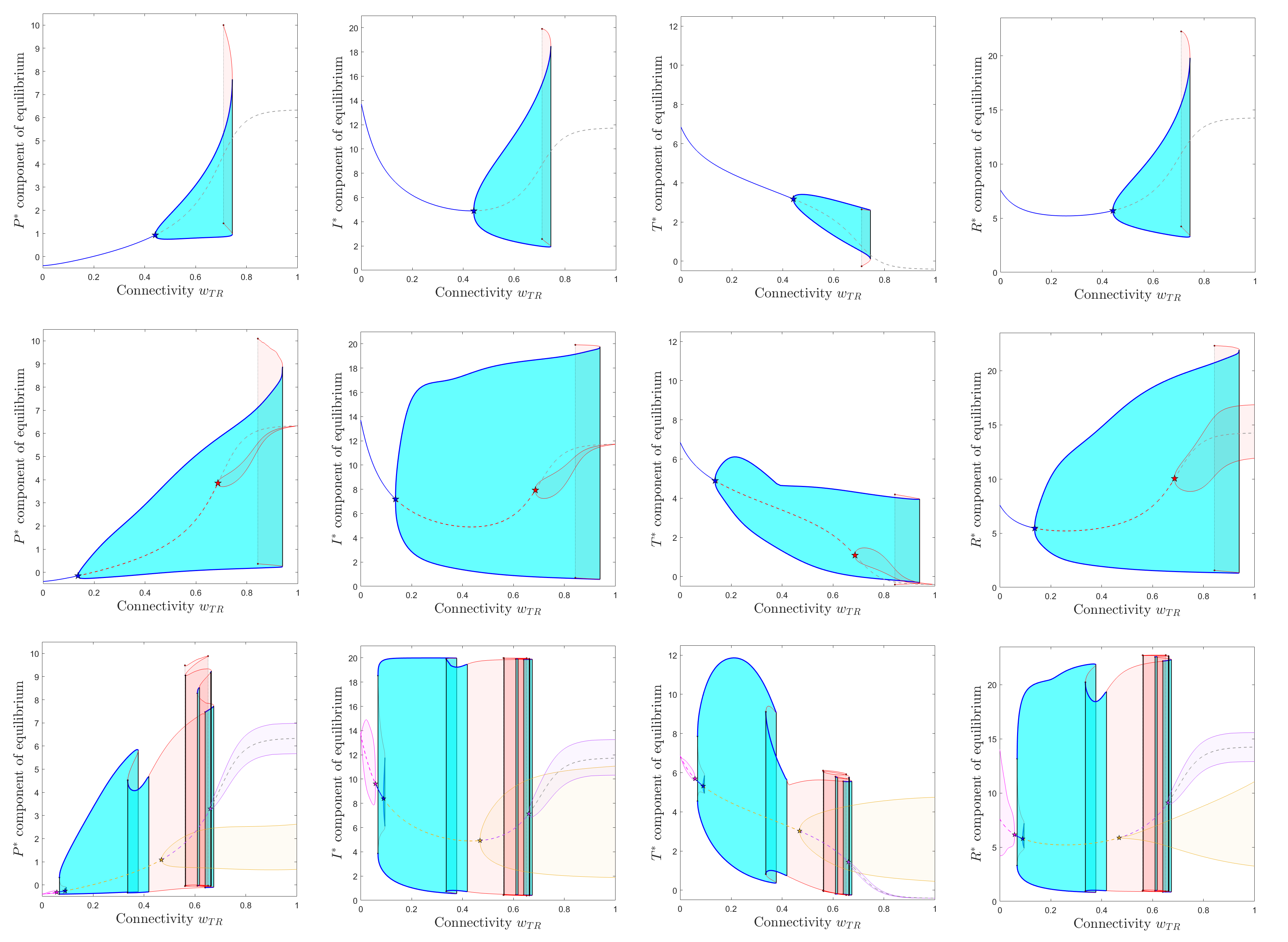}
\caption{\small \emph{{\bf Effect of delay on dependence of $w_{TR}$,} for fixed $w_{RT}=0.4$. The bifurcation diagram with respect to the connectivity $w_{TR}$ is sketched for three discrete delay values, as follows: $\rho=0.1$ (top row); $\rho=0.5$ (middle row); $\rho=1$ (bottom row). Each panel shows the evolution of an equilibrium branch, which undergoes supercritical Hopf bifurcations (marked with stars), and the limit cycles that are generated at these Hopf points. As before, stable cycles are shown in solid cyan, and unstable cycles (identified through the values of the Floquet multipliers and of the equilibrium eigenvalues, when applicable) are shown in various transparent shades, to be able to distinguish them from one another. Some of the cycle extensions identified multiple stable intervals, hence the cycle branches were followed until they appeared to have resolved in terms of producing any additional stable windows. The LPC and LCPD (flip) bifurcations are marked as vertical solid lines if they delimit regions of stable from unstable cycles. Other additional LPC bifurcations are shown as dotted vertical lines. Not all cycle bifurcations are shown, to avoid cluttering of the panels. All other parameters were fixed to their table baseline values.}}
\label{bif_wTR_for_multiple_rho}
\end{center}
\end{figure}

For $\rho=1$, the bifurcation structure becomes substantially more intricate. The stable cycle emerges through a bifurcation sequence initiated at the second subcritical Hopf bifurcation; this gives rise to an unstable limit cycle that subsequently reaches an LPC bifurcation, where the periodic-orbit branch folds and becomes stable. As this cycle is continued in $w_{TR}$, it repeatedly loses and regains stability through secondary cycle bifurcations, producing several distinct windows of stable oscillations, shown in cyan in Figure~\ref{bif_wTR_for_multiple_rho}. The intervening unstable portions are shown in various transparent shades (to be able to distinguish them from one another) and were continued until we were reasonably confident that no further recovery of stability would occur.


Stability alone, however, does not guarantee biological admissibility. A closer examination shows that all of the stable cycles obtained for $\rho=1$ at this parameter set enter the negative domain in the $P$ component and therefore fail the positivity requirement of the model. In our interpretation, this loss of positivity provides the effective biological termination of the oscillatory regime, even when the mathematical cycle itself remains stable. Thus, among the discrete-delay values examined here, biologically admissible stable oscillations are restricted to the smaller delays. For $\rho=0.1$, the stable cycle remains positive throughout essentially its entire oscillatory window, approximately $0.44 \lesssim w_{TR} \lesssim 0.74$. For $\rho=0.5$, the positive stable window is approximately the same (excluding small intervals near onset and termination, where $P$ and $T$ take negative values, respectively). 

\begin{figure}[h!]
\begin{center}
\includegraphics[width=\textwidth]{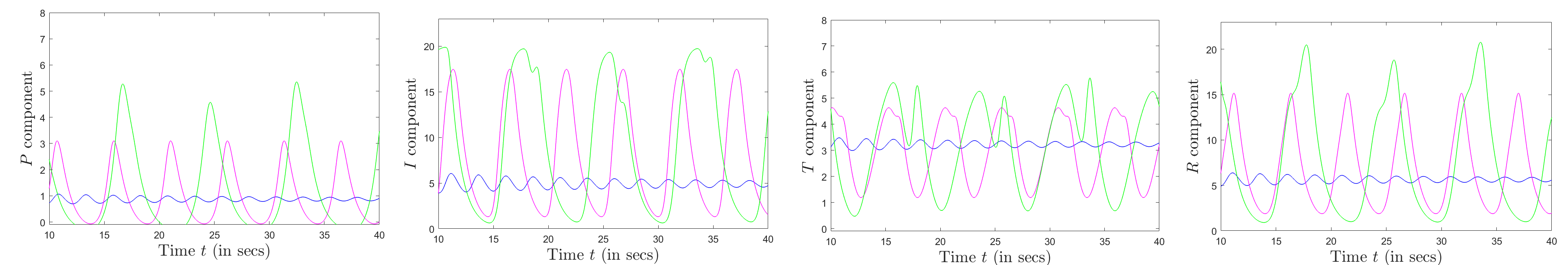}
\includegraphics[width=\textwidth]{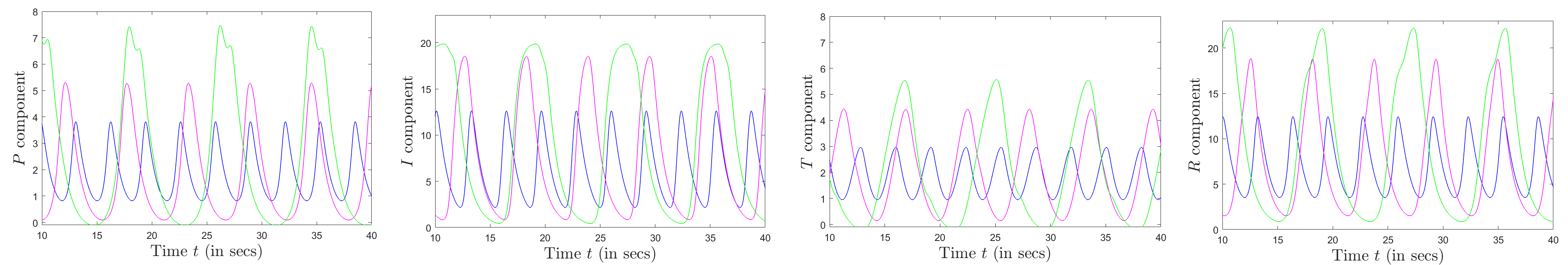}
\caption{\small \emph{{\bf Time evolutions of the system components for two thalamo-reticular connectivity levels:} $w_{TR}=0.44$ (top panels) and $w_{TR}=0.65$ (bottom panels). Each panel illustrates one component, for three delay values: $\rho=0.1$ (blue curve), $\rho=0.5$ (pink curve) and $\rho=1$ (green curve). The relay-TRN connectivity was fixed to $w_{RT}=0.4$. All other parameters were fixed to their table baseline values.}}
\label{solutions}
\end{center}
\end{figure}

To better contextualize, the four components of an example temporal solution are shown in Figure~\ref{solutions}, for $w_{TR}=0.44$ and for $w_{TR}=0.65$ (at which all three values of $\rho$ exhibit stable oscillations). These further illustrate how both $\rho$ and $w_{TR}$ shape the resulting rhythms. For $\rho=0.1$, increasing $w_{TR}$ from $0.44$ to $0.65$ strengthens the oscillations and brings their amplitudes closer to the spindle-compatible ranges, particularly for the more strongly recruited $I$ and $R$ populations. At $\rho=0.5$, the bursts become broader and larger; increasing $w_{TR}$ strengthens them further, but can also push some population firing levels beyond the ranges that provide the closest spindle match. By $\rho=1$, the oscillations are slower and larger still, and the $P$ and $T$ components become slightly negative. Importantly, this dependence differs qualitatively from that obtained with the weak Gamma kernel, where increasing $\rho$ separated a short-delay spindle regime from a long-delay infra-slow regime through an intermediate ``no-go'' region of non-admissible oscillations. The discrete-delay model does not show the same delay-driven separation of timescales. We return to the implications of this difference between delay kernels in the Discussion.

\begin{figure}[h!]
\begin{center}
\includegraphics[width=0.9\textwidth]{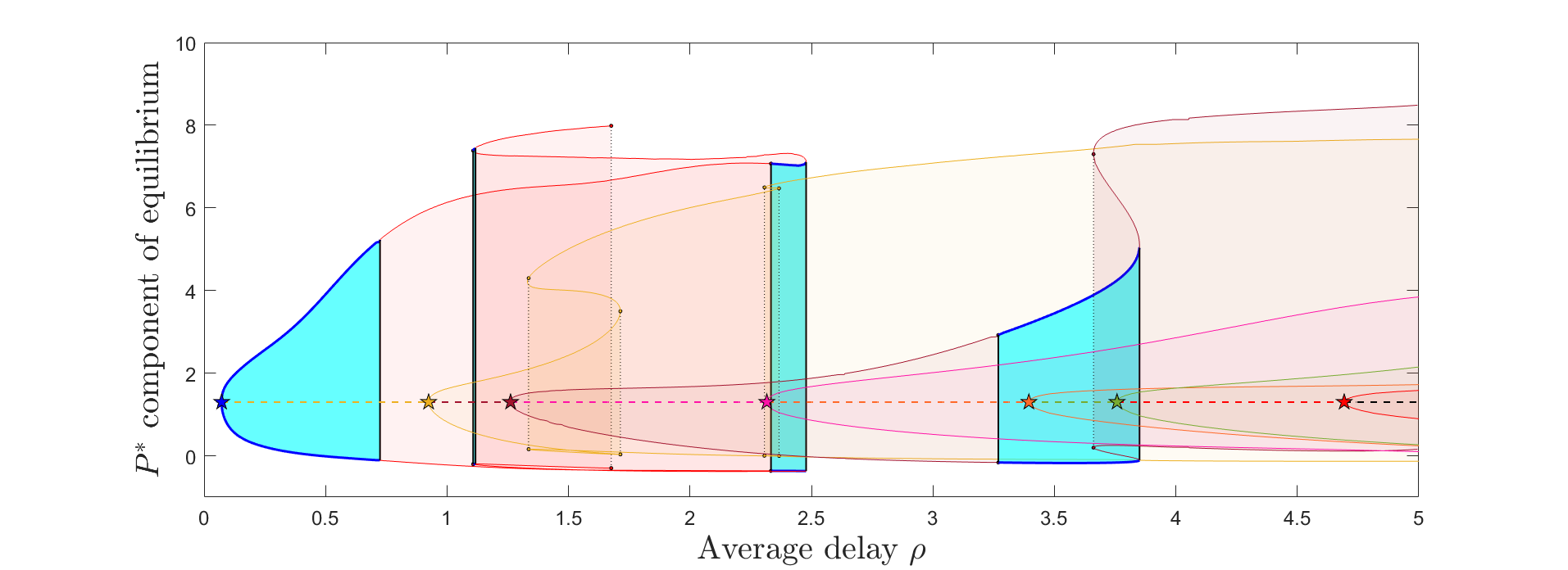}
\includegraphics[width=0.9\textwidth]{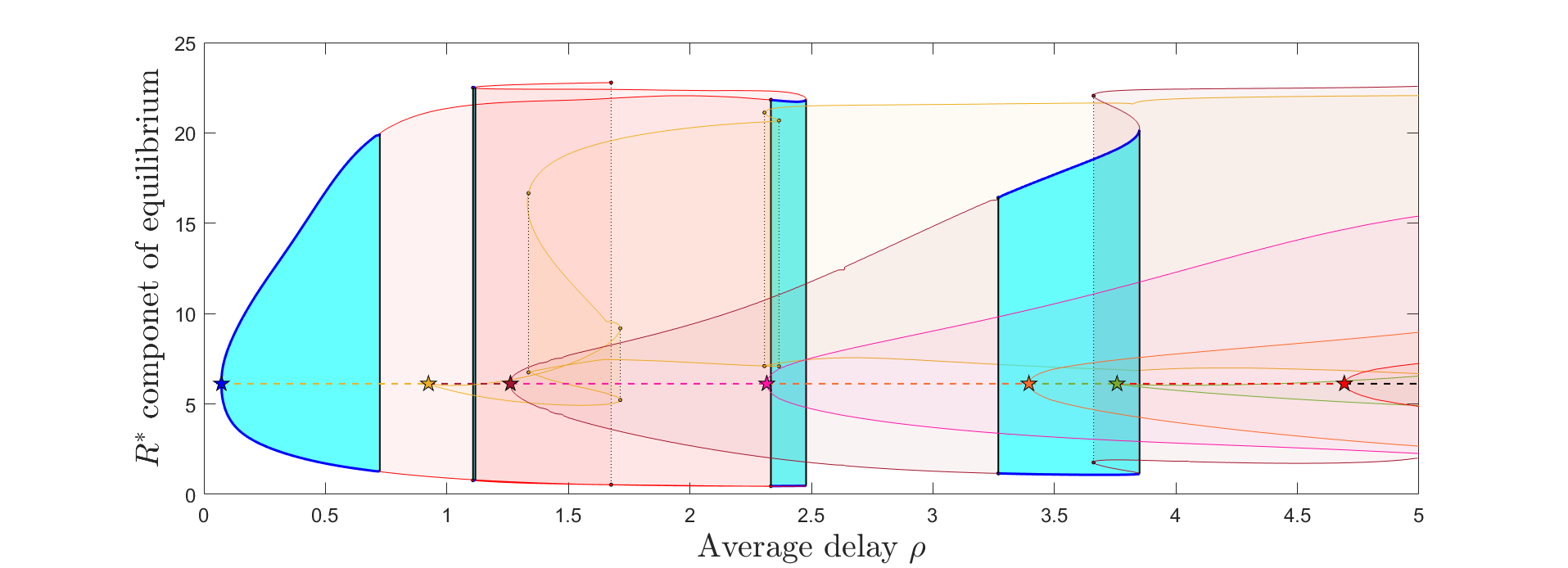}
\caption{\small \emph{{\bf Effect of discrete delay on oscillatory regimes.} The bifurcation diagram with respect to the average discrete delay $\rho$ is sketched for fixed connectivities $w_{TR}=0.5$ and $w_{RT}=0.4$. The panels show from the perspective of the cortical pyramidal component $P$ (top) and of the reticular nucleus $R$ (bottom) the evolution of the primary equilibrium branch with increasing $\rho$. The branch undergoes supercritical Hopf bifurcations (marked with stars), and the limit cycles that are generated at these Hopf points. As before, stable cycles are shown in solid cyan, and unstable cycles are shown in various transparent shades. The cycle branches were followed until they appeared to have resolved in terms of producing any additional stable windows. The LPC and LCPD (flip) bifurcations are marked as vertical solid lines if they delimit regions of stable from unstable cycles. Other additional LPC bifurcations are shown as dotted vertical lines. Not all cycle bifurcations are shown, to avoid cluttering of the panels.}}
\label{bif_rho_for_wTR_0_5_wRT_0_4}
\end{center}
\end{figure}

The three values of $\rho$ considered above provide only discrete snapshots of the dependence on delay. To examine this dependence more continuously, we next construct a one-parameter bifurcation diagram with respect to $\rho$, fixing the representative connectivity pair $w_{TR}=0.5$ and $w_{RT}=0.4$ (the same values used in the weak Gamma analysis). Figure~\ref{bif_rho_for_wTR_0_5_wRT_0_4} confirms the existence of an interval of small discrete delays for which stable oscillations are biologically admissible. At the same time, it shows that additional windows of stable oscillatory behavior can occur at larger values of $\rho$; however, these cycles extend into the negative domain and therefore do not satisfy the positivity requirement of the model.

Interestingly, the loss of biological admissibility in these higher-$\rho$ oscillatory windows is generally mild: the cycles typically enter the negative domain only slightly, and only in a single component. This suggests that these regimes may lie very close to the boundary of biological admissibility, rather than representing fundamentally unrealistic dynamics. In particular, relatively small changes in one or more of the secondary model parameters---many of which have been fixed throughout the present analysis at their baseline values from Tables~\ref{tab:noncoupling_refs} and~\ref{tab:coupling_refs} may be sufficient to shift these cycles entirely into the positive domain. Thus, the additional stable oscillatory windows identified at larger discrete delays may be viewed as latent biologically admissible regimes, accessible under modest changes in the broader parameter configuration. A systematic exploration of this possibility is important, but falls beyond the scope of the present work, whose primary aim is a broad comparison of the dynamical effects produced by different delay formulations. A more comprehensive investigation of these higher-delay oscillatory regimes, including their dependence on the remaining model parameters, is currently being pursued in conjunction with further refinements of the model architecture, as discussed in the section on limitations and future work.

\section{Discussion}

\subsection{Specific comments on the model}

In this paper, we studied the dynamics of a corticothalamic-reticular circuit (CTRC) under three temporal formulations: the original system without delays, the same system with weak Gamma distributed delays, and then with discrete delays. The delay-free analysis first established how the connectivity architecture organizes access to oscillatory behavior, while the delayed systems allowed us to ask how this underlying landscape is reshaped when temporal integration is introduced. Across the three formulations, the results point to a hierarchy of control in which recurrent cortical excitation largely determines whether oscillatory behavior is accessible, while the thalamo-reticular couplings provide finer control over where oscillations occur and what form they take.

An important mathematical feature of both delayed formulations is that the equilibria themselves are unchanged by the introduction of delay. For any fixed connectivity configuration, the equilibrium locations are therefore inherited directly from the delay-free system; what the temporal kernel changes is their stability and the bifurcation structure organized around them. Delays can consequently create or remove access to oscillatory regimes, shift their onset and termination, and generate substantially different periodic-orbit structures without altering the underlying fixed points. This separation is useful conceptually: connectivity determines the equilibrium landscape, while the temporal structure of feedback can profoundly reorganize the dynamics supported around it.

In the delay-free system, recurrent pyramidal excitation $w_{PP}$ acts as an important access parameter: increasing $w_{PP}$ can destabilize the equilibrium and open a bounded oscillatory window, while corticothalamic excitation of the relay population, $w_{TP}$, primarily shifts the location of this window. Once the cortical background is appropriately positioned, the balance between $w_{TR}$ and $w_{RT}$ becomes particularly important. Increasing $w_{TR}$ can drive the system into oscillations through a Hopf bifurcation and later out of them through an LPC, but the location, width, amplitude and biological viability of this window depend strongly on $w_{RT}$ and on the cortical operating point. Conversely, changing $w_{RT}$ can move the system between stable equilibria, stable cycles and, for some parameter combinations, more complicated dynamical regimes. These results motivated our focus on the thalamo-reticular pair $(w_{TR},w_{RT})$ in the delayed analysis: neither coupling acts as an isolated control parameter, and the relevant aspect is their balance within the broader CTRC loop.

This dependence becomes even clearer for weak Gamma distributed delays. In the $(w_{TR},w_{RT})$ plane, changing the average delay $\rho$ does not simply translate the oscillatory domain uniformly. Instead, its effect depends strongly on the underlying connectivity. For relatively large $w_{RT}$, increasing $\rho$ progressively compresses the oscillatory region: the Hopf boundary moves toward larger $w_{TR}$ while the terminal LPC moves toward smaller $w_{TR}$. In this regime, stronger delay therefore requires stronger TRN-to-relay coupling to initiate oscillations, while at the same time reducing how far that coupling can be increased before oscillations terminate. At lower $w_{RT}$, however, this dependence becomes nonmonotone. Moderate increases in $\rho$ can initially broaden the available $w_{TR}$ window, before further increases reverse the trend and again compress it. Thus, the influence of delay cannot be separated from the excitatory return from relay to TRN: the same change in $\rho$ can either facilitate or suppress oscillations depending on the connectivity background.

The complementary dependence on $w_{RT}$ reinforces this picture. The effect of increasing the delay depends strongly on the underlying strength of the reciprocal relay--TRN couplings: in some connectivity regimes, increasing $\rho$ progressively contracts the oscillatory window, whereas in others it can first suppress stable oscillations and then allow them to re-emerge within a different range of $w_{RT}$. Thus, delay does not act as a simple additive perturbation to the connectivity landscape. Rather, $w_{TR}$ and $w_{RT}$ define the dynamical background on which temporal integration operates, so that changing the connectivity can qualitatively alter the effect of increasing $\rho$.

The connectivity parameters also shape the temporal form of the rhythms once oscillations are present. Across the weak Gamma simulations, changes in $w_{TR}$ and $w_{RT}$ alter not only whether the system cycles, but also its amplitude, frequency and duty cycle. For a fixed delay, moving through the oscillatory window can therefore move the system between rhythms with very different temporal profiles; conversely, at fixed connectivity, increasing $\rho$ can slow and broaden the oscillations substantially. This joint dependence is important physiologically, since entry into a stable cycle alone is not sufficient to identify a spindle-like or infra-slow regime. The rhythm must also occupy an appropriate range of population activities and temporal organization.

Reticular self-inhibition $w_{RR}$ provides an additional control on the oscillatory dynamics, but its effect depends on the strength of the reciprocal relay--TRN loop. Reducing $w_{RR}$ weakens the inhibitory feedback acting directly on the reticular population, generally allowing larger oscillatory excursions and permitting stable cycles to persist over a wider range of delays. The extent of this effect, however, depends on how strongly changes in reticular activity are transmitted to the relay population through $w_{TR}$ and fed back to the TRN through $w_{RT}$. 

Against this detailed connectivity dependence, one of the most interesting results of the weak Gamma model is the persistence of a common higher-level organization. Across the connectivity configurations examined, short delays support biologically admissible spindle-like activity, whereas sufficiently long delays support a distinct infra-slow regime. Between them lies a broad interval in which the periodic orbit may remain mathematically stable but becomes biologically inadmissible because one population enters the negative domain. The delay therefore does not simply slow a spindle continuously into an infra-slow oscillation. Instead, it selects between two distinct functional operating regimes, with a biological ``no-go'' interval separating them. Connectivity controls where these regimes occur, how broad they are, and how strongly the populations are recruited, but does not erase this overall organization.

This separation is especially interesting because spindle and infra-slow rhythms play different roles in NREM sleep. The former reflect relatively rapid coordinated thalamo-cortical activity, whereas infra-slow fluctuations organize spindle-rich and spindle-poor periods over much longer timescales. Within the weak Gamma model, the same circuit can support both forms of organization without continuously passing through every intermediate rhythm. The inadmissible interval effectively partitions the available dynamics into two functional modes. In this sense, temporal integration acts not merely as a delay in transmission but as a selector of collective timescale.

The discrete-delay formulation preserves some aspects of this picture while departing from it in important ways. At short delays, it can still support stable spindle-like activity, and the resulting rhythms remain strongly shaped by the thalamo-reticular connectivity. As $\rho$ increases, however, the bifurcation structure becomes substantially richer, with multiple Hopf bifurcations and repeated losses and recoveries of periodic-orbit stability through LPC and period-doubling bifurcations. Stable oscillatory windows can therefore disappear and reappear throughout the parameter range, rather than remaining organized around one simple oscillatory corridor.

This richer structure does not reproduce the same clean delay-driven separation between spindle and infra-slow activity observed for weak Gamma integration. This does not necessarily make the discrete formulation biologically irrelevant. Indeed, many of the additional stable cycles lie only slightly outside the biologically admissible domain, often becoming weakly negative in only one population. They therefore represent a reservoir of nearby dynamics that could potentially become admissible under modest changes in parameters that were held fixed here. The difference is instead one of organization and robustness: under the parameterization considered in this study, weak Gamma integration naturally produces well-separated functional regimes, whereas the discrete delay exposes a larger repertoire whose biological realization appears to depend more strongly on finer parameter adjustment.

This distinction is consistent with our broader view that there is no reason to assume that one delay kernel should operate universally across neural systems, or even within a single system under all physiological conditions. Different temporal-integration profiles may be appropriate for different networks, states or functional demands. For the CTRC system considered here, however, the weak Gamma kernel provides a particularly economical organization of the two sleep-related timescales of interest, while the discrete formulation reveals greater dynamical flexibility and a substantial amount of latent oscillatory structure. We return below to the broader implications of this contrast and to the possibility that different temporal kernels may themselves form part of the functional repertoire of neural systems.

\subsection{General comments and significance}

The results above suggest a broader interpretation of temporal integration in the corticothalamic system. The TRN-centered circuit is already known to participate in dynamics spanning very different timescales. At the faster end, reciprocal interactions between TRN and thalamic relay populations are central to the generation of sleep spindles, with cortical feedback further shaping and propagating these events~\cite{halassa2011selective,krosigk1993cellular,mak2017coordination}. At a much slower scale, approximately $0.02$ Hz fluctuations during NREM sleep organize the occurrence and clustering of spindles and contribute to alternating periods of sleep continuity and fragility~\cite{lazar2019infraslow,lecci2017coordinated,watson2018cognitive}. Thus, the spindle and infra-slow regimes considered here should not be viewed as two versions of the same oscillation. Rather, they represent distinct levels of temporal organization that coexist within the same broader corticothalamic system.

From this perspective, the transition obtained with weak Gamma delays is physiologically plausible not because a spindle is expected to slow continuously until it becomes an infra-slow rhythm, but almost for the opposite reason. The model provides a mechanism by which the same connectivity architecture can support two distinct functional modes while preventing a biologically meaningful continuum between them. Short temporal integration favors spindle-like activity, whereas sufficiently extended integration supports a much slower network-level organization. Between these regimes, the periodic solution persists mathematically but leaves the biologically admissible state space. The resulting ``no-go'' interval therefore provides a dynamical separation between two functions that operate on very different timescales.

This separation should not be interpreted as a rigid boundary in parameter space. In the weak Gamma model, as in several of the discrete-delay regimes, loss of admissibility generally occurs through relatively small negative values in a single population. Consistently, changing the connectivity background can substantially shift the onset and termination of the no-go region, and can widen or narrow both the spindle-like and infra-slow windows. The robustness of the weak Gamma result therefore lies less in the precise location of these boundaries than in the persistence of the overall organization: across the connectivity configurations examined here, the two physiologically meaningful regimes remain separated by an intermediate range that is not biologically viable. The connectivity parameters determine where that separation occurs, while the temporal kernel helps determine its overall form.

Such an organization may itself be advantageous. Neural circuits must remain capable of changing state, but physiological function may also require protection against arbitrary transitions into every dynamical state that the underlying network can mathematically support. In the weak Gamma system, changes in temporal integration provide access to two relevant operating regimes, while the intervening loss of admissibility constrains the routes between them. The same anatomical circuit can therefore be reused for functions expressed on very different timescales without requiring a separate architecture for each. In this sense, temporal integration may provide an economical mechanism for state-dependent reorganization while preserving a degree of separation between functional modes.

The discrete delay system suggests a different balance. Its mathematical bifurcation structure is considerably richer, with multiple Hopf points, repeated losses and recoveries of cycle stability, and, at larger delays, oscillations with more complicated temporal profiles, including the multi-peaked structure seen for $\rho=1$. Importantly, however, this should not be equated with a correspondingly larger repertoire of physiologically available rhythms. Many of these stable-cycle windows enter the negative domain and therefore remain biologically inadmissible under the parameter values considered here. What the discrete system reveals instead is a larger collection of nearby dynamical possibilities, some of which lie only slightly beyond the biological boundary and could potentially become admissible under relatively small changes in other parameters.

This distinction suggests a possible tradeoff between robustness and dynamical flexibility. The weak Gamma formulation organizes the dynamics comparatively cleanly around two physiologically interpretable regimes. The discrete formulation, by contrast, places the system near a greater variety of alternative oscillatory states, but access to many of them depends more strongly on the surrounding parameter configuration. Such proximity may be useful: changes in connectivity, external input, or physiological state could potentially recruit behaviors that are not accessible in the baseline configuration. At the same time, an operating landscape containing many nearby stability windows and more complicated rhythms may be more sensitive to perturbation than one in which the dominant functional regimes are more clearly separated.

This possibility is particularly interesting in the thalamocortical system because the circuitry involved in normal spindle generation is also capable of supporting abnormal synchronized activity. Changes in thalamic and corticothalamic interactions have been associated with pathological oscillatory regimes, including epileptic thalamocortical dynamics~\cite{pinault2003cellular,Sheeba2008}. We therefore raise, as a hypothesis rather than a conclusion of the present study, the possibility that the temporal organization of feedback may contribute not only to the repertoire of rhythms available to a circuit, but also to its robustness against transitions into atypical dynamics. A broadly distributed temporal kernel such as the weak Gamma form considered here may favor stronger segregation of physiological operating regimes, whereas a more sharply concentrated temporal response may place the system closer to additional and potentially more complicated oscillatory states. The richer bifurcation structure and unusual cycle shapes observed in the discrete delay system are consistent with this possibility, but should not themselves be interpreted as signatures of pathology.

Conversely, it would be premature to identify weak Gamma integration as a uniquely ``healthy'' temporal kernel, or discrete delays as intrinsically pathological. We have argued previously that temporal integration profiles need not be universal, and that different kernels may effectively operate in different neural networks, individuals, physiological states, tasks, or contexts. The present results instead suggest that different kernels may provide different dynamical advantages. A more broadly distributed temporal response may favor robust separation between a limited number of functional regimes, whereas a more concentrated response may leave additional dynamical possibilities nearby, at the cost of greater dependence on the precise operating point of the network. Which organization is preferable may itself depend on the function being performed.

More generally, these results shift the question away from identifying a single ``correct'' representation of neural delay and toward asking what different forms of temporal integration allow a network to do. In the present CTRC model, weak Gamma integration provides a particularly economical mechanism for separating spindle-like activity from infra-slow organization, while discrete delays expose a more intricate surrounding dynamical landscape whose biological accessibility depends on finer parameter adjustment. Whether robust temporal dispersion is preferentially associated with normal physiological regulation, and whether changes in the temporal-integration profile can increase susceptibility to atypical or pathological rhythms, remain open questions. Addressing them will require greater physiological detail in the model, together with empirical constraints on how corticothalamic networks actually integrate activity over time.

\subsection{Limitations and future work}

The present model was deliberately formulated at the Wilson-Cowan population level, which provides a tractable framework for identifying bifurcations, mapping oscillatory regimes, and examining how connectivity and temporal integration interact across a relatively large parameter space. This level of description is well suited to the system-level questions addressed here, but compresses many cellular and within-population mechanisms into effective variables and coupling parameters. This limitation is particularly relevant for the TRN. The $R$ node represents the mean activity of a population whose internal architecture is itself dynamically important: TRN neurons exhibit intrinsic bursting, communicate through both chemical inhibition and electrical coupling, and may organize into spatially structured patterns of synchronization. None of this internal organization can be represented explicitly by a single Wilson-Cowan variable.

A natural next step is therefore to introduce cellular resolution selectively within the TRN, rather than replacing the entire population model by a uniformly more detailed description. In our previous work, we studied networks of reduced Rinzel-Golomb neurons representing TRN cells, including inhibitory interactions and structured gap-junctional coupling. We plan to incorporate such a network directly within the present corticothalamic architecture, replacing the single $R$ node by a heterogeneous population of reduced conductance-based neurons, coupled through all-to-all chemical inhibition together with clustered electrical coupling. The cortical excitatory and inhibitory populations and the thalamic relay population can initially remain at the Wilson-Cowan level. The resulting hybrid model would therefore connect cellular and population descriptions within the same CTRC circuit, allowing dynamics generated by the internal organization of the TRN to interact directly with the larger corticothalamic loop.

An important part of this extension will be to establish a meaningful translation between the two modeling scales. The signal transmitted from a Wilson-Cowan population to the cellular TRN network must be related to the inputs received by individual Rinzel-Golomb neurons, while the heterogeneous activity of the TRN network must in turn be coarse-grained into an effective population signal that can interact with the remaining Wilson-Cowan nodes. Developing this correspondence will allow us to ask not only whether the hybrid model reproduces the macroscopic regimes identified here, but also under what conditions a Wilson--Cowan node provides an adequate reduction of a heterogeneous neuronal network, and when microscopic mechanisms such as intrinsic bursting, gap-junctional organization, and cellular heterogeneity qualitatively alter the population-level bifurcation structure. The goal is therefore to preserve the Wilson-Cowan description and an overall envelope, but also introduce cellular detail specifically where mechanisms operating below the population scale are likely to be essential.

The delay formulations considered here are also intentionally idealized. Weak Gamma and discrete kernels provide two mathematically clear and dynamically distinct ways of representing temporal integration, but real corticothalamic pathways need not share a single delay distribution or even the same effective temporal profile across connections. Different pathways may have different characteristic delays and degrees of temporal dispersion, and these properties may themselves vary with physiological state. Future work can therefore extend the present analysis to heterogeneous and pathway-specific kernels. The multiscale formulation may eventually provide an additional route toward this problem: rather than prescribing the effective temporal kernel entirely at the population level, one may ask whether particular forms of distributed temporal integration emerge naturally from heterogeneous cellular, synaptic, and network dynamics.

Although the parameter ranges and relative coupling strengths were informed by the experimental and modeling literature, their specific numerical values remain phenomenological rather than being directly estimated from empirical data. A further limitation is that, in order to make the connectivity-delay interactions interpretable, many secondary parameters were held fixed at their baseline values while a smaller number of influential couplings were explored systematically. The resulting boundaries of biological admissibility should therefore not be interpreted as universal. This is especially important because, for both delay formulations, some periodic solutions cross into the negative domain only slightly and often in a single component. Relatively modest changes in the surrounding parameter configuration can consequently move the boundaries of the admissible and inadmissible regions. In the weak Gamma system, this may shift the extent of the spindle, ``no-go,'' and infra-slow windows without necessarily destroying their broader organization; in the discrete system, it may render some of the presently near-admissible stable windows fully positive. A more extensive exploration of these secondary parameter dependencies is therefore warranted.

Relatedly, positivity provides a necessary biological constraint for the current Wilson-Cowan formulation, but the behavior of trajectories close to zero also reflects the limits of the coarse population description. A small negative value in a mean-field variable is mathematically inadmissible and must be excluded here, but it does not by itself specify what cellular mechanism would terminate or reorganize the corresponding activity in a biological circuit. The hybrid model should help clarify this distinction by allowing some of the population-level boundaries identified here to be examined in terms of explicit neuronal recruitment, synchrony, and firing dynamics.

Finally, the present analysis is primarily deterministic and focuses on asymptotic attractors and their bifurcations. Physiological sleep spindles are finite, state-dependent events whose initiation and termination are influenced by ongoing fluctuations, transient inputs, and neuromodulatory state. Incorporating noise, transient perturbations, and slowly varying external drive would therefore provide a natural subsequent step, allowing the bifurcation structures identified here to serve as an organizing framework for understanding how the system actually moves between regimes. Together with the multiscale TRN extension, this would allow future work to connect cellular mechanisms, population dynamics, temporal integration, and state-dependent transitions within a common framework.

\section*{Appendix A}

\begin{thm}[Persistence of stability for small relative delays]
\label{thm:small-q-stability}
Suppose that the equilibrium $X^*$ is linearly asymptotically
stable for the delay-free system \eqref{eq.model.no.delay}. For \(\operatorname{Re}z\geq0\), consider $G(z):=\left(      zI+D-K\right)^{-1}K $.
Denote
\[
    r_K(z)
    :=
    \begin{cases}
    \displaystyle
    \frac{1}{\rho_{\mathrm{sp}}(G(z))},
    &
    \rho_{\mathrm{sp}}(G(z))>0,
    \\[3mm]
    +\infty,
    &
    \rho_{\mathrm{sp}}(G(z))=0.
    \end{cases}
\]
where $\rho_{\mathrm{sp}}$ denotes the spectral radius,
and \begin{equation}\label{eq.q0}
       q_0
    :=
    \inf_{\substack{
       z\in\mathbb{C},\ \operatorname{Re}z\geq0,\ z\neq0\\
       r_K(z)\leq2
    }}
    \frac{r_K(z)}{|z|},
\end{equation}
with the convention that the infimum of the empty set is \(+\infty\).

Then \(q_0\in(0,\infty]\), and for every
\(
    0\leq q<q_0
\)
the characteristic equation \eqref{eq.scaled-characteristic-equation} has no roots in the closed right half-plane.
Consequently, the equilibrium remains asymptotically stable for all
sufficiently small relative delays.
\end{thm}

\begin{proof}
Let us consider 
\(
    M(z,q)
    :=
    zI+D-\widehat h(qz)K
\)
and
\(
    M_0(z)
    :=
    zI+D-K.
\)
As $\sigma(K-D)$ is included in the open left half-plane, it follows that \(M_0(z)\) is
invertible for every \(z\) with \(\operatorname{Re}z\geq0\).

We first derive two estimates for the Laplace transform. For
\(\operatorname{Re}\xi\geq0\), we have
\[
    |\widehat h(\xi)|
    \leq
    \int_0^\infty
    |e^{-\xi s}|h(s)\,ds
    \leq
    \int_0^\infty h(s)\,ds
    =
    1.
\]
Moreover,
\[
    1-e^{-\xi s}
    =
    \xi s\int_0^1e^{-\theta\xi s}\,d\theta,
\]
and hence
\[
    |1-e^{-\xi s}|
    \leq
    |\xi|s
    \qquad
    \text{for }\operatorname{Re}\xi\geq0.
\]
As $h$ has unit mean, we obtain
\[
    |1-\widehat h(\xi)|
    =
    \left|
       \int_0^\infty
       \left(
          1-e^{-\xi s}
       \right)h(s)\,ds
    \right|    
    \leq
    |\xi|
    \int_0^\infty s h(s)\,ds
    =
    |\xi|.
\]
Together with \(|\widehat h(\xi)|\leq1\), this gives
\begin{equation}
    |1-\widehat h(\xi)|
    \leq
    \min\{|\xi|,2\}.
\label{eq:kernel-transform-bound}
\end{equation}
Let us assume, by contradiction, that for some \(q<q_0\) there exists a
characteristic root \(z\) with
\(
    \operatorname{Re}z\geq0.
\)
We first notice that \(z\neq0\), as
\(
    M(0,q)
    =
    D-K
\)
is nonsingular. 
Defining
\(
    \delta(z,q)
    :=
    1-\widehat h(qz)
\), we have 
\[
    M(z,q) =
    zI+D-K
    +
    \left(
       1-\widehat h(qz)
    \right)K
    =
    M_0(z)
    \left[
       I+\delta(z,q)G(z)
    \right].
\]
Since \(M_0(z)\) is invertible and \(\det M(z,q)=0\), it follows
that \(I+\delta(z,q)G(z)\)
is singular. Consequently, there exists
\(\mu\in\sigma(G(z))\) such that
\[
    1+\delta(z,q)\mu=0.
\]
Therefore,
\begin{equation}
    |\delta(z,q)|
    =
    \frac{1}{|\mu|}
    \geq
    \frac{1}{\rho_{\mathrm{sp}}(G(z))}
    =
    r_K(z).
\label{eq:structured-radius-inequality}
\end{equation}
On the other hand, \eqref{eq:kernel-transform-bound} yields
\begin{equation}
    |\delta(z,q)|
    \leq
    \min\{q|z|,2\}.
\label{eq:delta-bound}
\end{equation}
Equations \eqref{eq:structured-radius-inequality} and
\eqref{eq:delta-bound} imply
\(
    r_K(z)\leq2
\)
and
\(
    r_K(z)\leq q|z|.
\)
Hence
\[
    q
    \geq
    \frac{r_K(z)}{|z|}
    \geq
    q_0,
\]
which contradicts \(q<q_0\). Hence, no characteristic root can lie in the
closed right half-plane.

It remains to verify that \(q_0\in(0,\infty]\). Since \(K-D\) is Hurwitz,
\(G(z)\) is continuous on the closed right half-plane. Moreover,
\[
    G(z)
    =
    \frac{1}{z}
    \left(
       I-\frac{K-D}{z}
    \right)^{-1}K,
\]
which leads to
\[
    \rho_{\mathrm{sp}}(G(z))
    =
    O\!\left(\frac{1}{|z|}\right)
    \qquad
    \text{as }|z|\to\infty.
\]
Therefore \(r_K(z)\to+\infty\) as \(|z|\to\infty\) and the constraint \(r_K(z)\le2\) confines \(z\) to a bounded set. Since \(G\) is bounded near \(z=0\), \(r_K(z)\) is bounded below by a positive constant there, and therefore \(r_K(z)/|z|\to+\infty\) as \(z\to0\). Consequently, every minimizing sequence with finite objective can be restricted to a compact annulus \(\varepsilon\le|z|\le R\), on which \(r_K(z)/|z|\) has a strictly positive lower bound.

Hence, the infimum in \eqref{eq.q0} is either taken over an
empty set, in which case \(q_0=+\infty\), or over a compact set bounded
away from zero, in which case it is strictly positive.
\end{proof}

\begin{thm}[Instability for large relative delays]
\label{thm:large-delay}
For the equilibrium $X^*$, define the limiting characteristic function
\[
\Delta_\infty(u)
:=
\det\!\left(D-\widehat h(u)K\right).
\]
and suppose that it has a simple zero \(u_+\) in the open right
half-plane:
\[
\Delta_\infty(u_+)=0,
\qquad
\Delta_\infty'(u_+)\neq 0,
\qquad
\operatorname{Re}u_+>0.
\]
Then there exist \(Q>0\) and a characteristic root \(z(q)\) of \eqref{eq.scaled-characteristic-equation}, defined for
\(q\geq Q\), such that
\[
z(q)
=
\frac{u_+}{q}
+
O\!\left(\frac{1}{q^2}\right)
\qquad\text{as }q\to\infty.
\]
In particular, \(\operatorname{Re}z(q)>0\) for all sufficiently large \(q\),
and hence the equilibrium is unstable for all \(q\geq Q\).
\end{thm}

\begin{proof}
For $q>0$, denoting $\varepsilon=q^{-1}$ and
$u=qz$, the characteristic equation \eqref{eq.scaled-characteristic-equation}
is equivalent to
\[
F(u,\varepsilon)
:=
\det\!\left(
\varepsilon u I+D-\widehat h(u)K
\right)
=0.
\]
Since \(\operatorname{Re}u_+>0\), the Laplace transform \(\widehat h\) is
holomorphic in a neighborhood of \(u_+\). Moreover,
\[
F(u,0)=\Delta_\infty(u).
\]
By the implicit function theorem, there exists a branch
\(u(\varepsilon)\), defined for sufficiently small \(\varepsilon\geq0\),
such that
\[
u(0)=u_+,
\qquad
F(u(\varepsilon),\varepsilon)=0,
\qquad
u(\varepsilon)=u_++O(\varepsilon).
\]
Returning to \(z=\varepsilon u\) gives
\[
z(q)
=
\frac{1}{q}u\!\left(\frac{1}{q}\right)
=
\frac{u_+}{q}
+
O\!\left(\frac{1}{q^2}\right).
\]
Since \(\operatorname{Re}u_+>0\), it follows that
\(\operatorname{Re}z(q)>0\) for all sufficiently large \(q\). Therefore,
the equilibrium is eventually unstable.
\end{proof}

\begin{rem}
\label{rem:large-delay-kernels}
Let
$A=D^{-1}K.
$
For the weak Gamma kernel,
\[
\Delta_\infty(u)
=
\det(D)\,
\det\!\left(I-\frac{1}{1+u}A\right).
\]
Consequently, if \(A\) has a simple eigenvalue \(\mu\) satisfying
\(
\operatorname{Re}\mu>1,
\)
then \(u_+=\mu-1\) satisfies the hypotheses of
Theorem~\ref{thm:large-delay}.

For the discrete kernel,
\[
\Delta_\infty(u)
=
\det(D)\,
\det\!\left(I-e^{-u}A\right).
\]
Hence, if \(A\) has a simple eigenvalue \(\mu\) with
\(
|\mu|>1,
\)
then the corresponding roots
\[
u_k
=
\log|\mu|
+
i\bigl(\arg\mu+2k\pi\bigr),
\qquad k\in\mathbb Z,
\]
lie in the open right half-plane, and
Theorem~\ref{thm:large-delay} applies.
\end{rem}

\color{black}
\clearpage
\bibliographystyle{plain}
\bibliography{references}

\end{document}